\documentclass[10pt]{article}
\usepackage[margin=1in]{geometry}
\usepackage{mathptmx}
\usepackage{fullpage,url}
\usepackage[utf8]{inputenc}
\usepackage{titling}
\usepackage{authblk}
\usepackage{natbib}
\usepackage{setspace}
\usepackage{float}

\def\notes{1}
\def\tpdp{0}

\usepackage{amssymb, amsmath,amsthm}
\usepackage{bbm}
\usepackage{xcolor}
\usepackage{graphicx, float}
\usepackage[export]{adjustbox}
\usepackage{caption}
\usepackage{subcaption}
\usepackage{tikz}
\usepackage[normalem]{ulem}
\usepackage{hyperref}
\usepackage[ruled,vlined]{algorithm2e}

\newtheorem{definition}{Definition}[section]
\newtheorem{theorem}{Theorem}[section]
\newtheorem{lemma}{Lemma}[section]

\SetKwInput{KwInput}{Input}
\SetKwInput{KwPrivacyparams}{Privacy params}  
\SetKwInput{KwHyperparams}{Hyperparams} 
\SetKwInput{KwOutput}{Output}

\newcommand{\eps}{\varepsilon}

\newcommand{\alphaforcdf}{\alpha_{\textrm{CDF}}}
\newcommand{\rhoforbs}{\rho_{\textrm{BinSearch}}}
\newcommand{\rhoforcdf}{\rho_{\textrm{CDF}}}
\newcommand{\rhoperstep}{\rho_{\textrm{step}}}
\newcommand{\rhoinit}{\rho_{\textrm{init}}}

\newcommand{\rhousedthisstep}{\rho_{\textrm{t}}}
\newcommand{\rhoused}{\rho_{\textrm{used}}}

\newcommand{\betaperstep}{\beta_{\textrm{step}}}
\newcommand{\betainit}{\beta_{\textrm{init}}}

\newcommand{\betausedthisstep}{\beta_{\textrm{t}}}

\newcommand{\calX}{\mathcal{X}}
\newcommand{\calY}{\mathcal{Y}}
\newcommand{\calN}{\mathcal{N}}

\newcommand{\bx}{\mathbf{x}}

\newcommand{\median}{\texttt{med}} 
\newcommand{\range}{\mathcal{R}}
\newcommand{\datalb}{\lowerrange}
\newcommand{\dataub}{\upperrange}

\newcommand{\rangesmall}{\range_{\textrm{small}}}
\newcommand{\rangelarge}{\range_{\textrm{large}}}
\newcommand{\reals}{\mathbb{R}}
\newcommand{\posreals}{\reals_{+}}
\newcommand{\nonnegreals}{\reals_{\geq 0}}
\newcommand{\naturals}{\mathbb{N}}
\newcommand{\indicator}{\textbf{1}}

\newcommand{\bparen}[1]{\big( {#1} \big)}

\newcommand{\set}[1]{\left\{ {#1} \right\}}

\newcommand{\gooddist}{\Delta_{\mathcal{C}}(\mathbb{R})}
\newcommand{\intervals}{I_{\reals}} 
\newcommand{\pops}{\mathcal{P}} 
\newcommand{\NPL}[1]{\texttt{ci}_L^{#1}} 
\newcommand{\NPU}[1]{\texttt{ci}_U^{#1}} 

\newcommand{\privNPL}[1]{\widetilde{\texttt{ci}_L^{#1}}} 
\newcommand{\privNPU}[1]{\widetilde{\texttt{ci}_U^{#1}}} 
\newcommand{\PNPL}[1]{\texttt{N}_L^{#1}} 
\newcommand{\PNPU}[1]{\texttt{N}_U^{#1}} 
\newcommand{\PPL}[1]{\texttt{N}^{#1}_{\rho, L}} 
\newcommand{\PPU}[1]{\texttt{N}^{#1}_{\rho, U}} 
\newcommand{\PPLeps}[1]{\texttt{N}^{#1}_{\eps, L}} 
\newcommand{\PPUeps}[1]{\texttt{N}^{#1}_{\eps, U}} 

\newcommand{\upperrange}{r_u}
\newcommand{\lowerrange}{r_{\ell}}
\newcommand{\granularity}{\theta}
\newcommand{\targetquantile}{q_{\textrm{target}}}
\newcommand{\discreterange}{\mathcal{R}_{\rm discrete}}
\newcommand{\relwidth}{\textrm{rel-width}}
\newcommand{\cov}{\textrm{cov}}

\newcommand{\rank}[2]{\texttt{rank}_{#1}(#2)}
\newcommand{\depth}{m}
\newcommand{\dptree}{T}
\newcommand{\dptreespace}[1]{\text{BinTree}(\naturals, #1)}
\newcommand{\betasplit}{\gamma}
\newcommand{\alphasplit}{r_1}
\newcommand{\rhosplit}{\gamma}

\newcommand{\ninput}{n \in \naturals}
\newcommand{\datainput}{d = (d_1, \ldots, d_n) \in \reals^n}
\newcommand{\rangeinput}{\range = [\lowerrange, \upperrange] \subset \reals}
\newcommand{\alphainput}{\alpha \in (0,1)}
\newcommand{\rhoinput}{\rho \in \posreals}
\newcommand{\epsinput}{\eps \in \posreals}
\newcommand{\graninput}{\theta \in \posreals}
\newcommand{\rankinput}{k \in [n]}
\newcommand{\betasplitinput}{\betasplit \in (0,1)}
\newcommand{\treeinput}{\dptree \in \dptreespace{\depth}}
\newcommand{\depthinput}{\depth \in \naturals}
\newcommand{\discreterangeinput}{\discreterange \in \reals^{2^\depth}}

\newcommand{\nonprivconf}{\beta_1}
\newcommand{\privconf}{\beta_2}

\newcommand{\GS}{\text{GS}}
\newcommand{\CDFfunction}{C}
\newcommand{\DPCDF}{\tilde{C}}
\newcommand{\CDFbin}{C_{\Bin}}
\newcommand{\PDFbin}{C'_{\Bin}}
\newcommand{\Bin}{\text{\rm Bin}}

\newcommand{\Naive}{\text{$\mathtt{Union}$}}
\newcommand{\CDF}{\text{$\mathtt{CDFPostProcess}$}}

\newcommand{\noisyBS}{\text{$\mathtt{NoisyBinSearch}$}}

\newcommand{\BSCDF}{\text{$\mathtt{BinSearch+CDF}$}}
\newcommand{\EM}{\text{$\mathtt{ExpMech}$}}
\newcommand{\EMPointEstimator}{\texttt{ExpMechPoint}}
\newcommand{\EMNaive}{\texttt{ExpMechUnion}}
\newcommand{\ComputeEMTargets}{\texttt{ComputeExpMechTargets}}

\ifnum\notes=1

\newcommand{\removeforgood}[1]{}
\newcommand{\removeforsubmission}[1]{\color{black}{#1}\color{black}}
\newcommand{\justforsubmission}[1]{#1}

\else 

\newcommand{\removeforgood}[1]{}
\newcommand{\removeforsubmission}[1]{{#1}}
\newcommand{\justforsubmission}[1]{}

\fi

\begin{document}
\begin{singlespace}
\title{Non-parametric Differentially Private Confidence Intervals for the Median}

       \author[1, 2]{J\"{o}rg Drechsler}
       \author[3]{Ira Globus-Harris}
       \author[4]{Audra McMillan\footnote{Part of this work was completed while the author was at Boston University and Northeastern University.}}
       \author[5]{Jayshree Sarathy}
       \author[6]{Adam Smith\footnote{Authors in alphabetical order.  }}
       \affil[1]{Institute for Employment Research, Germany}
       \affil[2]{The Joint Program in Survey Methodology, University of Maryland, USA}
       \affil[3]{University of Pennsylvania, USA}
       \affil[4]{Apple, USA}
       \affil[5]{Harvard John A. Paulson School of Engineering and Applied Sciences, USA}
       \affil[6]{Department of Computer Science, Boston University, USA}

\date{\today}
\maketitle

\ifnum\tpdp=0
\pagenumbering{arabic}
\fi

\begin{abstract}
 Differential privacy is a restriction on data processing algorithms that provides strong confidentiality guarantees for individual records in the data. 
 However, research on proper statistical inference, that is, research on properly quantifying the uncertainty of the (noisy) sample estimate regarding the true value in the population, is currently still limited. This paper proposes and evaluates several strategies to compute valid differentially private confidence intervals for the median. Instead of computing a differentially private point estimate and deriving its uncertainty, we directly estimate the interval bounds and discuss why this approach is superior if ensuring privacy is important. We also illustrate that addressing both sources of uncertainty--the error from sampling and the error from protecting the output--simultaneously should be preferred over simpler approaches that incorporate the uncertainty in a sequential fashion.
 We evaluate the performance of the different algorithms under various parameter settings in extensive simulation studies and demonstrate how the findings could be applied in practical settings using data from the 1940 Decennial Census.
\end{abstract}
\end{singlespace}
\section{Introduction}
Statistical agencies constantly need to find the right balance between the two competing goals of disseminating useful information from their collected data and ensuring  the confidentiality of the units included in the database. Many methods have been developed in the past decades to address this trade-off. However, with the advent of modern computing and the massive amounts of data collected every day, many of the data protection strategies commonly used at statistical agencies are no longer adequate to sufficiently protect the data \citep{abowd2018,GarfinkelAM19}. The problem's difficulty is amplified by the continual appearance of new data sources that facilitate attacks.

One promising strategy to circumvent this dilemma is to rely on formal privacy guarantees such as those provided by differential privacy (DP)~\citep{DMNS06}. These guarantees hold no matter what background knowledge a potential attacker might possess, or how much computational power they have. However,  methodology for differential private statistical inference has mostly been studied from a theoretical perspective under asymptotic regimes. Although many  algorithms have been proposed to ensure formal privacy guarantees for various estimation tasks, evaluations of their relative performance on real data with limited sample sizes and complex distributional properties are still limited, and only a small fraction of that literature has focused on inference and associated measures of uncertainty.
\ifnum\tpdp=0
Section~\ref{sec:related} surveys related work.
\fi

In this paper, we address these issues, focusing on one of the key measures of location: the median. We chose the median for two reasons. On one hand, it is a widely used summary statistic for skewed variables such as income  (see, for example, the U.S. Census Bureau's tables of median incomes for various subgroups of the population \citep{CensusTables20}). On the other hand, medians provide an interesting technical challenge for differentially private computation. The accuracy of differentially private median computations depends on the exact data distribution; as a result, providing sound and narrow confidence intervals appears to require releasing strictly more information about the data than is required for point estimation. 

The discussion of confidence intervals is an important contribution of our paper. None of the previously proposed algorithms for DP median estimation come equipped with a method for additionally releasing DP uncertainty estimates on the point estimator. 
In fact, the level of uncertainty in the point estimate is typically data dependent, and hence measuring it requires additional privacy budget.
Thus, the optimal algorithm for differentially private point estimates can be different from the optimal algorithm for differentially private confidence intervals. 
Instead of deriving the variance of some differentially private point estimate, we suggest estimating DP confidence intervals directly.
We show that our proposed methodology ensures proper confidence interval coverage in a frequentist sense and discuss why this strategy requires less privacy budget than starting from the protected point estimates. 

When designing and analysing differentially private algorithms it is tempting to 
separate the error due to sampling from the error due to privacy and bound the two separately. A main finding in our work is the limitation of this approach.  
We find that one can obtain considerably tighter confidence intervals by analysing the relationship between the two sources of error. Unlike approaches which treat the analysis of the non-private algorithm as a black-box, this involves looking at the different ways that the sampling error can result in the confidence interval failing to capture the median, and considering how the error due to privacy affects each of these modalities. 

We assume simple random sampling throughout the paper. This assumption is often violated in survey practice. However,  understanding the implications of complex sampling designs on the privacy guarantees is an open research problem \citep{drechsler2021} and we are not aware of any DP applications that take complex sampling designs into account. 
We  see our contribution as an important first step towards the goal of better serving the needs of statistical agencies, while acknowledging the limitations of the current findings. We will come back to this point in the conclusions. 

We evaluate several algorithms for computing valid differentially private confidence intervals. We discuss algorithms that satisfy two versions of differential privacy: the strictest version \citep{DMNS06}, now known as \textit{pure differential privacy}, as well as a slight relaxation, \textit{concentrated differential privacy} \citep{BunS16, Dwork:2016}. The focus of our paper is on empirical evaluation, using a mix of simulated and real data. Nevertheless, we found that new methodology and theory was also needed to adapt existing algorithms for confidence interval computation. 
We include an application using data from the U.S. Census 1940 to illustrate how statistical agencies willing to adopt the methodology could decide which algorithm and parameter settings to pick for their data release.

The algorithms we developed are all \textit{sound} in the nonparametric, frequentist sense: when run with nominal coverage $1-\alpha$ the probability that the true population median is contained in the computed confidence interval is at least $1-\alpha$, where the probability is taken over the entire process of sampling from the population and computing the private confidence intervals based on the drawn sample. Since all algorithms rely on non-parametric strategies for computing the confidence intervals, the intervals are valid for every IID distribution on observations.  
We summarize our findings briefly:
\begin{itemize}
    \item In our comparison of several algorithms, the best choice across a range of settings was a variant of the exponential mechanism \citep{McSherryT07}, a generic framework for DP algorithm design that we adapt for confidence interval estimation. This algorithm is tailored to the median, and releases only a single confidence interval. 
    
    \item A different algorithm, based on a differentially private CDF estimate~\citep{Li:2010}, consistently produced confidence intervals that were slightly wider than those of the exponential mechanism. However, the algorithm's output can be used to produce a confidence interval for any quantile of the data set or even a confidence band for the entire CDF. In principle, the approximation to the entire CDF would also allow the incorporation of a parametric model or a Bayesian prior after the fact. Its flexibility makes it a better choice for settings where eventual users will be interested in more than a single median.
    
    \item For settings where only a very loose bound on the range of the data range is known a priori, a hybrid algorithm that uses binary search to narrow the range and then switches to the CDF-based estimator produces narrower intervals than other methods. 
    
    \item The methods we tested exhibited noticeable bias that depends on the underlying distribution and appears to be hard to correct. The bias was low relative to the width of the confidence intervals, and so would not be an issue for one-shot applications. However, it might be a concern when aggregating estimates across many small areas. It is not clear whether bias is necessary for accurate nonparametric DP median approximations.~\footnote{The one unbiased method that we tested (based on the \textit{smooth sensitivity} framework~\citep{NRS07}) produced relatively poor point estimates, and we did not include it in the tests of confidence interval width. 
    \ifnum\tpdp=0
        See Figure~\ref{fig:pointestimates} in Appendix~\ref{online supplement: otheralgs} for analysis of performance of CDP point estimators for the median.
    \else
        These results are presented in the full version.     
    \fi 
    }

\end{itemize}

\ifnum\tpdp=0
    The remainder of the paper is organized as follows: In Section \ref{sec:prelim} we review some of the privacy definitions that are relevant for this paper and discuss confidence interval estimation for the median without privacy considerations. We extend these discussions to differentially private confidence intervals in Section \ref{sec:DP-CIs}. Section \ref{algorithms} contains a high-level review of the algorithms we considered (detailed descriptions of the different algorithms can be found in the Appendix). In Section \ref{sec:simulation} we present the results from extensive simulation studies that evaluate the performance of the algorithms under various parameter settings. Section \ref{sec:application} illustrates how the methodology could be applied in practice by replicating one of the income tables published by the U.S. Census Bureau using publicly available data from the 1940 U.S. Census. The paper concludes with some final remarks. 
\else 
    Our full version, forthcoming on arXiv, presents both the necessary theoretical development and experiments based on a mix of simulated data and publicly available data from the 1904 US Census. 
\fi

\ifnum\tpdp=0 
\section{Preliminaries}
\label{sec:prelim}

\subsection{Differential Privacy} 

The algorithms in this paper 
satisfy a version of differential privacy (DP) called \emph{concentrated differential privacy} (CDP). This notion of privacy lies between the more common notions of \emph{pure differential privacy} and \emph{approximate differential privacy}. 
Since our algorithms often include hyperparameters, we state a definition of DP for algorithms that take as input not only the dataset, but also the desired privacy parameters and any required hyperparameters. Let $\mathcal{X}$ be a data universe 
(e.g., $\reals$ for medians)
and $\mathcal{X}^n$ be the space of datasets of size $n$. Two datasets $d, d' \in \mathcal{X}^n$ are neighboring, denoted $d \sim d'$, if they differ on a single record. 
Let $\mathcal{H}$ be the space of hyperparameters and $\mathcal{Y}$ be an output space. In order to build some intuition, let us first define pure and approximate DP.

\begin{definition}[$(\eps, \delta)$-Differential Privacy \citep{DMNS06, Dwork2006}]\label{def:DP} Given $\eps\ge0$ and $\delta\in[0,1]$,
a randomized mechanism $M: \mathcal{X}^n \times \mathcal{H} \rightarrow \mathcal{Y}$ is $(\eps,\delta)$-\emph{differentially private} if for all datasets $d \sim d' \in \mathcal{X}^n$, $\textrm{hyperparams} \in \mathcal{H}$, and events $E\subseteq\calY$,
\begin{align*} \label{def:dp-with-inputs}
    &\Pr[M(d, \text{hyperparams}) \in E]
    \leq e^\eps \cdot\Pr[M(d', \text{hyperparams}) \in E] + \delta,
\end{align*}
where the probabilities are taken over the random coins of $M$.
\end{definition}

The key intuition for this definition is that the distribution of outputs on input dataset $d$ is almost indistinguishable from the distribution on outputs on input dataset $d'$. Therefore, given the output of a differentially private mechanism, it is impossible to confidently determine whether the input dataset was $d$ or $d'$. If $\delta=0$, then we refer to this as $\eps$-\emph{pure differential privacy}. If $\delta>0$, we refer to $(\eps,\delta)$-\emph{approximate differential privacy}.
For strong privacy guarantees, the privacy-loss parameter is typically taken to be a small constant less than $1$ (note that $e^\eps \approx 1+\eps$ as $\eps \rightarrow 0$). However, in practice, larger values of $\eps$ are occasionally used to satisfy utility constraints while providing some level of non-trivial privacy guarantee. 

Concentrated differential privacy has the same intuition; it bounds the divergence between the distributions $M(d)$ and $M(d')$.

\begin{definition}[$\rho$-Concentrated Differential Privacy \citep{BunS16}] Given $\rho\ge 0$, a randomized mechanism $M: \mathcal{X}^n \times \mathcal{H} \rightarrow \mathcal{Y}$ satisfies $\rho$-\emph{concentrated differential privacy} if for all datasets $d \sim d' \in \mathcal{X}^n$, $\textrm{hyperparams} \in \mathcal{H}$, and $\alpha\in(1,\infty)$,
\begin{align*}
    D_{\alpha}(M(d,\textrm{hyperparams})\|M(d',\textrm{hyperparams}))\le \rho
\end{align*}
where $D_{\alpha}$ is the $\alpha$-R\'enyi divergence and the probabilities are taken over the random coins of $M$.
\end{definition}

In order to give some intuition for concentrated DP, let us elaborate more on its relationship with pure and approximate DP. Given data sets $d\sim d'$, and a randomised mechanism $M$, we can define a random variable, called the \emph{privacy loss random variable}, denoted $Z=\texttt{Priv}(M(d),M(d'))$, as follows. Let $y\sim M(d)$ (i.e. $y$ is the output of the mechanism $M$ on input $d$), then $Z = \ln\left(\frac{\Pr(M(d))=y}{\Pr(M(d'))=y}\right)$. 
Then $M$ is $\eps$-pure differentially private if and only if $\Pr(|Z|>\eps)=0$, and 
$M$ being $(\eps,\delta)$-approximately differentially private  is (almost) captured by the requirement that $\Pr(|Z|>\eps)\le\delta$. Now, $\rho$-concentrated differential privacy essentially translates to the requirement that $Z$ is a subgaussian random variable with mean $\rho$ and variance $2\rho$. From this perspective, it is clear that concentrated differential privacy lies between pure and approximate DP.

\begin{lemma}
 If $M$ is $\rho$-CDP, then $M$ is $(\rho + 2\sqrt{\rho \log(1/\delta)}, \delta)$-DP for any $\delta > 0$. If $M$ is $\eps$-DP then $M$ is $\frac{1}{2}\eps^2$-CDP.
\end{lemma}
 
We will focus in this paper on algorithms that satisfy concentrated differential privacy. While still satisfying a rigorous notion of privacy, this will allow our algorithms to be significantly more accurate than their corresponding purely differentially private counterparts. For most of our algorithms little accuracy is gained from transitioning to approximate differential privacy. Additionally, CDP has the desirable property of being a one-parameter property, which allows for simpler privacy accounting. The lemma below captures the fact that the class of  $\rho$-CDP algorithms is closed under adaptive composition and post-processing.
\begin{lemma}\citep{BunS16}\label{composition}
Let $M: \mathcal{X}^n \times \mathcal{H} \rightarrow \mathcal{Y}$ and $M': \mathcal{X}^n \times \mathcal{H'} \rightarrow \mathcal{Y}'$, where $\mathcal{H'}=\mathcal{Y}\times\mathcal{H''}$. Define $M'':\mathcal{X}^n\times (\mathcal{H'}\times\mathcal{H''})$ by \[M''(d,\textrm{hyperparams}, \textrm{hyperparams}') = M'(d, (M(d, \textrm{hyperparams}), \textrm{hyperparams}')).\] If $M$ is $\rho$-CDP and $M'$ is $\rho'$-CDP then $M''$ is $\rho+\rho'$-CDP.
\end{lemma}

\subsection{Confidence Intervals for the Median}\label{confidenceintervals}

In many statistical applications, we assume that data are drawn i.i.d. from an underlying population distribution, and the statistic of interest is a property of the underlying population. However, one typically only has access to a sample from that population, so the statistic computed on the sample is used as an \emph{estimate} of the true population statistic. We will refer to the median of the underlying population as the \emph{population median} and the median of a given sample as the \emph{sample median}.
Since there is randomness in the sampling process, there is always uncertainty in how well the sample median matches the true population median. 
As this uncertainty can be large, sample statistics should be accompanied by a measure of the uncertainty. Providing a measure of uncertainty is even more important for differentially private statistics since randomness in the algorithm provides an additional source of uncertainty. 

\begin{figure}
    \centering
    \includegraphics[scale=0.42]{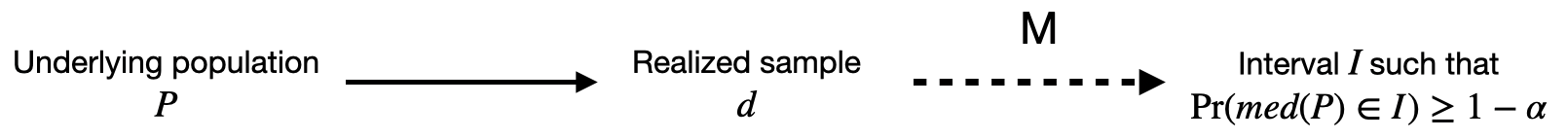}
    \caption{A graphical representation of the process of computing a confidence interval. When privacy is not a concern, no restrictions are placed on the function $M$. When computing a differentially private confidence interval, we require that $M$ is differentially private. The probability is taken over all the randomness in the system, both the randomness due to sampling, and the randomness in $M$.}
    \label{fig:DPCI}
\end{figure}

One method for capturing the uncertainty in an estimate is a  confidence interval. We consider the standard set-up for statistical inference. Let $\pops\subset\Delta(\mathbb{R})$ be the set of possible population distributions over the data domain $\mathbb{R}$. For any $P\in\pops$, a median of $P$ is defined to be any value $m$ such that \[\int_{-\infty}^m P(x) dx\ge 1/2 \;\;\text{and} \;\; \int_m^{\infty} P(x) dx\ge 1/2. \] For every distribution the set of medians is a non-empty, compact set. Since defining a convention here will be convenient, we will refer to the midpoint of the set of medians as \emph{the median}, denoted $\median(P)$.
Let $\intervals$ be the set of intervals in $\reals$. Given $n\ge 0$, let $M:\mathcal{X}^n\to\intervals$ be a randomised mechanism that takes as input a data set of size $n$ and outputs an interval in $\reals$. Given a desired confidence level $1-\alpha$, the goal of $M$ is to, with probability $1-\alpha$, output an interval that contains $\median(P)$.

\begin{definition}\label{defineCI}
For any $\alpha\in[0,1]$ and $n\in\mathbb{N}$, $M:\mathcal{X}^n\to\intervals$ is a $(1-\alpha)$-\emph{confidence interval for the median} for $\pops$ if for all $P\in\pops$, \[\Pr(\median(P)\in M(d))\geq 1-\alpha,\] where the randomness is taken over both the randomness $M$ and the randomness in the sample $d\sim P^n$.
\end{definition}

A graphical representation of the framework for computing a confidence interval is given in Figure~\ref{fig:DPCI}. We will refer to $1-\alpha$ as the \emph{coverage} of the confidence interval.

When one is not concerned with privacy, 
a non-parametric confidence interval for the median can be computed using the order statistics of the sample. 
The rank of the median $\median(P)$ in a data set $d\in P^n$ is distributed as the binomial $\Bin(n,\beta)$, for some $\beta=\Pr_P(x<\median(P))$.
We can exploit this to obtain a confidence interval for the median. 
For a data set $d\in\reals^n$, let $d_{(k)}$ denote the $k$-th smallest value in $d$, referred to as the $k$-th order statistic. The median of $d$ is the midpoint of $d_{(\lfloor n/2 \rfloor)}$ and $d_{(\lceil n/2 \rceil)}$.

\begin{lemma}[Non-private $(1-\alpha)$-confidence interval] \label{nonprivCI} Let $\CDFbin$ be the CDF of the binomial random variable $\Bin(n, 1/2)$ and let \[\PNPL{\alpha} = \max_{m\in\mathbb{N}} \{m\;|\; \CDFbin(m)\le\alpha/2\} \;\;\text{and}\;\; \PNPU{\alpha} = \min_{m\in\mathbb{N}} \{m\;|\;\CDFbin(m)\ge 1-\alpha/2\}. \] For any data set $d\in\reals^n$, let \[\NPL{\alpha}(d)=d_{(\PNPL{\alpha})} \;\;\text{and}\;\; \NPU{\alpha}(d)=d_{(\PNPU{\alpha})}.\] Then $M:\chi^n\to\intervals$ given by $M(d)=[\NPL{\alpha}(d),\NPU{\alpha}(d)]$ is a $(1-\alpha)$-confidence interval for the median for $\Delta(\mathbb{R})$.
\end{lemma}

We will often refer to the interval, $[\NPL{\alpha}(d), \NPU{\alpha}(d)]$, output by the mechanism in Lemma~\ref{nonprivCI} as \emph{the non-private $(1-\alpha)$-confidence interval for the median}, but we note that it is the mechanism $M$ that satisfies Definition~\ref{defineCI}, not the output.  

In this vein, the goal of CDP confidence intervals is not to privately estimate the specific interval $[\NPL{\alpha}(d),\NPU{\alpha}(d)]$, but to output valid confidence intervals. These confidence intervals may, or may not, contain $[\NPL{\alpha}(d),\NPU{\alpha}(d)]$. Referring to Figure~\ref{fig:DPCI}, producing a CDP confidence interval involves the same procedure, with the additional requirement that $M$ is CDP. That is, the realised sample $d$ is only accessed through a CDP mechanism.

\begin{definition}\label{defineprivCI}
$M:\mathcal{X}^n\times \mathcal{H}\to\intervals$ is a $\rho$-CDP, $(1-\alpha)$-confidence interval for the median for $\pops$ if 
\begin{itemize}
    \item $M$ is $\rho$-CDP
    \item For any $\textrm{hyperparams} \in \mathcal{H}$, $M(\cdot, \textrm{hyperparams})$ is an $(1-\alpha)$-confidence interval for the median for $\pops$.
\end{itemize}
\end{definition}

The definition of confidence intervals as stated in Definition~\ref{defineCI} and Definition~\ref{defineprivCI} only requires that the confidence interval is valid for distributions $P\in\mathcal{P}$. \emph{Parametric} estimation is when one defines $\mathcal{P}$ to be only distributions of a particular, often quite simple, form. For example, $\mathcal{P}$ might be the set of all log-normal distributions over $\mathbb{R}$. In \emph{non-parametric} estimation, we assume no knowledge of the underlying population and set $\mathcal{P}=\Delta(\mathbb{R})$, the set of all distributions over $\mathbb{R}$. If one has accurate knowledge of the underlying population, then parametric estimation can result in a tighter confidence interval. However, if there is model mismatch (for example if the underlying population is not exactly log-normal) then parametric estimation can result in invalid confidence intervals. This effect can be amplified by private algorithms which may rely on the modeling assumptions in non-trivial ways. In this paper our goal is to focus on non-parametric confidence intervals, which means our algorithms will always produce valid confidence intervals. As a minor caveat, we restrict ourselves to the set of distributions with continuous probability density functions on $\mathbb{R}$.\footnote{Note this caveat is minor since continuous distributions are dense in $\Delta(\mathbb{R})$. That is every distribution on $\mathbb{R}$ is within negligible distance of a continuous distribution. We discuss a practical method for handling non-continuous distributions in Appendix~\ref{convolutionisgood} } Denote the set of all continuous distributions on $\mathbb{R}$ by $\gooddist$.  

\removeforsubmission{Note that a confidence interval does not directly output a point estimate for the median itself. In the absence of privacy constraints, one can simply additionally release the sample median $\median(d)$. However, under privacy constraints, it is typically desirable to compute as few statistics as possible, in order to allocate the maximum amount of privacy budget to each statistic. As such, rather than allocating some of the privacy budget to providing a point estimate of the median, it is often preferable to allocate the entire budget to estimating the confidence interval, then use the midpoint of that interval as a point estimate of the median. }

\subsection{Related Work}
\label{sec:related}

Computing confidence intervals for the median is one of the most fundamental statistical tasks. However, finding a differentially private estimator for this task that is accurate across a range of datasets and parameter regimes is surprisingly nuanced. There has been a significant amount of prior work on differentially private point estimators for the median \citep{NRS07, BunS19, Asi:2020, Alabi:2020, Tzamos:2020} and other quantiles \citep{gillenwater:2021}. To the best of our knowledge, none of these works addressed DP confidence intervals for the median. However, there has been significant work on DP confidence intervals for other estimation tasks like (Gaussian or sub-Gaussian) mean estimation \citep{Karwa:2018, Gaboardi:2018, Du:2020, biswas2020coinpress}, and linear regression \citep{Barrientos:2017, Evans:2021}. There are also several works on designing more general DP confidence intervals using bootstrapping, or a technique called subsample-and-aggregate~\citep{NRS07}, to account for the combined uncertainty from sampling and noise due to privacy \citep{Barrientos:2017, Ferrando:2020, Brawner:2018, DOrazio:2015, Evans:current}. These algorithms typically require a parametric model on the data or a normality assumption on the quantity being estimated; neither hold in our setting.

The areas of differentially private bayesian inference \citep{Christos:2014, Wang:2015, Foulds:2016, Mikko:2017, Bernstein:2018, GarrettS19, gong:2019} and hypothesis testing \citep{Vu:2009, CKSBG19, Degue:2018, Rogers:2016, Wang:2015hypothesis} study related problems of quantifying uncertainty, but specific goals differ. 
\cite{Wang18}, \cite{Du:2020}, and \cite{biswas2020coinpress}  perform experimental evaluations of DP confidence intervals, however they focus on different estimators (linear regression and mean estimation) and focus on large datasets of at least 1,000, and generally many more, data points.

To the best of our knowledge, our work is unique in focusing on valid non-parametric differentially private confidence intervals for the median. This approach allows us to define algorithms that provide accurate and private confidence intervals without requiring distributional assumptions on the underlying population. 

\section{Designing DP Confidence Intervals} 
\label{sec:DP-CIs}

\subsection{Roadblocks and first attempts}\label{roadblocks}

There are several roadblocks in designing CDP confidence intervals for the median. Firstly, CDP algorithms that estimate the median using data independent output perturbation methods (methods that involve simply adding noise to the non-private estimate) necessarily perform poorly since the median is very sensitive for worst-case data sets. Thus, in order to design algorithms that perform well on ``typical" data sets, the noise addition must be data dependent. \removeforsubmission{This creates difficulties when releasing information regarding the uncertainty in the private estimate since the uncertainty itself might reveal sensitive information. In order to explore these roadblocks in more detail, let us first consider the simpler task of designing a CDP point estimator for the median. A first attempt may be to consider
the global sensitivity \citep{DMNS06}: 

\begin{definition}[Global Sensitivity]
For a query $f:~\calX^n~\rightarrow~\reals$, the \textbf{global sensitivity} is
\[
\GS_f = \max_{d \sim d'}|f(d) - f(d')|.
\]
\end{definition}

For any function $f$, one can create a differentially private mechanism by adding noise proportional to $\GS_f/\sqrt{\rho}$. If one has no bound on the data, then $\GS_{\median}$ is infinite. Even if one knows that all the data lie in a bounded range $[a,b]$, $\GS_{\median}=|b-a|$, and adding noise proportional to $|b-a|/\sqrt{\rho}$ essentially removes the signal for any reasonable value of $\rho$.

However, for the type of datasets that we typically see in practice, changing one data point, or even a few data points, does not result in a major change in the median. For such data sets, one might consider the local sensitivity~\citep{NRS07}, which can be substantially smaller than the global sensitivity.

\begin{definition}[Local Sensitivity~\citep{NRS07}] The \textbf{local sensitivity} of a query $f:~\calX^n~\rightarrow~\reals$ with respect to a dataset $d\in\calX^n$ is 
\[LS_f(d) = \max_{d\sim d'} |f(d) - f(d')|.\]
\end{definition}

Unfortunately, since the local sensitivity itself is data dependent, adding noise proportional to the local sensitivity is not differentially private. \removeforsubmission{Several approaches have been explored in the DP literature for adding noise that is \emph{close} to the local sensitivity, or at least significantly less than the global sensitivity for typical data sets. In~\citep{NRS07}, Nissim et al. define the \emph{smooth sensitivity}, $SS_{f},$ a smooth upper bound on the local sensitivity such that adding noise proportional to $SS_{f}/\sqrt{\rho}$ is differentially private. They showed that for many statistics, including the median, the smooth sensitivity can be much smaller than the global sensitivity on typical data sets. 
    }}
Several \removeforsubmission{other\ } DP mechanisms for median estimation have been proposed that \removeforsubmission{avoid the large $\GS_{\median}$ by calibrating\ } the noise introduced to the specific data set \citep{NRS07, DworkL09}. While these estimators can perform well as point estimators for the median, our goal is not just to provide an estimate of the median, but also to quantify the uncertainty in our estimate. In  algorithms \removeforsubmission{like the smooth sensitivity mechanism }that tailor the noise to the specific data set, the uncertainty itself is data dependent and thus sensitive. The task of differentially privately releasing an estimate of this uncertainty is nontrivial. Even if one could release a DP estimate of the amount of noise added to the non-private median, the uncertainty in the non-private median is still unaccounted for. For this reason, we focus on DP algorithms that attempt to directly estimate the confidence interval. 

\subsection{Accounting for all sources of randomness}\label{randomnessdiscussion}

Accurate and tight coverage analysis is a crucial component of designing good algorithms since overly conservative coverage estimates can result in confidence intervals that are wider than necessary. 
Valid differentially private confidence intervals need to account for two sources of error; sampling error and error due to privacy. Sampling error, also present in the non-private context, captures how well the realised sample $d$ represents the underlying population $P$. The error due to privacy takes into account the additional randomness in $M$ as a result of the privacy guarantee. Our experimental results highlight that it is important to carefully exploit the dependence between the two sources of randomness.  

As a primer, let us first consider the coverage analysis of the non-private algorithm described in Lemma~\ref{nonprivCI}. This coverage analysis relies on the fact that if $P$ is continuous then for all $m\in n$, \[\Pr(\rank{d}{\median(P)}=m)=\Pr(\Bin(n,1/2)=m).\]
There are two ways that the interval $[\NPL{\alpha}(d),\NPU{\alpha}(d)]$ can fail to capture $\median(P)$; $\median(P)<\NPL{\alpha}(d)$ or $\median(P)>\NPU{\alpha}(d)$. Let us focus on the probability of the first type of failure, $\median(P)<\NPL{\alpha}(d)$.
For every $P\in\gooddist$, 
\[\Pr(\median(P)< \NPL{\alpha}(d)) = \Pr(\median(P)< d_{(\PNPL{\alpha})})=\CDFbin(\PNPL{\alpha}-1)\le\alpha/2,\]
where $\CDFbin$ is the CDF of the binomial random variable $\Bin(n, 1/2)$. The probability of failure at the upper end of the confidence interval is analogous.   

\begin{figure}
    \centering
    \includegraphics[scale=0.5]{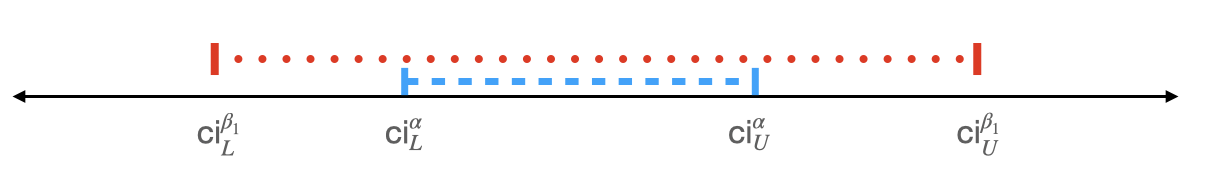}
    \caption{Graphical representation of naive coverage analysis}
    \label{fig:naive}
\end{figure}
Now, let us turn to the coverage analysis of a $\rho$-CDP algorithm $M:\mathcal{X}^n\times \mathcal{H}\to\intervals$. Let $M(d)=[M(d)_L, M(d)_U]$. A naive way to analyse the coverage error of $M$ is to attempt to find $\beta_1$ and $\beta_2$ such that assuming $\beta_1$ denotes the failure probability for the non-private confidence interval, the $M(d)$ contains the non-private interval $[\NPL{\beta_1}, \NPU{\beta_1}]$ with probability $1-\beta_2$. Then $M$ has coverage at least $1-(\beta_1+\beta_2)$. Even if $\beta_1$ and $\beta_2$ are chosen carefully, this analysis can be overly conservative. In particular, it assumes that the only way that $M(d)$ can succeed in containing $\median(P)$ is if both $\median(P)\in[\NPL{\beta_1},\NPU{\beta_1}]$ and $[\NPL{\beta_1},\NPU{\beta_1}]\subset M(d)$. In Figure~\ref{fig:naive}, this corresponds to ensuring that $M(d)$ contains the red dotted interval with high probability. It's clear from this figure that neither of these events are necessary.

A more careful analysis of the relationship between the sampling error and the error due to privacy results in a tighter coverage analysis. As in the non-private setting, there are two ways that $M(d)$ can fail to contain $\median(P)$, and we will focus on analysing the probability that $\median(P)<M(d)_L$
\begin{align}
\Pr(\median(P)< M(d)_L) &= \sum_{m=0}^{n} \Pr(\texttt{rank}_d(\median(P))=m) \cdot \Pr(\median(P)< M(d)_L\;|\; \texttt{rank}_d(\median(P))=m)\nonumber\\
&= \sum_{m=0}^{n} \Pr(\Bin(n,1/2)=m) \cdot \Pr(\median(P)< M(d)_L\;|\; \texttt{rank}_d(\median(P))=m)\label{betteranalysis}
\end{align}
Now, we have reduced the problem to analysing the failure probability conditioned on the empirical rank of $\median(P)$ in the data set $d$. This is a helpful reduction since, as we will see in the following section, most of our algorithms will come with accuracy guarantees on the rank of $M(d)_L$. Accuracy guarantees of this form can then be exploited, via Equation~\eqref{betteranalysis}, to obtain a coverage analysis of $M$.

Our experiments show a stark difference between the performance of algorithms designed using the naive analysis, and those using the tighter, more careful analysis. In Figure~\ref{fig:naive-v-new} we directly compare the confidence intervals that arise from the different analyses. This highlights the importance of understanding the relationship between the two sources of error.

\section{Algorithms}\label{algorithms}

In this section we will introduce the four algorithms for releasing CDP confidence intervals for the median that will be the focus of this paper; $\EM, \CDF, \noisyBS$, and $\BSCDF$. The first algorithm, which we call $\EM$, is based on the exponential mechanism. This mechanism is efficient, satisfies the stronger privacy guarantee of pure differential privacy and outputs the tightest, or close to the tightest confidence intervals in a majority of parameter regimes we studied. 
The remaining three algorithms partially address a common frustration with differentially private data analysis; that exploratory data analysis to visualise the data set and verify findings typically requires additional privacy budget. For many tasks, this means allocating privacy budget away from the primary task resulting in a noisier algorithm. A key feature of the three algorithms $\CDF$, $\noisyBS$ and $\BSCDF$ is that they release additional information about the data set without consuming additional budget. In particular, $\CDF$ releases a full CDP estimate to the empirical CDF. It is then notable, and perhaps surprising, that in many settings these algorithms perform almost as well as $\EM$, which releases no side information.

In this section we will give a high level description of each algorithm. Further information for all algorithms, including pseudo-code and proofs of the privacy and validity guarantees, can be found in the online supplement accompanying this paper. Real code is available in our GitHub repository.\footnote{\url{https://github.com/anonymous-conf-medians/dp-medians}} 
We note that we also experimented with several other algorithms that are not discussed in this section. Brief descriptions of these additional algorithms can be found in the online supplement~\ref{online supplement: otheralgs}, but we do not focus on them here since they are outperformed by other algorithms in every parameter regime we studied.

\subsection{Confidence intervals based on exponential mechanism, $\EM$}\label{EM}

\begin{figure}[t]
    \centering
    \includegraphics[width=\textwidth]{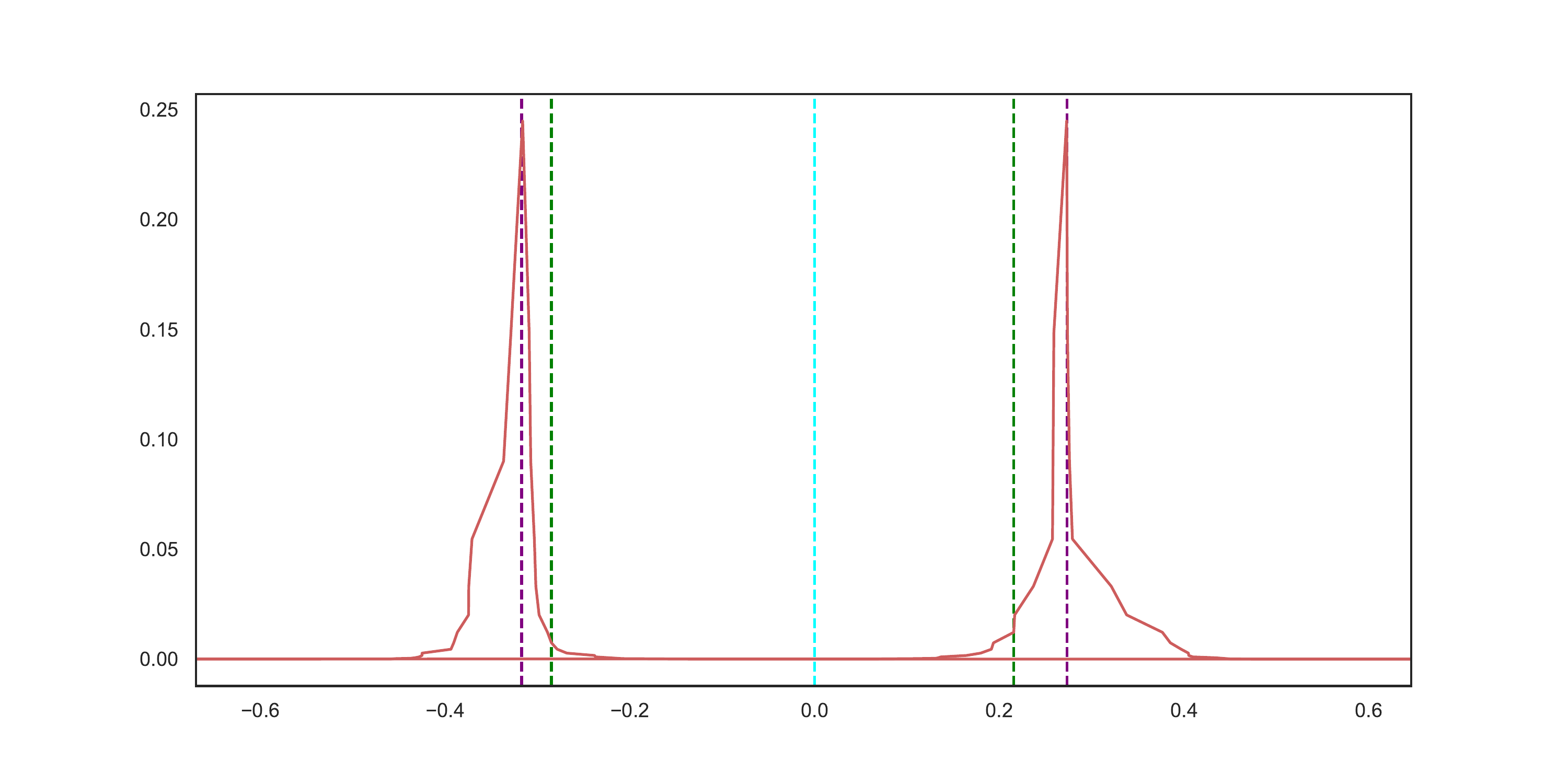}
    \caption{Graphical representation of the distribution of a $1.0$-CDP $\EM$ confidence interval on a single dataset $d$, whose $500$ datapoints are sampled i.i.d. from $\calN(0, 4)$. The range $\range = [-5,5]$, granularity $\theta = 0.05$, and $\alpha=0.05$. 
    The cyan line indicates the population median,
    the green interval represents $\NPL{\alpha}$ and $\NPU{\alpha}$, and the purple interval represents  $d_{(k_L)}$ and $d_{(k_U)}$, where $k_L, k_U$ are chosen according to Equation~\eqref{eq:em-better-analysis}. The red curves illustrate the theoretical distributions of the outputs of $\EM(d)$.
    }
    \label{algoexplainEM}
\end{figure}
     
Our first private mechanism is an instantiation of the exponential mechanism~\citep{McSherryT07}, a differentially private algorithm designed for general optimization problems. The exponential mechanism has been used in prior work to give DP point estimates for the median \citep{DworkL09,Thakurta:2013,Johnson:2013,Alabi:2020,Asi:2020}. Our extension to providing confidence intervals for the median, while using similar ideas to prior work, requires a careful coverage analysis that is new to this work. \removeforgood{This algorithm provides the tightest confidence interval of the CDP algorithm we studied in most of the parameter regimes we studied.
 This algorithm has the advantage that it is the only algorithm we consider that satisfies the stronger privacy guarantee of pure differential privacy.} 

The exponential mechanism is defined with respect to a utility function $u$, which maps (data set, output) pairs to real values. For a data set $d$, the mechanism aims to output a value $r$ that maximizes $u(d,r)$ by sampling a value $r$ with probability proportional to $e^{\eps u(d,r)/\Delta u}$ where $\Delta_u=\max_r\max_{d,d' {\rm neighbours}}|u(d,r)-u(d',r)|$. Recall that the rank of a value $r\in\reals$
in a data set $d$, denoted $\rank{d}{r}$, is the number of data points in $d$ that are less than or equal to $r$.
For any $k\in[n]$, one way to instantiate the exponential mechanism to compute the $k$-th order statistic is by using 
the following utility function. Let 
\[
u_k(d, r) = - |\rank{d}{r}-k|.
\]
In practice we will use a slight variant of the exponential mechanism described in online supplement~\ref{sec:exp-mech-details},
which uses a granularity parameter $\theta$. 
This version of the exponential mechanism has the guarantee that with high probability it will output a value within $\theta$ of some $r$, such that $\rank{d}{r}$ is close to $k$. 
Let us denote the exponential mechanism with privacy parameter $\eps$ for the $k$-th order statistic by $A_k^{\eps}$. 

While it is tempting to use the exponential mechanism to release DP estimates of $\NPL{\alpha}(d)=d_{(\PNPL{\alpha})}$ and $\NPU{\alpha}(d)=d_{(\PNPU{\alpha})}$, this will not create a DP $\alpha$-confidence interval. This is because the private algorithm may under or over estimate these quantities resulting in an invalid confidence interval. Instead we need to choose $k_L$ and $k_U$ carefully so that $M_{k_L,k_U}(d)=[A_{k_L}^{\eps/2}(d) - \theta, A_{k_U}^{\eps/2}(d) + \theta]$ is a valid confidence interval. In order to analyse the coverage, we return to Equation~\eqref{betteranalysis}, \begin{align}
\Pr(\median(P)< A_{k_L}^{\eps/2}(d) - \theta) 
&= \sum_{m=0}^{n} \Pr(\Bin(n,1/2)=m) \cdot \Pr(\median(P)< A_{k_L}^{\eps/2}(d) - \theta\;|\; \texttt{rank}_d(\median(P))=m) \nonumber \\
&\le \CDFbin(k_L-1)+\sum_{m=k_L}^{n} \Pr(\Bin(n,1/2)=m) \cdot \Pr(|\texttt{rank}_d(A_{k_L}^{\eps/2}(d) - \theta)-k_L|\ge m-k_L) \label{eq:em-better-analysis}
\end{align}
\removeforsubmission{In this computation, we use the inequality $\Pr(\median(P)\le A_{k_L}^{\eps/2}(d) - \theta\;|\; \texttt{rank}_d(\median(P))=m)\le 1$ for all $m < k_L$. When $m < k_L$, $\median(P) < d_{(k_L)}$ so $\median(P) < A_{k_L}^{\eps/2}(d) - \theta$ occurs with high probability if $A_{k_L}^{\eps/2}$ outputs something close to $d_{(k_L)}$ (which is the goal of this algorithm) with high probability. Some additional performance could be obtained by tighter analysis of these terms, but we expect the improvement to be small. }
Now the first term is only due to sampling error and easily bounded. The second term depends on both the sampling error and the error rate of the private algorithm $A_{k_L}^{\eps/2}$. In Appendix~\ref{online supplement:EM}, we show how to obtain a distribution independent upper bound on the error rate $\Pr(|\rank{d}{A_{k_L}^{\eps/2}(d)-\theta}-k_L|\ge m)$. This allows us to give an upper bound for $\Pr(\median(P)\le A_{k_L}^{\eps/2}(d)-\theta)$ 
which holds for any distribution $P\in\gooddist$. Given this coverage analysis, we can search for $k_L$ and $k_U$ that maximise $\Pr(\median(P)~\notin~M_{k_L,k_U}(d))$ subject to the constraint that $\Pr(\median(P)\notin M_{k_L,k_U}(d))\le\alpha$.

We will refer to this algorithm as $\EM$. Pseudo-code and details of the analysis can be found in online supplement~\ref{online supplement:EM}. A graph representing the distribution of $\EM$ on a single data set can be found in Figure~\ref{algoexplainEM}.

\subsection{Confidence intervals based on CDF estimator, $\CDF$}

\begin{figure}[t]
    \centering
    \includegraphics[width=0.8\textwidth]{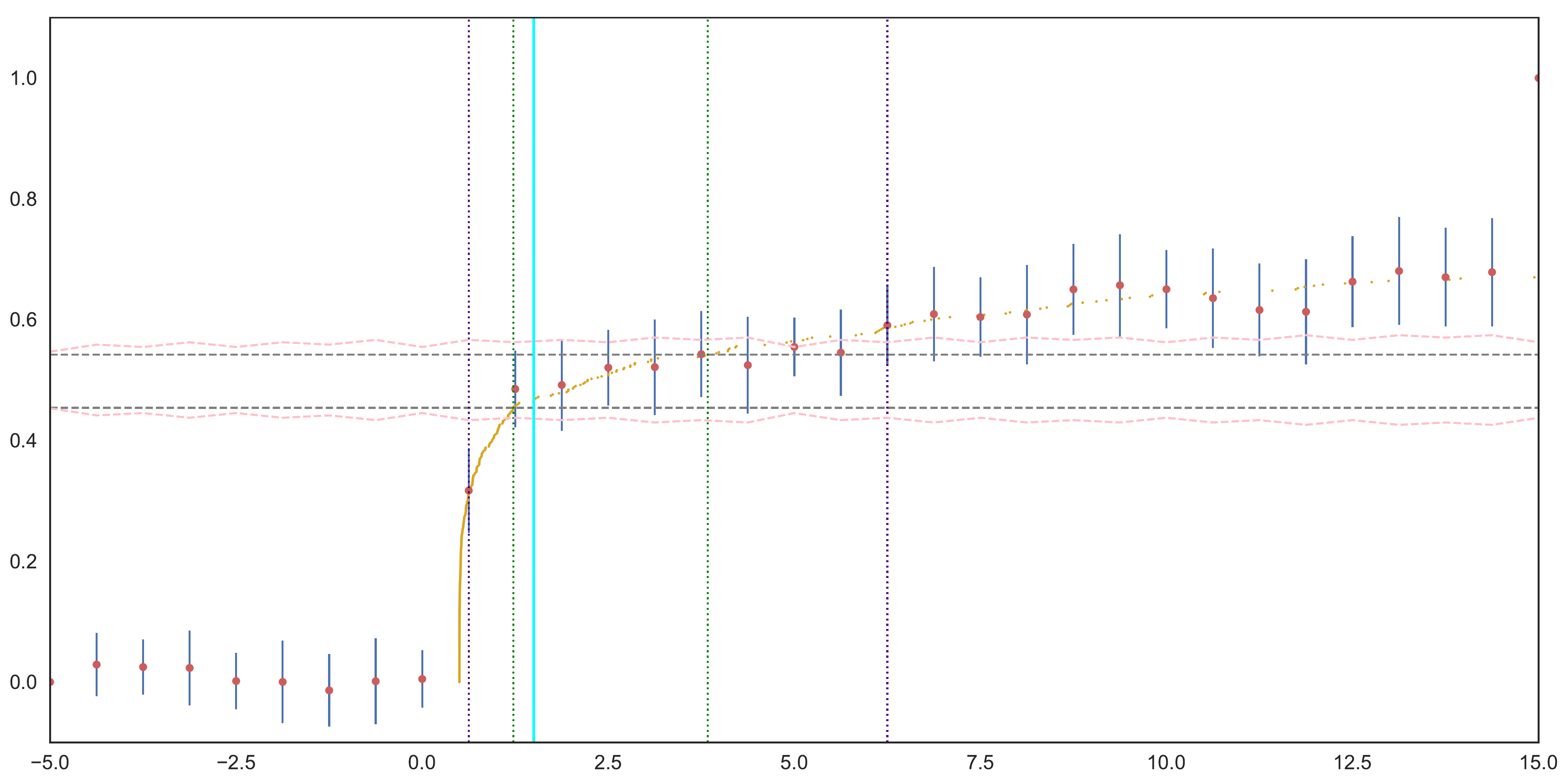}
    \caption{Graphical representation of a single run of $0.1$-CDP $\CDF$ (with hyperparameters $\range=[-5,15], \theta=0.5$) on a single dataset.
    The $500$ datapoints are drawn i.i.d. from Lognormal$(\ln(1.5), 5.0)$. The cyan line is the population median, the gray dashed horizontal lines represent $\PNPL{0.05}/n, \PNPU{0.05}/n$,
    and the green dashed vertical lines represent $\NPL{0.05}(d)$ and $\NPU{0.05}$.
    For each $x\in\discreterange$, the red dot is at $\DPCDF(x)$, the yellow dot is at $\hat{\CDFfunction}(x)$, and the pink dot is at the upper and lower thresholds $a_x^L$ and $a_x^U$. 
    For each $x$ the blue vertical line represents a point wise $99\%$ confidence interval on $\hat{\CDFfunction}(x)$ based on the measurement $\DPCDF(x)$. 
    The purple dashed vertical lines represent the DP interval 
    $\CDF(d)$.}
    \label{algoexplainCDF}
\end{figure}

Our second CDP confidence interval $\CDF$ is obtained from post-processing the output of a CDP cumulative distribution function (CDF) estimator. Unlike $\EM$, which only released the confidence interval, $\CDF$ can additionally release a CDP estimate of the CDF without consuming additional privacy budget. 
There has been considerable work in the DP literature on DP cumulative distribution function (CDF) estimators, both for the parametric and non-parametric models \citep{Diakonikolas:2015, Brunel:2020}.  
We will focus on a particular CDF estimator based on the tree-based mechanism introduced in \citet{Li:2010,Dwork:2010}, and \citet{ Chan:2011}. This mechanism was further refined in \citet{Honaker:2015}, whose algorithm we base our mechanism on. 

Given a range for the data $[\lowerrange, \upperrange]$, and a discretization of this range $\discreterange = [\lowerrange, \lowerrange+\granularity, \lowerrange+2\granularity, \cdots, \upperrange]$, the algorithm in~\citet{Honaker:2015} outputs a DP estimate, $\DPCDF$, of the empirical CDF, $\hat{\CDFfunction}$, of a data set restricted to $\discreterange$. The relevant feature of this estimator for our analysis is that we understand the marginal distribution of $\hat{C}(x)$ for each $x\in\discreterange$.  Specifically, for all $x\in\discreterange$, there exists $\sigma_x>0$ such that the value we release is equal to the empirical CDF with normally distributed noise, that is:
\(\DPCDF(x)=\hat{\CDFfunction}(x)+\mathcal{N}(0,\sigma_x^2) \, . \)
The noise values are not independent across different values of $x$, nor are they perfectly correlated. The estimates $\DPCDF(x)$ don't even increase monotonically with $x$ (see Figure~\ref{algoexplainCDF}). Thus, we seek a procedure that uses only our knowledge of the marginal distribution. The main observation is that we can rewrite the distribution of $\DPCDF(x)$ in terms of the true distribution CDF value, $\CDFfunction(x)$: 
\[\DPCDF(x)=\tfrac{1}{n} \cdot \Bin(n,\CDFfunction(x))+\mathcal{N}(0,\sigma_x^2) \, . \]

We will use this, for each value $x$ in $\discreterange$, to test the hypothesis that $x$ is less than the median. Observe that $x< \median(P)$ if and only if $\CDFfunction(x)< \frac 1 2$. Thus, if $x< \median(P)$, then $\DPCDF(x)$ is stochastically dominated by the distribution $\tfrac{1}{n} \cdot\Bin(\bparen{n,\frac 1 2}+ \mathcal{N}(0,\sigma_x^2)$. We have 
\begin{align*}
Pr(\DPCDF(x)>a)
&\le \sum_{m=0}^{n} \Pr\bparen{\Bin(n,1/2)=m}\Pr\bparen{\tfrac{m}{n}+\mathcal{N}(0,\sigma_x^2)>a} \, ,
\end{align*}
which we can easily numerically evaluate.
Thus, for each $x$, there exists a value $a_x$ (that just depends on $\sigma_x$) such that if $x<\median(P)$ then $\Pr(\DPCDF(x)>a_x)\le\alpha/2$, where $1-\alpha$ is our desired coverage. Equivalently, if 
$\DPCDF(x)>a_x$ then we can be confident that $x \geq \median(P)$. The upper end of our confidence interval is then defined as the grid point just to the right of the largest value $x$ for which the test accepts, that is,
\begin{align*}
     \privNPU{\alpha}(d) &= \theta  + \max\set{x\in \discreterange\ :\ \text{test for $x<\median(P)$ accepts}  }  \\
     &= \theta + \max\set{x\in \discreterange\ : \ \DPCDF(x) \leq a_x} \\
     & = \min\set{x\in \discreterange\ : \ (\forall x' \geq x)\   \DPCDF(x') > a_{x'}}\, .
\end{align*}

To see why this leads to a valid confidence interval, let $x^*$ be the largest value in $\discreterange$ that is less than $\median(P)$. With probability at least $1-\frac{\alpha}{2}$, the test at $x^*$ accepts. In that case, $\privNPU{\alpha}(d)$ will be a grid point greater than $x^*$, and thus we will have $\privNPU{\alpha}(d) \geq \median(P)$. 

The process and reasoning for the left end of the interval are symmetric.
Pseudo-code and a proof that this process gives valid confidence intervals can be found in Appendix~\ref{online supplement:CDF}. We will refer to this algorithm as $\CDF$.
A graph representing a single run of this algorithm can be found in Figure~\ref{algoexplainCDF}. In particular, note that the horizontal pink dotted lines represent the values of $a_x$, which are closer to $\PNPL{\alpha}$ and $\PNPU{\alpha}$ when $\sigma_x$ is small.~\footnote{Of potential independent interest is our efficient implementation for computing the pointwise error rates $\sigma_x$.
While Honaker discusses computing this error for a single $x$, we developed and implemented an efficient algorithm for computing this error for all $x$. This has potential for impact beyond confidence intervals for the median.}

\subsection{Confidence intervals based on noisy binary search, \noisyBS}

One draw-back of the CDF-based estimator $\CDF$ is that it spends its privacy budget roughly evenly across the entire range $\range$. This can result in substantially reduced performance if the data is concentrated in a small subset of the range. We can see this effect in Figure~\ref{algoexplainCDF} where the values of the $\DPCDF$ below $x=0$ are very unlikely to impact the confidence interval. 
$\noisyBS$ attempts to locate the data within the region $\range$ using a small amount of the privacy budget then delves more deeply into the region actually containing the data. This defines a CDP confidence interval in its own right but we will see the real power of it in the design of the next algorithm. 

\justforsubmission{$\noisyBS$ uses noisy queries $\hat{\CDFfunction}(x)+\mathcal{N}(0,\sigma^2)$ to the empirical CDF to search for the relevant quantiles using binary search. Unlike $\CDF$ that obtains a CDP estimate to $\hat{\CDFfunction}$ over the entire $\discreterange$, $\noisyBS$ minimises the number of values we need to evaluate $\hat{\CDFfunction}$ on. Full details and pseudo-code are given in online supplement~\ref{online supplement:BS}.}
\removeforsubmission{As in the design of $\EM$, this algorithm starts with two target quantiles $k_L$ and $k_U$ which are defined to ensure that the resulting confidence interval has the desired coverage. We will discuss this choice in online supplement~\ref{online supplement:BS}. Let $\hat{\CDFfunction}$ be the empirical CDF. In the non-private setting, we can find an approximation to $d_{(k_L)}$ by iteratively asking queries of the form ``is $\hat{\CDFfunction}(x)\le k_L$?" and adjusting our search accordingly. This search method is called binary search and minimises the number of values we need to evaluate $\hat{\CDFfunction}$ on. 
We can perform a CDP version of binary search using noisy queries of the form ``is $\hat{\CDFfunction}(x)+\mathcal{N}(0,\sigma^2)\le k_L$?" In setting the variance $\sigma$ we need to balance the desire for the noisy query to answer correctly (so the binary search moves in the right direction) and adding enough noise to guarantee privacy.
Notice that if $|\hat{\CDFfunction}(x)-k_L|$ is large, for example if $x$ is inside $\range$ but far from the concentration of the data, then  $\sigma$ can be chosen to be quite large while still ensuring the response is correct with high probability. 
A single query of this form is $\frac{1}{n\sigma}$-CDP \citep{BunS16} so if $|\hat{\CDFfunction}(x)-k_L|$ is large, then we only need to consume a small amount of privacy budget in order to ensure the binary search moves in the right direction.
The privacy guarantee for the entire algorithm is sum of the privacy guarantees for each iterate. Since a priori we don't know $|\hat{\CDFfunction}(x)-k_L|$, at each iteration we'll start with a quite noisy query, then keep decreasing the noise until we are confident which direction to move in. We repeat until the privacy budget is consumed. We search for both the upper and lower end points of the confidence interval using noisy binary search then a final post-processing step on the noisy measurements ensures we release a valid confidence interval. Full details and pseudo-code are given in online supplement~\ref{online supplement:BS}. 
}
At the points $x_i$ at which $\noisyBS$ obtains noisy measurements, 
the noisy measurement $\hat{\CDFfunction}(x_i)+\mathcal{N}(0,\sigma^2)$ can be released in addition to the confidence interval without consuming additional privacy budget. While this is not as informative as the full CDP CDF released using $\CDF$, it does provide useful additional information about the distribution. We note that unlike $\EM$ and $\CDF$, our analysis of $\noisyBS$ separates the two sources of randomness coming from sampling and noise added for privacy. An interesting open problem is whether this analysis can be improved by considering the relationship between the two sources of randomness.

\begin{figure}[htbp]
     \begin{subfigure}[b]{0.49\textwidth}
         \centering
         \includegraphics[width=\textwidth]{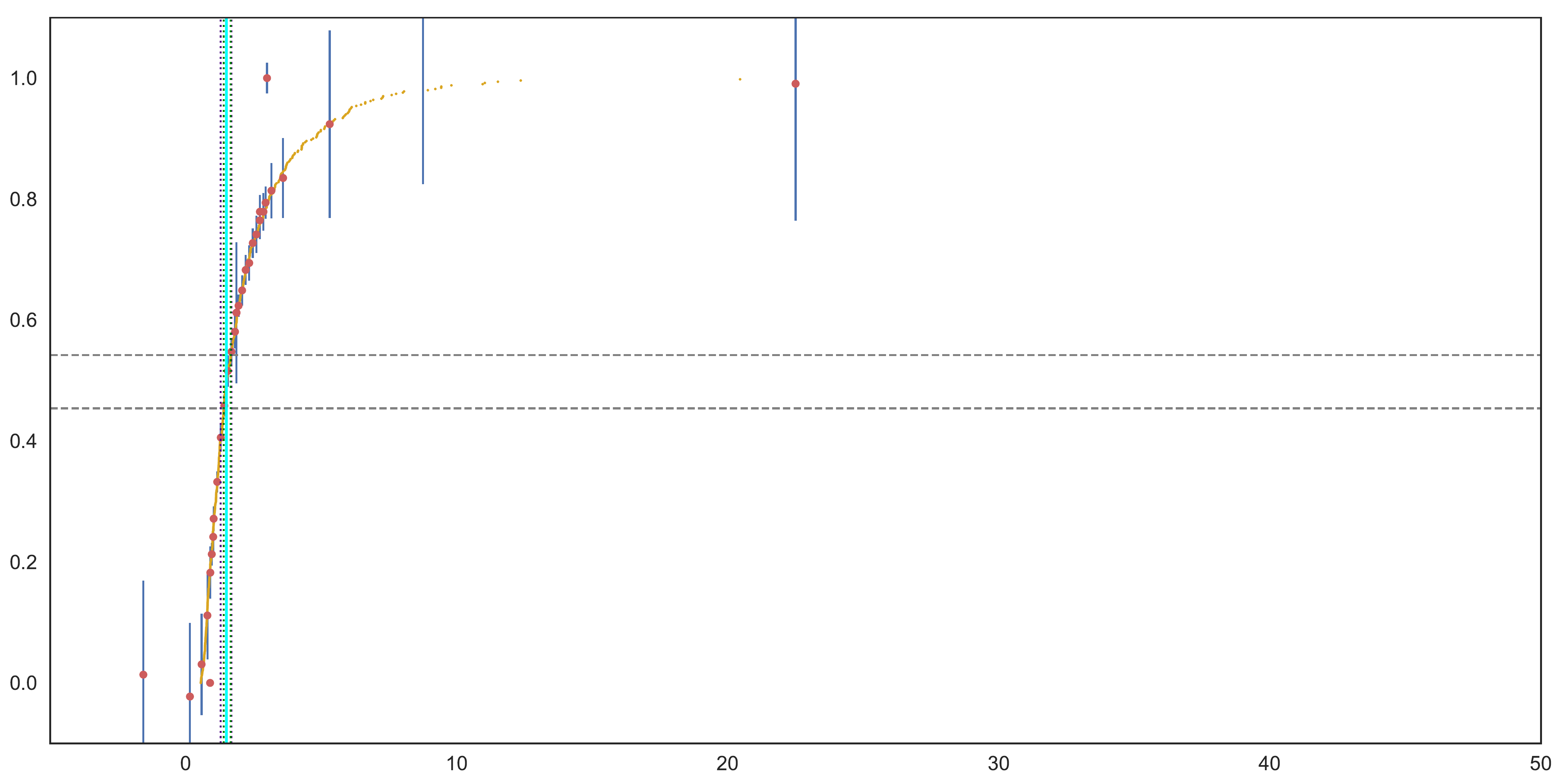}
         \caption{}
     \end{subfigure}
          \begin{subfigure}[b]{0.49\textwidth}
         \centering
         \includegraphics[width=\textwidth]{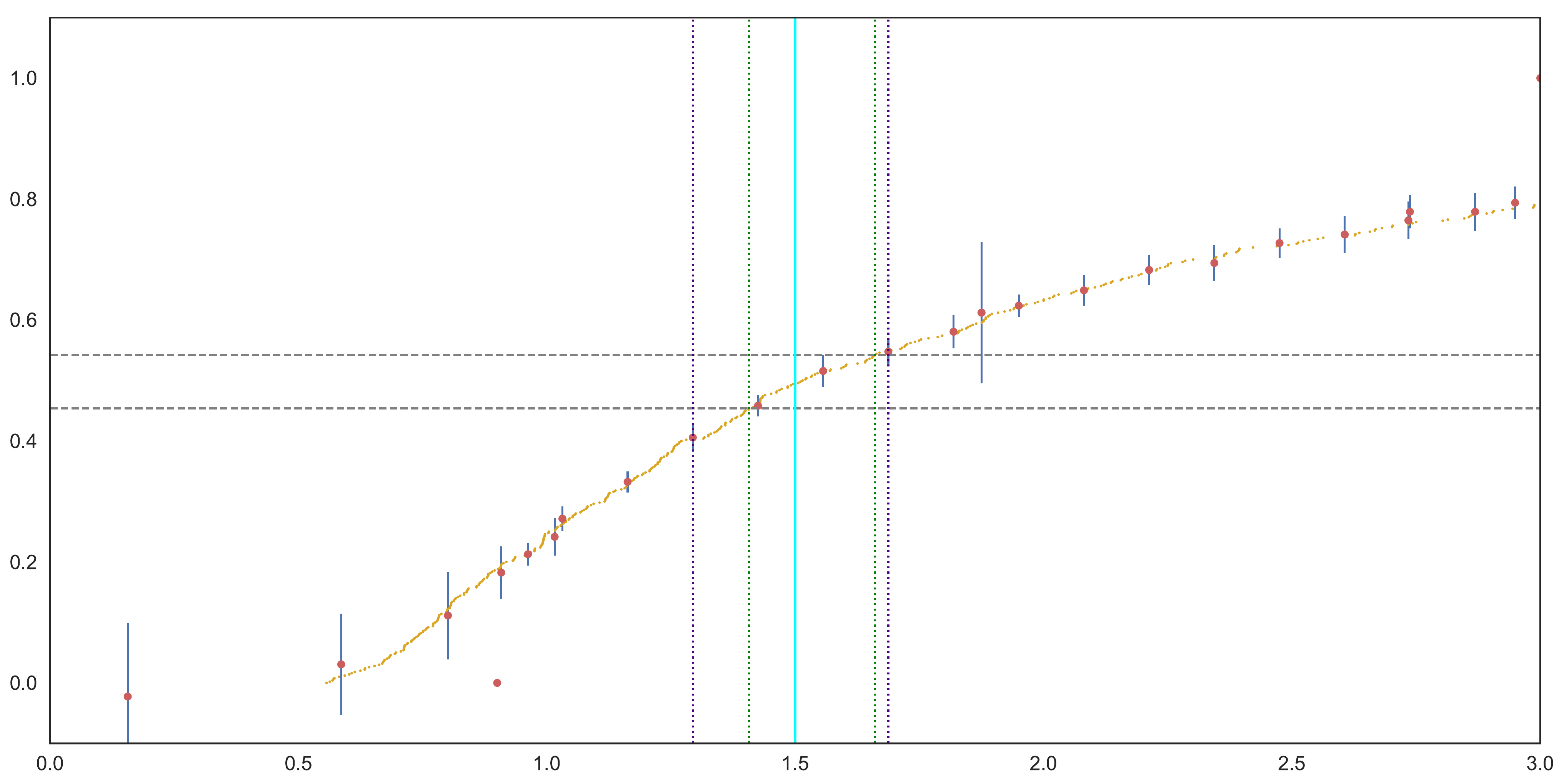}
         \caption{}
     \end{subfigure}
    \caption{Graphical representation of a single run of $0.8$-CDP 95\%-confidence interval using $\BSCDF$ (with hyperparameters $\range = [-5, 50]$, $\theta = 0.1$) on a single dataset. The $500$ datapoints are sampled from lognormal$(\ln(1.5), 1.0)$. The yellow dots represent the empirical CDF, the gray horizontal dashed lines correspond to $\PNPL{0.05}/n, \PNPU{0.05}/n$, and the green vertical dashed lines correspond to $\NPL{0.05}, \NPU{0.05}(d)$. 
    The cyan line is the population median. The red dots correspond to DP estimates, with $99 \%$ error bars in blue. Figure (a) shows how $\noisyBS$ is used to narrow down the range from [-5, 50] to roughly [0.8,3], and Figure (b) zooms in to display the estimates made by $\CDF$ within this smaller range. The vertical purple dashed lines represent the DP interval that $\BSCDF$ outputs.
    }
    \label{algoexplainBSCDF}
\end{figure}

\subsection{Range-robust estimator based on CDF estimator, $\BSCDF$}

The mechanism $\noisyBS$ was designed to improve the performance of $\CDF$ when the data is concentrated in a small region within $\range$. That is, when $\noisyBS$ can greatly reduce the search range for the median using a small amount of privacy budget. However, when the data is well spread out within the range, we do not expect $\noisyBS$ to perform as well as $\CDF$ which is highly optimised to provide an accurate approximation to the CDF, $\CDFfunction$. Our final algorithm, $\BSCDF$ first uses $\noisyBS$ with a portion of its privacy budget to narrow down the search space $\range$. It then clips the data to this range, and runs $\CDF$ with the remaining privacy budget within this smaller range to obtain a confidence interval. Full details and validity proofs are given in online supplement~\ref{online supplement:BSCDF}.

\removeforsubmission{The privacy budget and coverage $\alpha$ both need to be partitioned between the two stages of the algorithm. We expect the optimal split to be distribution dependent. In particular, it likely depends on how large a region the data occupies within the range $\range$. We found experimentally, for the parameter regimes we studied, using $\rho/4$ for the first step, and $3\rho/4$ for the second step seemed to be a good choice. Similarly, we ensure that the region found in the first setp contains the median with probability $1-\alpha/4$, and the second step finds a $1-3\alpha/4$-confidence interval within that region.}
A representation of a single run of $\BSCDF$ is shown in Figure~\ref{algoexplainBSCDF}. Notice that relative to $\CDF$ in Figure~\ref{algoexplainCDF}, it uses only a few measurements to narrow into the range of interest, about $[0.8,3]$ then takes more measurements within this range.

\subsection{A note on hyperparameters} 

All our algorithms require some hyperparameter tuning and domain knowledge. Throughout this manuscript we will attempt to give guidance on how one may set these parameters and how sensitive the performance is to these choices. 

All our algorithms require as input a range space for the median. That is, an interval $\range\subset\reals$ that is promised to contain $\median(P)$. This should be chosen as small as is reasonable, but one can typically be quite conservative when defining $\range$. We expect the dependence of all algorithms to be approximately $\log|\range|$, we see this explicitly for $\EM$ in Lemma~\ref{lem:exp-mech-output-goodness} in the online supplement. We will empirically explore the dependence on $|\range|$ in Figure~\ref{fig:rel-width-boxplots-range} below. 
Note that requiring a range on the median is different to requiring a range on the data points themselves. It is a preferable condition since the data points may occupy a considerably larger range than $\range$, and, in practice, guaranteeing a bound on outliers in the data may be difficult. The fact that a bound on the median suffices comes from the fact that if $\range$ contains the median then projecting the data into $\range$ leaves the median unchanged.

\removeforsubmission{All the algorithms we consider require an additional hyperparameter we refer to as the ``granularity" parameter, $\granularity$.
In order to gain some intuition, one can imagine each algorithm as discretizing the range $\range$ so that any two potential output points are $\granularity$ apart. This intuition is not exact, please refer to the online supplement for more details. 
This granularity parameter affects each algorithm differently. $\EM$ is not very sensitive to this parameter in general, and it can typically be taken to be very small. In fact, setting this parameter to 0 is a reasonable choice in a wide variety of parameter settings. $\CDF$ and $\BSCDF$ can be sensitive to this parameter. In all algorithms, the width of the confidence interval will be at least $\granularity$.}

\section{Simulation Studies}
\label{sec:simulation}
In this section we present extensive simulation studies to evaluate the different algorithms under various parameter settings.\footnote{Code for producing these simulations can be found at \url{https://github.com/anonymous-conf-medians/dp-medians}.} 
We focus on log-normal data, as this is the natural use-case for computing a non-private median and corresponding confidence interval.
We expect many of our findings to extend to other types of skewed data. 
We evaluate the performance of the four algorithms described in Section~\ref{algorithms}, as well as the non-private confidence interval described in Lemma~\ref{nonprivCI}, in terms of width of confidence interval, coverage, and bias.

\subsection{Data description} 

In order to visualize the distribution of the noisy confidence intervals, we run each private algorithm 5 times on 100 independently drawn datasets.
Let $\bx_1, \ldots, \bx_{100}$, each contain $n=1000$ i.i.d. draws from the underlying log-normal distribution. The underlying normal random variable has mean $\mu = \ln(1.5)$ 
and standard deviation of either $\sigma = 1$ or $\sigma = 5$.
We run each DP algorithm on each $\bx_i$ for $5$ trials. In our experiments, we show the relative performance of the algorithms as we vary $n$, $\rho$, $\sigma$, and $|\range|$.

\subsection{Utility measures}

We consider two main utility measures in our experiments. The first measure is the relative width of the CDP confidence interval $M(d)$ compared to the width of the non-private confidence interval, $[\NPL{\alpha}, \NPU{\alpha}]$. For a data set $d \in \mathcal{D}^n$, $\alpha \in [0,1]$, and interval $I=[I_L, I_U]$, the relative width is defined as
\begin{align*}
    \relwidth^{\alpha}(d, I) = \frac{I_U - I_L}{\NPU{\alpha}(d) - \NPL{\alpha}(d)}
\end{align*}
We are concerned with the distribution of $\relwidth^{\alpha}(d, I)$ when $I=M(d)$.
We expect the private confidence intervals to be wider than the non-private confidence intervals, so if $I=M(d)$ then we expect $\relwidth^{\alpha}(d)\ge 1$ with high probability.
If $\relwidth^{\alpha}(d)\le 2$, then, intuitively, the additional uncertainty due to privacy is less than the uncertainty due to sampling. We are interested in the distribution of the relative width over multiple trials, so for each algorithm we show box plots of this metric over 500 trials (100 datasets times 5 trials of the DP mechanism on each). 

The second utility measure is \emph{empirical} coverage of the DP confidence interval. For intervals $I_1, \cdots, I_T$ and distribution $P$, let
\[\cov_{T}(P, I_1, \cdots, I_T) = \frac{1}{T}\sum_{t=1}^T \indicator_{\median(P)\in I_t}.\] Given $n\in\mathbb{N}$ and a $1-\alpha$-confidence interval $M$, for all $t\in[T]$, let $d_t\sim P^n$ and $I_t=M(d_t)$ then \[\cov_{T, n, P}(M) = \cov_{T}(P, I_1, \cdots, I_T)\] gives an estimate of the empirical coverage of $M$ on the distribution $P$. We estimate the coverage over 5,000 trials (1,000 drawn samples of size $n=1,000$ times 5 trials of the DP mechanisms on each). A key component of the confidence intervals presented in the previous section is that the coverage should be at least $1-\alpha$. However, the empirical coverage may, and in many settings will, exceed $1-\alpha$.

\subsection{Results and Discussion}

\begin{figure}[hptb]
     \begin{subfigure}[b]{0.49\textwidth}
         \centering
         \includegraphics[width=\textwidth]{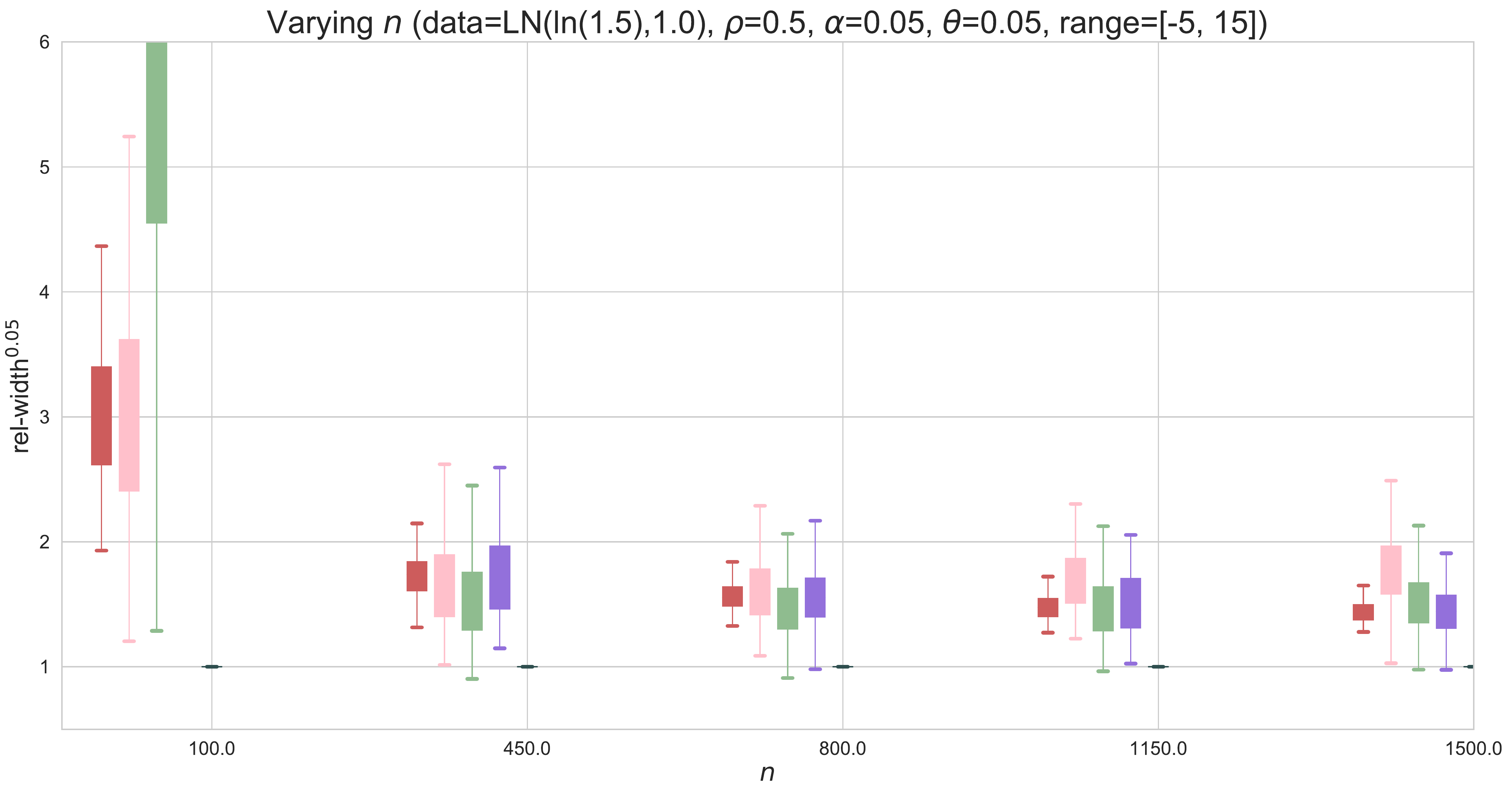}
         \caption{Varying $n$}
         \label{fig:rel-width-boxplots-n}
     \end{subfigure}
     \hfill
     \begin{subfigure}[b]{0.49\textwidth}
         \centering \includegraphics[width=\textwidth]{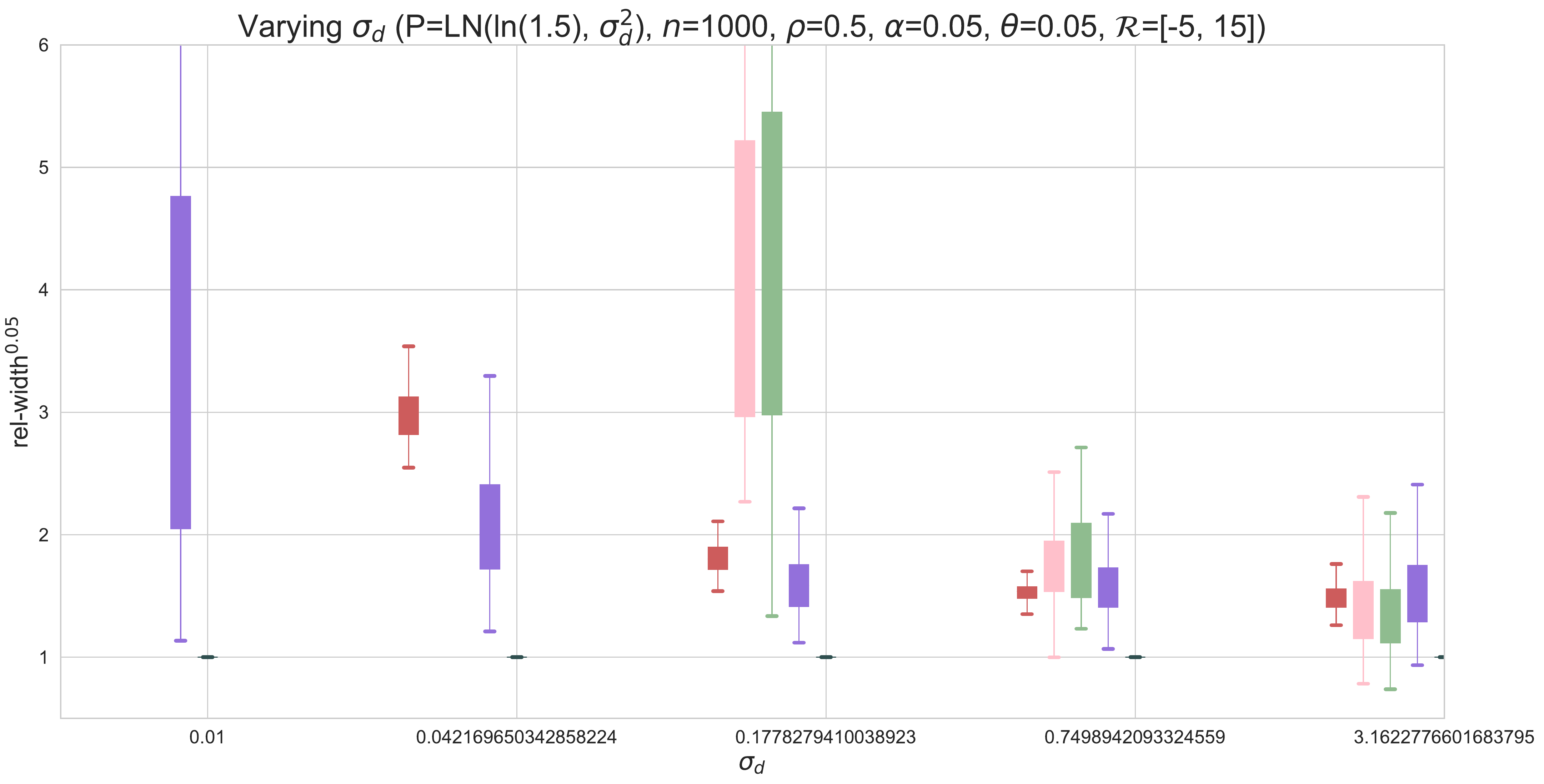}
         \caption{Varying $\sigma_d$}
         \label{fig:rel-width-boxplots-sigma}
     \end{subfigure}
     \begin{subfigure}[b]{0.5\textwidth}
         \centering
         \includegraphics[width=\textwidth]{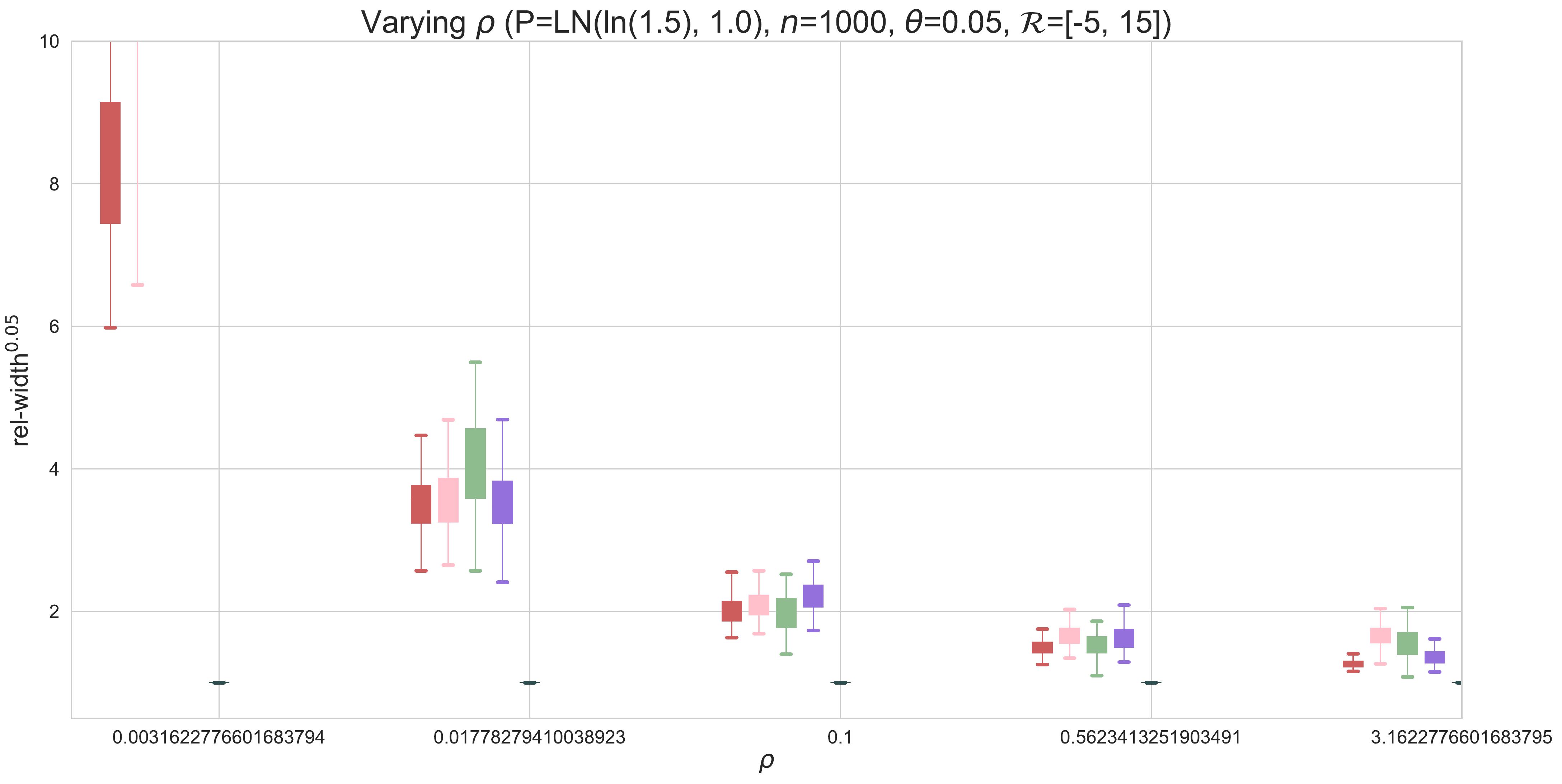}
         \caption{Varying $\rho$ 
         }
         \label{fig:rel-width-boxplots-rho}
     \end{subfigure}
     \begin{subfigure}[b]{0.5\textwidth}
         \centering
         \includegraphics[width=\textwidth]{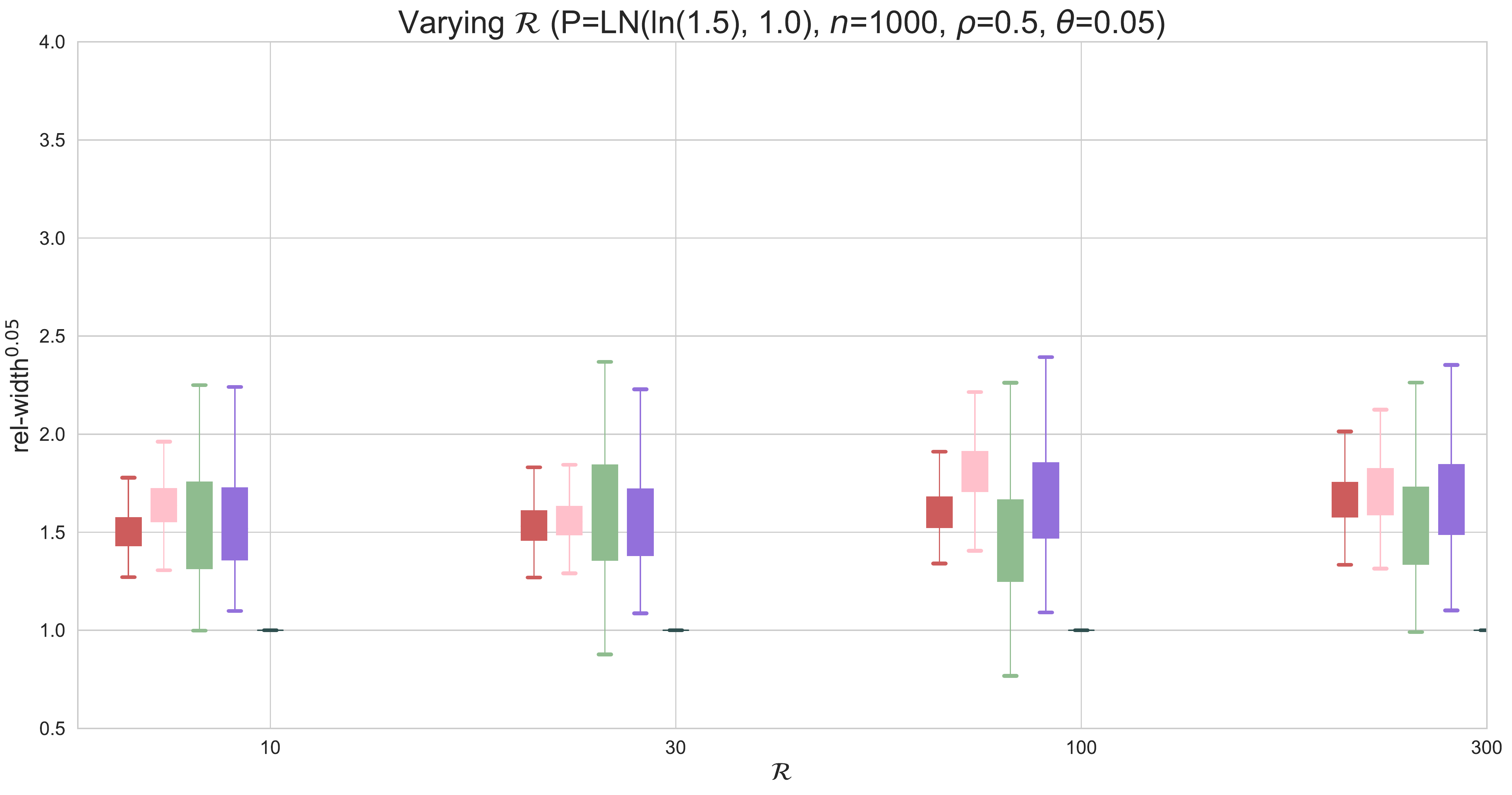} \caption{Varying $|\range|$}
         \label{fig:rel-width-boxplots-range}
     \end{subfigure}
     \begin{subfigure}[b]{\textwidth}
        \centering
        \includegraphics[width=\textwidth]{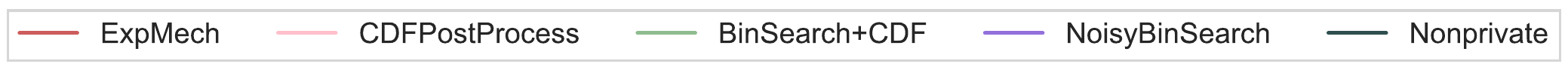}
     \end{subfigure}
     \caption{Relative widths of DP confidence intervals as we vary (a) dataset size $n$, (b) dataset standard deviation $\sigma_d$, (c) privacy parameter $\rho$, and (d) size of range $|\range|$ on log-normal data. By definition, $\relwidth^{\alpha}(d, I)=1$ when $I=[\NPL{\alpha}, \NPU{\alpha}]$.
     }
     \label{fig:rel-width-boxplots}
\end{figure}

\paragraph{Comparison Among Algorithms}

\begin{figure}[hbtp]
\centering
     \begin{subfigure}[b]{0.49\textwidth}
         \centering
         \includegraphics[width=\textwidth]{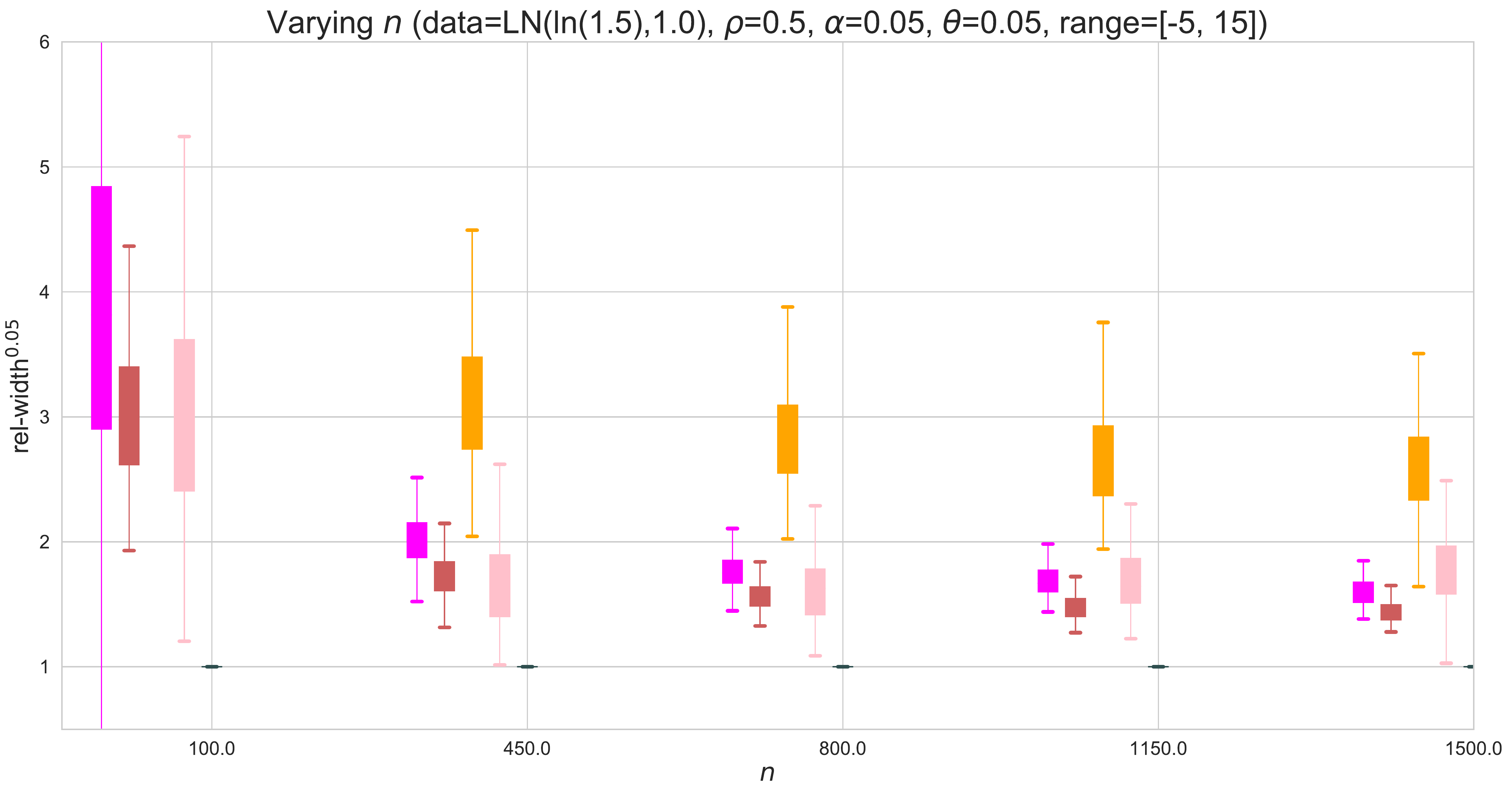}
         \caption{ Relative widths as we vary $n$}
         \label{fig:naive-v-new-widths-n}
     \end{subfigure}
     \hfill
     \begin{subfigure}[b]{0.49\textwidth}
         \centering
         \includegraphics[width=\textwidth]{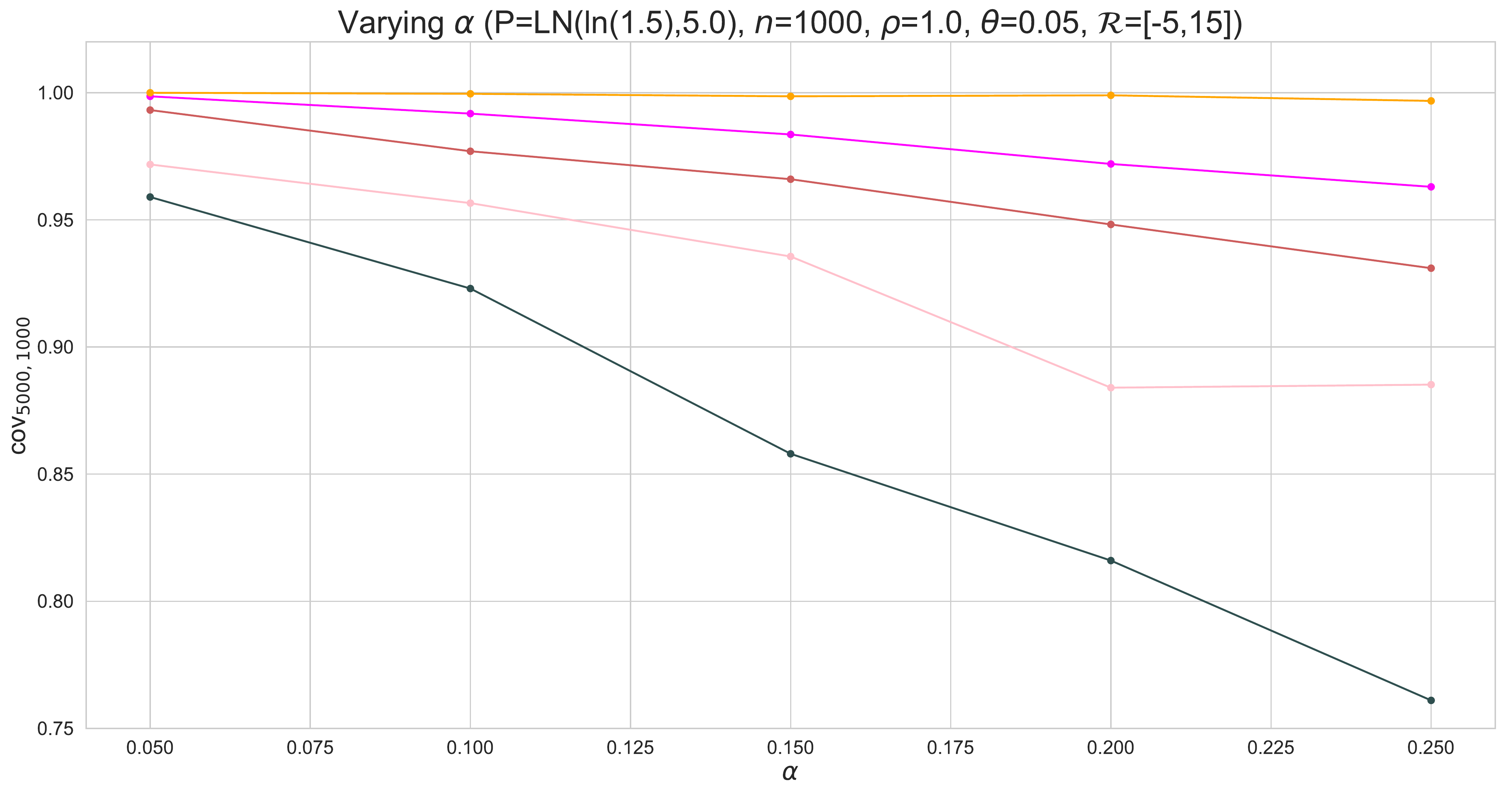}
         \caption{Coverage as we vary $\alpha$ }
         \label{fig:naive-v-new-coverage-alpha}
     \end{subfigure}
     \begin{subfigure}[b]{\textwidth}
        \centering
        \includegraphics[width=\textwidth]{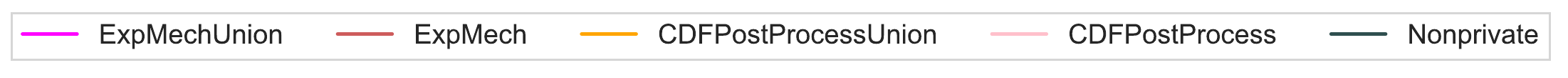}
     \end{subfigure}
     \caption{Performance of naive vs. more careful DP confidence intervals
     }
     \label{fig:naive-v-new}
\end{figure}

\begin{figure}[hbtp]
\centering
     \begin{subfigure}[b]{0.49\textwidth}
         \centering
         \includegraphics[width=\textwidth]{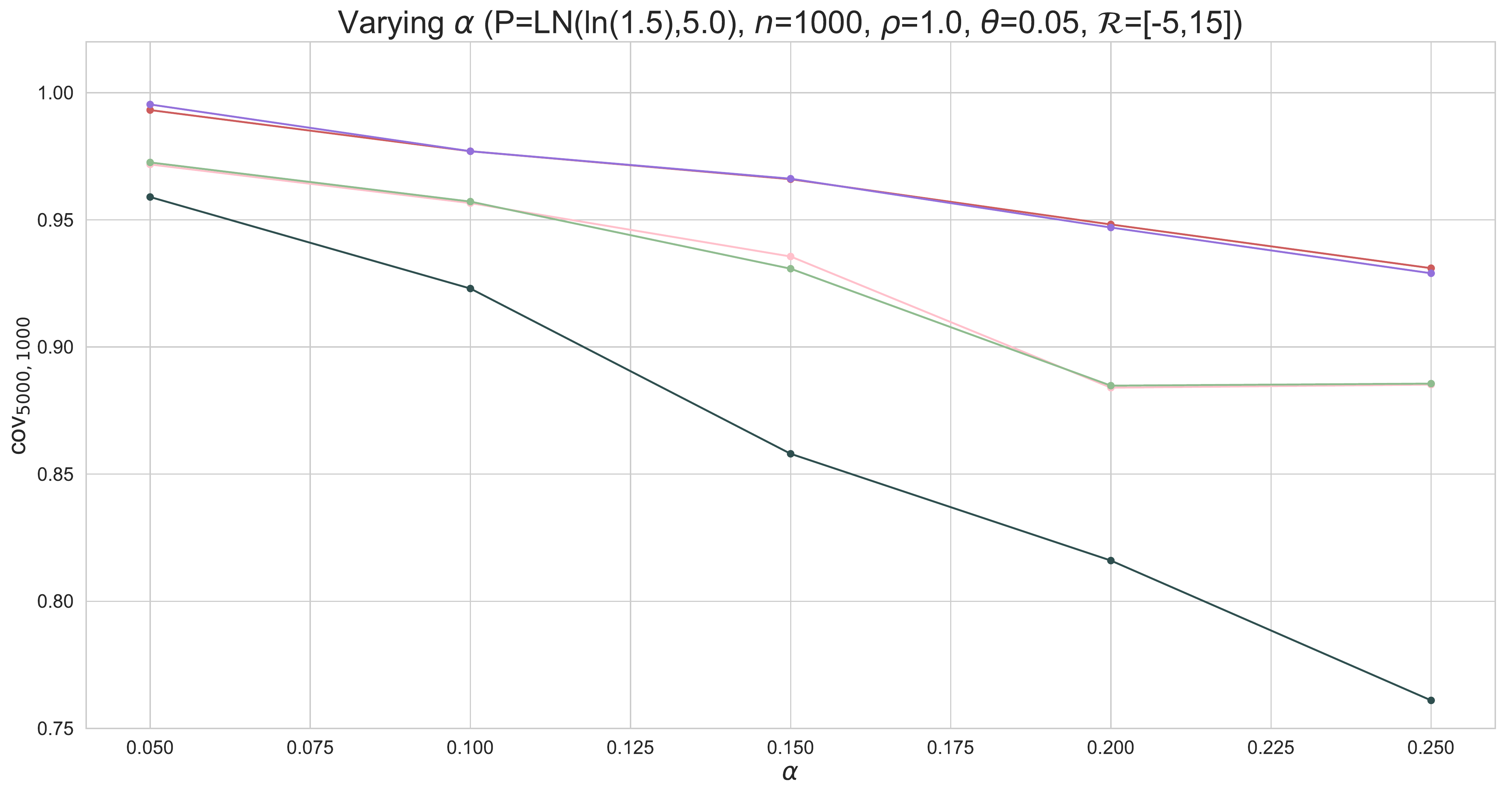}
         \caption{ Varying $\alpha$ 
         }
         \label{fig:coverage-alpha}
     \end{subfigure}
     \hfill
     \begin{subfigure}[b]{0.49\textwidth}
         \centering
         \includegraphics[width=\textwidth]{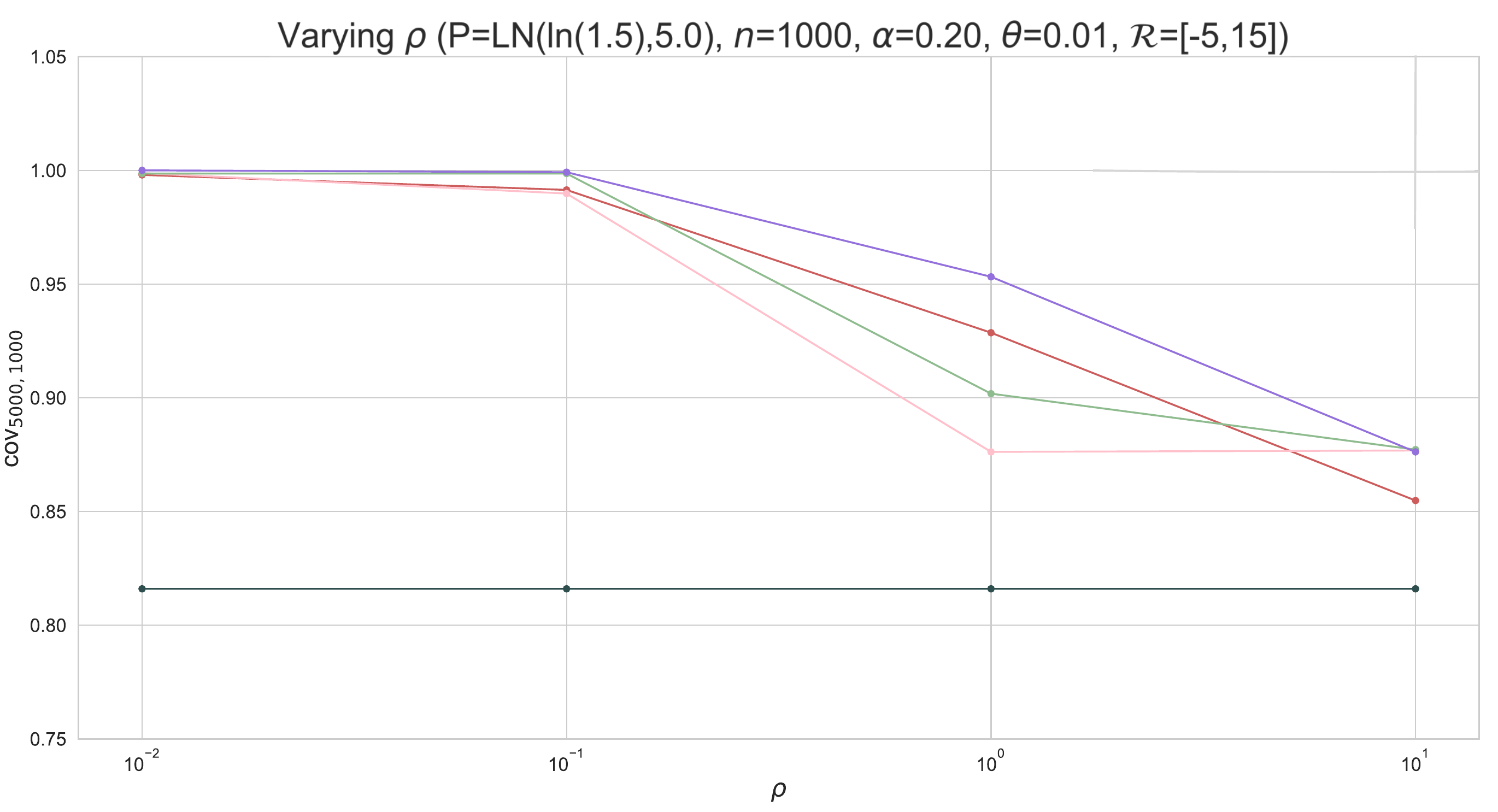}
         \caption{Varying $\rho$ }
         \label{fig:coverage-rho}
     \end{subfigure}
    \begin{subfigure}[b]{\textwidth}
        \centering
        \includegraphics[width=\textwidth]{main-labels.pdf}
     \end{subfigure}
     \caption{Coverage of DP confidence intervals as we vary $\alpha$ and $\rho$ 
     }
     \label{fig:coverage}
\end{figure}

Figure~\ref{fig:rel-width-boxplots} demonstrates the performance of our four CDP confidence interval algorithms across a range of parameter regimes on log-normal data, in terms of the relative width metric. Notice that in a variety of regimes, including large $n$, large $\rho$ and large $\sigma_d$ all of the CDP algorithms provide confidence intervals that are at most twice the width of the non-private confidence interval with high probability. 
Our results indicate that $\EM$ provides the tightest, or close to the tightest, confidence intervals in most parameter regimes we studied. 
This algorithm is the most targeted of the CDP algorithms we discuss and is carefully calibrated to not waste privacy budget on estimating additional information about the underlying distribution. 
It is a good general choice when one is solely interested in confidence intervals for the median. There are a few regimes in which the other algorithms outperform $\EM$ which we will discuss in this section.

The $\CDF$ algorithm is appealing in practice since it allows a CDP estimate of the CDF to be released without consuming additional privacy budget. This can be used not only to produce confidence intervals for the median, but to produce a wealth of other insights about the distribution $P$. 
Surprisingly, in a variety of parameter regimes, $\CDF$ provides confidence intervals that are almost as tight as those obtained by $\EM$. In fact, when $\sigma_d$ is large, $\CDF$ can result in tighter confidence intervals than $\EM$ (Figure~\ref{fig:rel-width-boxplots-sigma}). We explore this further in Appendix~\ref{online supplement: otherregimes}. Conversely, when $\sigma_d$ is small, or $|\range|$ is large, $\CDF$ is not a good choice. As discussed earlier, these are regimes where $\CDF$ spends a lot of its privacy budget estimating the CDF in regions that are far from the median.

The remaining CDP confidence intervals $\BSCDF$ and $\noisyBS$ were designed to outperform $\CDF$ when the data is concentrated in a small subset of $\range$. This setting is shown in Figure~\ref{fig:rel-width-boxplots-sigma}, when the data is very concentrated and Figure~\ref{fig:rel-width-boxplots-range} when $|\range|$ is large. In both settings we see either $\BSCDF$ or $\noisyBS$ outperforming the other CDP algorithms. $\BSCDF$ involves a hyperparameter that decides how much of the privacy and coverage budget is spend on each step of the algorithm. In our experiments, 
we used one quarter of the privacy budget to perform the range-finding step (using $\noisyBS$), and three quarters to find the confidence interval within that region (using $\CDF$). This choice was made experimentally by testing multiple choices for this partition. Since this parameter interpolates between $\CDF$ and $\noisyBS$, the optimal choice appears to be data dependent. In particular, we conjecture that the smaller the fraction of the range that is occupied by the data, the higher the fraction of the budget that should be allocated to the $\noisyBS$ step. Of course, one typically does not know a priori how concentrated the data is, so the (1/4, 3/4) seems to be a reasonably safe split that performs well in a variety of contexts.

\paragraph{Empirical coverage analysis.}

A key component of the algorithmic design of each of the CDP confidence intervals was the coverage analysis. We discussed in Section~\ref{randomnessdiscussion} how a careful coverage analysis that leverages the relationship between the two sources of the randomness in the CDP confidence intervals potentially results in a much tighter coverage analysis than the naive analysis that separates the sources of randomness. Our experimental results presented in Figure~\ref{fig:naive-v-new} highlight two key findings regarding the coverage; that the careful analysis does result in substantially tighter intervals, and that the empirical coverage of the CDP confidence intervals is still notably above the target coverage.

Figure~\ref{fig:naive-v-new-widths-n} compares the relative width of the different confidence intervals. $\texttt{ExpMechUnion}$ and $\texttt{CDFPostProcessUnion}$ refer to the versions of $\EM$ and $\CDF$ resulting from the naive coverage analysis. While the relative width is only slightly reduced for $\EM$, the relative width of $\CDF$ is almost halved if the improved approach is used to produce the confidence intervals. These findings are also reflected in   Figure~\ref{fig:naive-v-new-coverage-alpha}, which compares the empirical coverage of the naive coverage analyses of $\texttt{ExpMechUnion}$ and $\texttt{CDFPostProcessUnion}$ and the more careful analyses described in Section~\ref{algorithms}. While theoretically we can show that the careful analysis will result in empirical coverage that is much closer to the target coverage, Figure~\ref{fig:naive-v-new-coverage-alpha} shows that this improvement is practically relevant. While the improved analysis only leads to modest reduction in the overcoverage for $\EM$, the changes for $\CDF$ are more substantial. The naive approach results in coverage rates that are close to 1 irrespective of the selected $\alpha$. The improved approach leads to coverage rates that are much closer to the nominal coverage rates. This highlights the importance of designing algorithms that consider the relationship between the two sources of randomness. 

Despite the substantial improvement, Figure~\ref{fig:coverage-alpha} shows that all the CDP algorithms exhibit empirical coverage higher than the target coverage for moderate values for $\rho$. As expected, Figure~\ref{fig:coverage-rho} reveals that the coverage improves with $\rho$ as each algorithm trends towards outputting the non-private confidence interval $[\NPL{\alpha}, \NPU{\alpha}]$ when $\rho=\infty$. Over-coverage does not necessary correspond to substantially larger confidence intervals. We see in Figure~\ref{fig:rel-width-boxplots} that in a wide range of parameter regimes our CDP algorithms still result in confidence intervals that are at most twice as wide as their non-private counterparts. However, it does suggest an opportunity for improvement. An important question for future work is to what degree this over-coverage is necessary? In particular, is there an inherent tension between privacy guarantee, and learning enough about the data set to accurately quantify the uncertainty? 

In many estimation tasks defining a non-parametric confidence interval that gives close to nominal coverage rates is difficult. Without information regarding the underlying distribution, the confidence intervals need to be wide enough to ensure valid coverage rates for any possible distribution. This can result in the non-parametric confidence intervals having higher than expected empirical coverage when the data is drawn from a nice distribution, e.g a log-normal distribution.  
This effect is one possible explanation for the fact that the CDP confidence intervals have empirical coverage higher than $1-\alpha$ in Figure~\ref{fig:coverage}. In fact, we see evidence of this in the analysis of $\EM$. The error of the exponential mechanism is data dependent (and hence distribution dependent), but in our coverage analysis we are forced to use the worst case error of the exponential mechanism over all datasets. 

\begin{figure}[hbtp!]
\centering
     \begin{subfigure}[b]{0.6\textwidth}
         \centering
         \includegraphics[width=\textwidth]{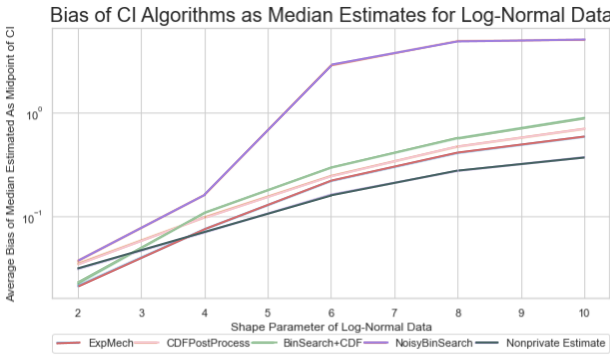}
     \end{subfigure}
     \caption{Bias of the algorithms (average of the difference between the value obtained and the true median) for log-normal data centered at 1 as the shape parameter of the distribution (the variance of the normal distribution that is exponentiated) increases from 2 to 10.}
      \label{fig:bias}
\end{figure}

\paragraph{Bias.} The goal of this work was to design algorithms that output valid confidence intervals for the median, not to estimate the median itself. An ad-hoc estimate of the median can be obtained from a confidence interval by taking the estimate to be the mid-point of the interval. This approach is preferable in the differential privacy context since it allocates its entire budget to the object of interest (the confidence interval) and we discussed in Section~\ref{roadblocks} some of the reasons why direct estimators of the median are difficult to generalise to CDP confidence interval algorithms.
For all of our CDP algorithms, as well as the non-private confidence interval, this results in a biased estimator for the median, if the underlying distribution is skewed. 
In Figure~\ref{fig:bias}, we explore the bias of the inherited median estimators. As expected, the bias increases with the skew of the data. The bias of most of the DP algorithms (except for $\noisyBS$) is not substantially different from the bias for the non-private estimate. This implies that  most of the bias can be attributed to the ad-hoc strategy of using the mid-point of the confidence interval as the point estimate for the median. As mentioned earlier, one benefit of $\CDF$, $\noisyBS$ and $\BSCDF$ is that they come with additional information about the distribution which could potentially be used to release a less biased estimate of the median. We leave this for future work.

\section{Real data application}
\label{sec:application}

\begin{table}[htpb!]
    \centering
    \resizebox{\textwidth}{!}{
    \begin{tabular}{|c|c|c|c|c|c|c|}
        \hline
        \hline
         Characteristic & Number in & CDP median  & DP  & Non-private & Non-private & Population \\
          of & 1\% sample & income & 90\% CI & median & 90 \% CI  & median \\
          household(er) & & & 
          & income & 
          & income \\
          & & & & & & \\
         \hline
         \hline
         \textbf{Type of Household} & & & & & &\\
         \hline
        Family households & 9,142 & 489.00 & (469.99, 508.01) & 499.97 & (480.01, 500.93)& 499.95  \\
         \hline
         Nonfamily households & 1,479 & 65.50 & (0.0*, 136.01)& 20.12 & (0.17, 99.89) & 0.20 \\
         \hline
        \textbf{Metropolitan status} & & & & & & \\
         \hline
         Not in metropolitan area & 8243 & 324.99 & (290.03, 359.95) & 329.85 & (300.06, 359.85) & 360.07 \\
         \hline
         In metropolitan area & 2380 & 708.12 & (640.00, 776.23) & 699.99 & (659.94, 749.85) & 699.91 \\
         \hline
         \textbf{Age} & & & & & & \\
         \hline
         Age $< 65$ & 9,259 & 564.89 & (529.94, 599.84) & 560.05 & (540.03, 597.95) & 540.10 \\
         \hline
         Age $ \geq 65$ & 1,366 & 0.00 & (0.0*, 5.0) & 0.03 & (0.02, 0.03) & 0.03 \\
         \hline
         \hline
    \end{tabular}
    }
    \caption{Income summary measures by selected characteristics based on 1\% simple random sample of mountain division householder records (i.e., in Arizona, Colorado, Idaho, Montana, Nevada, New Mexico, Utah, Wyoming) from the 1940 decennial census. Income is shown in 1940s dollars and is top-coded at $\$5001$. Differentially private estimates are obtained using $\EM$ on a single sample with total privacy budget $\rho=0.5$, range $\range = [0, 5001]$ and granularity $\theta=5$. Note that the lower DP confidence interval values with a * are truncated to 0 based on the assumption that incomes are nonnegative. The DP point estimates are computed before the truncation to avoid introducing bias.
    All values are rounded to the nearest cent.
    }
    \label{tab:Census_results}
\end{table}

In this section, we illustrate how the findings from the previous sections could inform the implementation of a differentially private median release strategy in practice. We also demonstrate what level of accuracy one could reasonably expect for realistic applications. Our motivating example is the median income tables published by the U.S. Census Bureau for various subgroups of the population. Specifically, we aim to replicate a subset of statistics from Table \textit{A1.Income Summary Measures by Selected Characteristics: 2018 and 2019} \citep{poverty19}. This table reports median household income broken down by the following characteristics: Type of household, Race and Hispanic Origin of Householder, Age of Householder, Nativity of Householder, Region, and Residence. For each of the 32 subgroups specified, the table provides the estimated median income and estimated margin of error (based on $\alpha=0.1)$ for 2018 and 2019. The estimates are computed using the  Current Population Survey, 2019 and 2020 Annual Social and Economic Supplements (CPS ASEC).		

Since we want to assess the accuracy of the CDP estimates, we use income data from the 1940 Decennial Census~\citep{Census1940Dataset}, which enables us to compare the noisy estimates to the true values in the population. We restrict the population data to heads of households
in the mountain division region (Arizona, Colorado, Idaho, Montana, Nevada, New Mexico, Utah, and Wyoming) and focus on the variables type of household (two categories), metropolitan area (two categories) and age (two categories). To mimic the illustrative application described above, we repeatedly sample from this population and treat the resulting data as the survey from which the (noisy) estimated medians will be computed. For simplicity, we draw 1\% simple random samples without replacement. We acknowledge that the sampling design for the CPS ASEC is far more complex. However, as indicated in the introduction understanding the subtle effects of complex sampling designs on the privacy guarantees is currently an area of active research and is beyond the scope of this paper.

The variable we use for our evaluations is `INCWAGE,' which ``reports each respondent's total pre-tax wage and salary income for the previous year." The amounts are displayed in ``contemporary dollars," which means they are not adjusted for inflation. In the 1940's dataset, the variable is topcoded at $5,001$ dollars. We remove all N/A and missing values from the dataset, and only consider records corresponding to the head of each household. We note that we do not propose simply dropping all cases with missing values in practice as this will likely introduce bias. However, properly integrating any non-response adjustments into the DP algorithms is beyond the scope of the paper. Thus, we treat the fully observed data on household heads as our population of interest. Finally, for the purpose of error analysis we treat the empirical distribution of the entire population dataset (from which we sample 1\%) as the true underlying distribution $P$.\footnote{Note that 
many respondents report incomes rounded to the nearest \$5, \$10 or \$50 dollars, which results in substantial overcoverage even for the non-private confidence intervals due to the spikes in the data. Therefore, we add a negligible amount of noise ($\mathcal{N}(0,0.01)$) to each population income data point, so that the distribution being sampled from is continuous. See Appendix~\ref{convolutionisgood} for further discussion.
} 

\subsection{Selecting the algorithm and (hyper)parameters}
To generate the privatized confidence intervals, we use the algorithm identified as the winner in a wide range of regimes in the simulation studies: $\EM$. The hyperparameters are set to $\range=[0, 5001]$, and $\theta=5$. The lower and upper bounds are chosen based on the assumption that the threshold used for top coding is public knowledge and the reasonable assumption that the median income will not be less than zero. The granularity parameter $\theta$ is chosen based on Census Bureau data visualizations that report median incomes from 1967 to present, which are rounded to the nearest $\$100$ \citep{Census_Press} indicating that a granularity of $\$5$ for 1940 median incomes is likely sufficient for data users. 
We split the overall privacy budget of $\rho=0.5$ equally across three characteristics: type of household, metropolitan status, and age. 

\subsection{Results}

\begin{figure}[hptb!]
    \centering
    \includegraphics[width=0.8\textwidth]{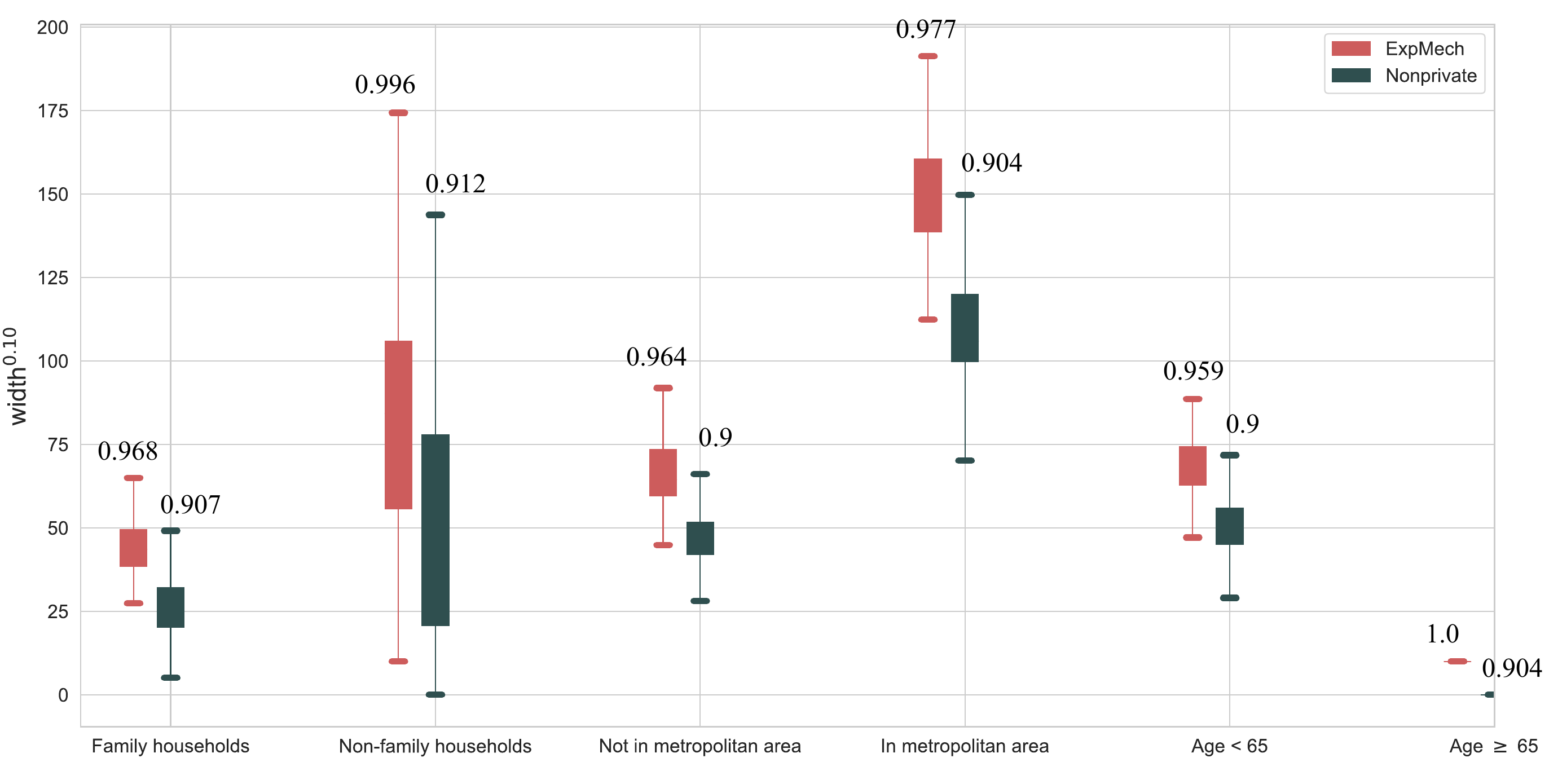}
    \caption{Comparing widths of $90\%$ $\EM$ and non-private confidence intervals. Algorithms are run
    on 1,000 samples of income data by selected characteristics from the 1940 Decennial Census (the DP algorithm is run 20 times for each sample). Empirical coverage rates are displayed for each algorithm. }
    \label{fig:income-widths}
\end{figure}

\begin{figure}[hptb!]
    \centering
    \includegraphics[width=0.8\textwidth]{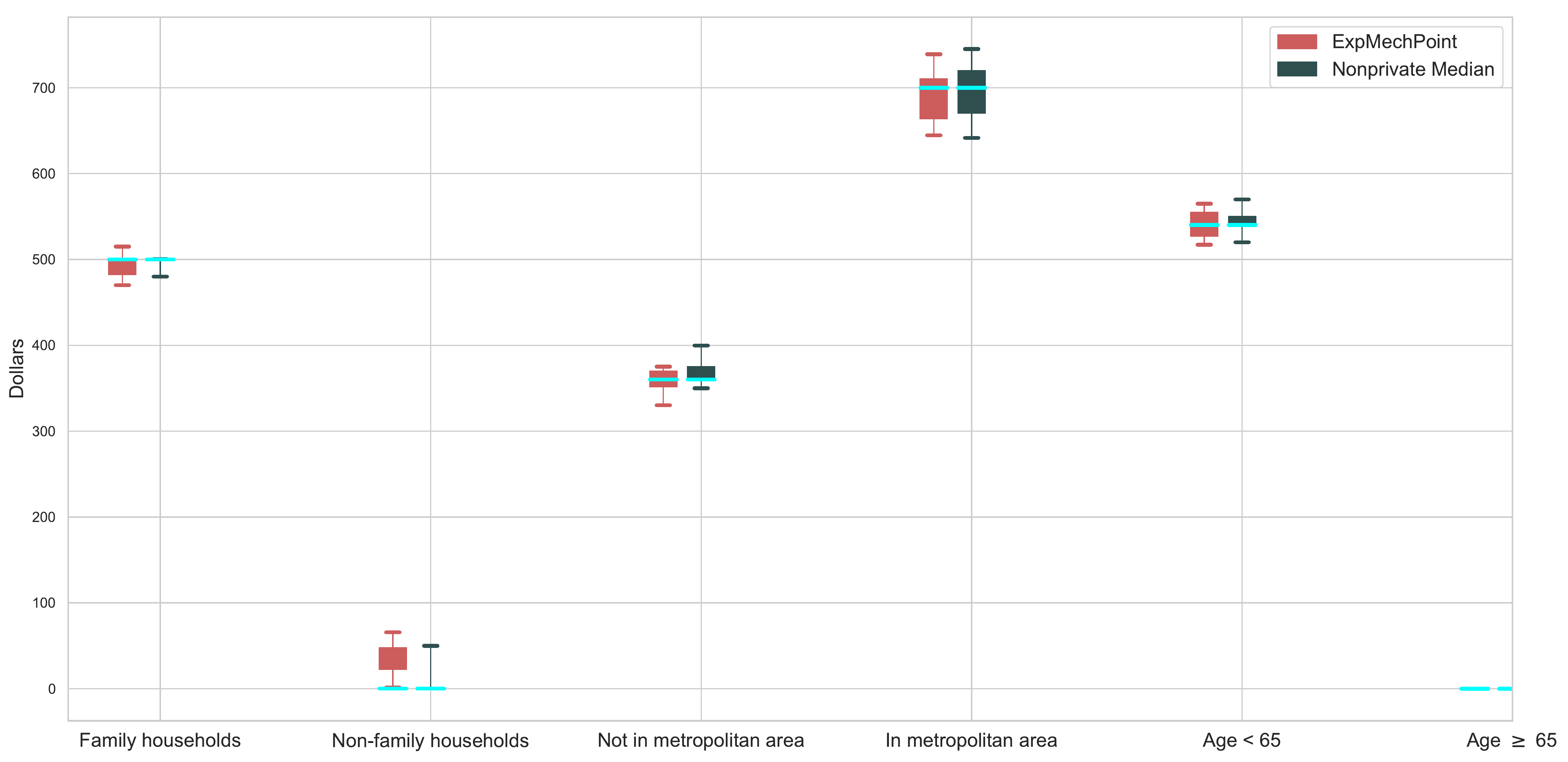}
    \caption{Comparing point estimates of $\EMPointEstimator$ and non-private medians. Algorithms are run on 1,000 samples of income data by selected characteristics from the 1940 Decennial Census (the DP algorithm is run 20 times for each sample). Population medians by characteristic are denoted by the cyan lines. The whiskers of the boxplots denote the 5th and 95th quantiles of the estimates.
    }
    \label{fig:income-point-ests}
\end{figure}

Results based on the first simulation run are included in Table \ref{tab:Census_results}. Note that the CDP confidence intervals and median estimate are the result of a single run of $\EM$ so this table is indicative of what we would expect in practice. 
The CDP point estimates for the median incomes are chosen as the midpoint of the corresponding CDP confidence intervals. 
Note that we could also leverage a prior assumption of the right-skewness of income data by choosing the CDP point estimator from the left half of the CDP confidence interval, rather than from its center, but we leave this type of parametric estimation to future work.
We leverage the assumption that the incomes are non-negative, so we set the lower endpoint of the CDP confidence intervals at the maximum of the output of the algorithm and 0. However, we compute the point estimates before the truncation step to avoid introducing bias. The table also provides non-private and private 90\% confidence intervals (the margin of error reported in the Census tables could be computed as the half-width of these intervals).

The non-private median estimates are closer to the true values than the DP estimates for all sub-populations. However, for many statistics the difference between the point estimates is small relative to the width of the confidence interval indicating that the bias introduced by the ad-hoc approach of using the center of the confidence interval as the point estimator for the median is only minor. Except for the householders aged 65 and above, the relative increase in uncertainty also seems to be acceptable. The relative increase in the length of the confidence intervals ranges between 20.7\% and 81.7\%, that is, the uncertainty from data protection is always less than the uncertainty from sampling. The large relative increase of the confidence intervals for householders aged 65 and above can be explained if we note that the width of the CDP intervals is lower bounded by the granularity hyperparameter (which we chose to be $\theta = 5$), which leads to a large relative uncertainty if the non-private interval has width close to 0. However, the absolute increase in uncertainty is still acceptable.

In Figure~\ref{fig:income-widths} and Figure~\ref{fig:income-point-ests} we explore the performance of $\EM$ and the non-private algorithm over 1,000 randomly sampled datasets. Figure~\ref{fig:income-widths}, which contains boxplots showing the width of the private and non-private confidence intervals,  confirms the findings based on one simulation run. While the private confidence intervals are typically wider than the non-private intervals, the increase in width is less than a multiplicative factor of two for all sub-populations except for head of households older than 65, and less than \$100 in all sub-populations. The figure also reports the coverage rates for the non-private and private confidence intervals. The coverage rates are computed over 1,000 simulation runs and 20 trails of the CDP algorithm within each simulation run. While the non-private coverage rates are close to the nominal 90\% coverage, the CDP confidence intervals overcover substantially. The coverage rates vary between 95.9\% and 100\%. The results are in line with the simulation studies. 

Figure \ref{fig:income-point-ests} contains boxplots showing the variability of the private and non-private point estimates (where $\EMPointEstimator$ is the exponential mechanism point estimator). The whiskers of the boxplots indicate the 5th and 95th quantile of the empirical estimates to ensure consistency with the 90\% confidence intervals reported in Figure \ref{fig:income-widths}. The true medians from the population are indicated by cyan lines for each of the sub-populations. 

We find that from an inferential perspective the difference between the private and non-private estimates is small. For most of the estimates, the range of the boxplots overlap to a large extent and similar to Table \ref{tab:Census_results} the bias is small relative to the variability in the estimates. The only estimate for which we find noticeable bias is the private estimate for non-family households. The bias arises because of the large fraction of zeros among this sub-population in the original data. Since 51\% of the records in the original data report an income that is essentially zero (except for the small amount of noise that we introduce to make our data approximately continuous), the sample median will also be close to zero in many simulation runs. However, the CDP point estimate is based on the midpoint of the CDP confidence interval and the upper limit of this confidence interval will almost always be larger than the 51st quantile in the population, that is, it will almost always be larger than zero. Thus, the estimated CDP median will be biased. This highlights that the midpoint strategy can run into problems if the data is highly concentrated in certain areas of the distribution. However, such a scenario does not necessarily introduce substantial bias. This can be seen for the head of households that are 65 years or older. For this sub-population, the percentage of householders with zero reported income is about 81\%. As a result, the upper limit of the CDP confidence interval is still close to zero in most simulation runs and the point estimate remains almost unbiased.

\section{Conclusion}
Measuring the uncertainty in differentially private estimates is a challenging task, especially if the input data is a sample from a larger population. In this case, both sources of randomness---the sampling error as well as the error from the CDP algorithm---need to be taken into account. If the mechanism is data dependent, the two sources are no longer independent making it difficult to quantify the uncertainty in the final output. 

In this paper we addressed this challenge for the median, evaluating several strategies to obtain differentially private confidence intervals for this commonly used statistic. All the algorithms proposed produced valid non-parametric confidence intervals.
We also demonstrated that directly accounting for both sources of uncertainty simultaneously allowed us to give tighter confidence bounds than relying on naive approaches that account for the two components  sequentially. Our simulation results showed that an algorithm we called $\EM$ produced reliable and consistent confidence intervals which were less than twice the width of the non-private confidence intervals in a wide variety of parameter regimes. A pair of algorithms called $\CDF$ and $\BSCDF$ provide confidence intervals that are almost as tight, or slightly tighter, than $\EM$ in a variety of regimes. These algorithms are practically appealing since they release a wealth of additional information about the distribution $P$ without consuming additional privacy budget. 

The private confidence intervals in the application based on the 1940 Decennial Census were not substantially wider than the confidence intervals released by the non-private algorithm, illustrating that the extra uncertainty due to data protection can be small in practice. It should be noted that the total privacy budget always needs to be divided among all the characteristics of interest (type of household, metropolitan status and age in our application) so the accuracy will necessarily decrease if more statistics are to be released under the same privacy budget. 

We also found that the bias introduced by the ad-hoc strategy of using the midpoint of the confidence interval as an estimate for the median was limited for most estimates in our real data application. One strategy to further reduce this bias would be to use available information regarding the skewness of the data to come up with a better point estimate for the median. The fact that most of the algorithms provide additional information regarding the CDF of the data could be helpful for this endeavour as the information could be exploited in a post-processing step to model the distribution of the data. We leave this for future research.

We saw in our experiments on both simulated and real data that the empirical coverage rate of our private confidence intervals was often (sometimes substantially) higher than the nominal coverage rate. An interesting open question is whether this is inherent for non-parametric CDP confidence intervals for the median. Further, if this is unavoidable, then what distributional assumptions are required to narrow the gap between the empirical and nominal coverage rates? 

Finally, perhaps the strongest limitation of our paper is the reliance on the assumption that the sample is drawn using simple random sampling with replacement. Such a sampling design will never be used in the survey context in practice. Thus, the important next step will be to extend the methodology to allow for more complex designs. However, this raises many challenging problems. First, all algorithms assume that the sample is iid, which is typically not true in practice and it is not obvious which adjustments would be necessary to account for this. Second, many sampling designs are informative, meaning that the sampling design is data dependent, which has consequences on the privacy guarantees that are difficult to quantify. Third, the sampling weights that would be used to compute the median (or any other quantile) would influence the sensitivity of the statistic. Fourth, further adjustments such as calibration or dealing with non-response have additional impacts on the privacy guarantees. Fifth, computing confidence intervals for the median is challenging for many sampling designs even for the non-private case. Addressing all these aspects is well beyond the scope of this paper. However, each of these aspects would be an interesting and important area of future research with impact well beyond the median application considered in this paper.

\section*{Acknowledgments}

The work of Drechsler, Globus-Harris, Sarathy and Smith on this project was funded in part by US Census Bureau cooperative agreements CB16ADR0160001 and CB20ADR0160001.  
The work of  McMillan (while at BU) and Smith was also supported in part  by NSF award CCF-1763786 as well as a Sloan Foundation research award. Part of this work was done while McMillan was supported by a Fellowship from the Cybersecurity \& Privacy Institute at Northeastern University and NSF grant CCF-1750640. Globus-Harris' work at BU was supported by funding from the Hariri Institute for Computing, and was supported by the CIS PhD Graduate Fellowship at U Penn.
The opinions, findings, conclusions and recommendations
expressed herein are those of the authors and do not necessarily
reflect the views of the US Census Bureau or other funding sources.

Our work was prompted in part by discussions with sociologists John Logan and Brian Stults, in the context of their work on integrating data across time-varying tract boundaries \citep{LoganZSG21}.
We are also grateful for helpful conversations with and comments from (in no particular order) Rolando Rodriguez, Ryan Cummings, Thomas Steinke, Shurong Lin, Eric Kolaczyk and Salil Vadhan.

\bibliographystyle{plainnat}
\bibliography{main}

\clearpage

\appendix
\section*{Appendix}

This appendix contains further details and pseudocode for the four mechanisms evaluated in the main paper: \EM, \CDF, \noisyBS, and \BSCDF~(Sections \ref{online supplement:EM} to \ref{online supplement:BSCDF}). Note that code for these algorithms can be found at \url{https://github.com/anonymous-conf-medians/dp-medians}. This appendix also provides simulation results for some additional algorithms (Section \ref{online supplement: otheralgs}) that we did not explore further as they were strictly dominated by other algorithms in all the parameter settings considered in the main paper. 

\section{From Continuous Distributions to All Distributions}\label{convolutionisgood}

The algorithms and proofs in this paper focus on confidence intervals for the class of continuous distributions $\gooddist$. The relevant property of this class of distributions is that the distribution of the rank of the median is exactly given by \[\Pr(\rank{d}{\median(P)}=m)=\Pr(\Bin(n,1/2)=m).\]
This property can fail for distributions where the median itself has non-zero mass. However, we can use a simple transformation to extend our confidence intervals for continuous distributions to a function that is arbitrarily close to a confidence interval for the set of all distributions on $\mathbb{R}$, $\Delta(\mathbb{R})$. The transformation involves adding a small amount of Gaussian noise to the samples from $P$, in order to produce samples from a continuous distribution that are close to samples from $P$. A confidence interval algorithm for $\gooddist$ is then run on the resulting samples. 

We'll say a function $M:\mathcal{X}^n\to\intervals$ is a \emph{$(\beta, 1-\alpha)$-good confidence interval} for $Q \in \gooddist$ if with probability at least $1 - \alpha$, 
\[
\exists m\in M(X) \text{ s.t. } \Pr_{x\sim Q}(x\le m)\in[1/2-\beta, 1/2+\beta]
\]
where the probability is taken over the randomness of $M$ and $X \sim Q^n$.
Note that a $(0, 1-\alpha)$-good confidence interval is simply a $1-\alpha$-confidence interval. 

Let $\Phi_{\sigma}$ be the cumulative distribution function (CDF) of the Gaussian $\mathcal{N}(0,\sigma^2)$. Algorithm~\ref{transformation} describes the transformation of $M$. Note that $M'$ expands the confidence interval by $\Phi_{\sigma}^{-1}(1-\beta)$. We can make $\beta$ and $\Phi_{\sigma}^{-1}(1-\beta)$ both arbitrarily small by setting $\sigma^2$ to be arbitrarily small.

\begin{algorithm}[H]
  \KwData{$X\sim P^n$, where $P \in \Delta(\mathbb{R})$}
  \KwHyperparams{$(0, 1-\alpha)$-confidence interval algorithm $M$ for distributions in $\gooddist$, $\sigma^2 > 0$, $\beta\in[0,1/2]$}
  $X' = X+\mathcal{N}(0, \sigma^2I_n)$\\
  Let $a = \Phi_{\sigma}^{-1}(1-\beta)$, where $\Phi_{\sigma}$ is the CDF of $\mathcal{N}(0,\sigma^2)$ \\
  \Return [M(X')-a, M(X')+a]
  \caption{$M'$, a $(\beta, 1-\alpha)$-good confidence interval for $P$} \label{transformation}
\end{algorithm}

\begin{lemma} For all $M$, $\sigma^2$ and $\beta\in[0,1]$,
if $M$ is an $\alpha$-confidence interval for $\gooddist$ then $M'$ (as defined in Algorithm~\ref{transformation}) is a $(\beta, 1-\alpha)$-good confidence interval for $\Delta(\mathbb{R})$.
\end{lemma}

\begin{proof} Let $a=\Phi_{\sigma}^{-1}(1-\beta)$, and
let $P\in\Delta(\mathbb{R})$ and $Q$ be the distribution of the sum $z=x+y$ where $x\sim P$ and $y\sim \mathbb{N}(0,\sigma^2)$. Note first that $Q\in\gooddist$. Now,
let $M(Z)=[L,U]$ be the output of the confidence interval on $Q^n$. Notice that $\median(Q)\in M(Z)$ if and only if $\Pr_{z\sim Q}(z\le L)\le \frac{1}{2}$ and $\Pr_{z\sim Q}(z\ge U)\le \frac{1}{2}$. Suppose that $\median(Q)\in M(Z)$ (which happens with probability $\geq 1-\alpha$), then \[\Pr_{x\sim P}(x\le L-a) = \Pr(z-y\le L-a) \le \Pr(z\le L \text{ or } y\ge a)\le \frac{1}{2}+(1-\Phi_{\sigma}(a)) = \frac{1}{2}+\beta\] and similarly \[\Pr(x\ge U+a) = \Pr_{x\sim P}(z-y\ge U+a) \le \Pr(z\ge U \text{ or } y\le -a)\le \frac{1}{2}+\Phi_{\sigma}(-a) = \frac{1}{2}+\beta.\]
Therefore, there exists a pair $m$ and $\gamma$ 
such that $\Pr_{x\sim P}(x\le m)=\gamma$, $m\in[L-a, U+a]$ and $\gamma\in[\frac{1}{2}-\beta, \frac{1}{2}+\beta]$. 
\end{proof}

\section{Details: Confidence intervals based on exponential mechanism, \EM}\label{online supplement:EM}
\label{sec:exp-mech-details}

\begin{theorem}[Exponential Mechanism~\cite{McSherryT07}] 
\label{thm:exp-mech}
Given a space of datasets $\calX^n$ and an arbitrary range, $\range$,
let $u: \calX^n \times \range \rightarrow \reals$ be a utility function that maps dataset/output pairs to utility scores.
For a fixed dataset $d \in \calX^n$ and privacy parameter $\eps \in \posreals$, the \textit{exponential mechanism} outputs $x \in [\datalb, \dataub]$ with probability proportional to $\exp\left(\frac{\eps u(d, x)}{2 \Delta u}\right)$, where 
\[\Delta(u)=\max_{x \in \range}\max_{d, d' \text{neighbours}} |u(d, x) - u(d', x)|.\]
The exponential mechanism is $\eps$-DP.
\end{theorem}

Let the range of possible outputs be $\range = [\datalb, \dataub]$. The standard exponential mechanism to estimate the value $d_{(k)}, k \in [1, 2, \ldots, n]$ (described in~\citep{Smi11} for $k = n/2$) uses the following utility function.  
\[
u(d, x) = - \left \vert \rank{d}{x} - k \right\vert 
\]

This utility function captures how far $x$ is \textit{in rank} from $d_{(k)}$. However, this standard mechanism estimates $d_{(k)}$ poorly if the datapoints in $d$ are highly concentrated around this value.\footnote{See ~\citep{Alabi:2020} for an explanation of this case.} Below is a variant on the standard exponential mechanism designed to perform well even in this situation.

\begin{definition}[$\theta$-Widened Exponential Mechanism~\citep{Alabi:2020}]
\label{def:exp-mech-pe}
For a widening parameter $\theta>0$ and target rank $k \in [1, 2, \ldots, n]$, the \textit{$\theta$-widened exponential mechanism} uses the following utility function. 
\[u(d, x) = -\min\left\{\left|\rank{d}{a}-k\right|\;:\; |a-x|\le\theta\right\} \]
\end{definition}

The $\theta$-widened utility function can be implemented in different ways; Algorithm~\ref{alg:exp-mech-pe} offers one method of doing so. To sample efficiently from the distribution defined by the utility function, we implement a two-step strategy as shown in prior work~\citep{Alabi:2020,C}: First, we sample an interval, using the fact that sampling from the exponential mechanism is equivalent to choosing the value with maximum utility score after i.i.d. Gumbel-distributed noise has been added to the utility scores~\citep{ALT16}. Second, we sample an output uniformly at random from that interval. 

\begin{algorithm}[h]
  \KwData{$\datainput$}
  \KwPrivacyparams{$\epsinput$}
  \KwHyperparams{$\rankinput, \rangeinput, \graninput$}
  
  Clip $d$ to the range $[\lowerrange, \upperrange]$, setting values less than $\lowerrange$ or greater than $\upperrange$ to $\lowerrange$ and $\upperrange$ respectively.
    
  $n = |d|$
  
  Sort $d$ in increasing order

  \For{$i \in [1, k]$}{
    $d_i = \max(\datalb, d_i - \theta)$
   }
   
   \For{$i \in [k+1, n]$}{
    
    $d_i = \min(\dataub, d_i + \theta)$
  }
  
  Insert $\datalb$ and $\dataub$ into $d$ and set $n = n+2$
  
  Set $\textrm{maxNoisyScore} = -\infty$
  
  Set $\textrm{argMaxNoisyScore} = -1$

  \For{$i \in [2, n)$} {
    
    $\textrm{score} = \log(d_i - d_{i-1}) -\frac{\eps}{2} \cdot | i - k | $
    
    $N \sim \textrm{Gumbel}(0,1)$ 
    
    $\textrm{noisyScore} = \textrm{score} + N$
    
    \If {$\textrm{noisyScore} > \textrm{maxNoisyScore}$} {
        $\textrm{maxNoisyScore}= \textrm{noisyScore}$
        
        $\text{argMaxNoisyScore} = i$
        }
  }
  
  $\text{left} = d_{\text{argMaxNoisyScore}-1}$
  
  $\text{right} = d_{\text{argMaxNoisyScore}}$
  
  Sample $\tilde{m} \sim \text{Unif}\left[\text{left} , \text{right} \right]$
  
  \Return $\tilde{m}$

  \caption{\EMPointEstimator: $\theta$-Widened Exponential Mechanism for Quantile Estimation}  \label{alg:exp-mech-pe}
\end{algorithm}

\begin{lemma}
\label{lem:exp-mech-pe-privacy}
$\EMPointEstimator$ (Algorithm~\ref{alg:exp-mech-pe}) is $\eps$-DP.
\end{lemma}
\begin{proof}
Follows directly from Theorem~\ref{thm:exp-mech}.
\end{proof}

\begin{definition}[$(t,\theta, \beta)$-good] 
\label{def:t-theta-beta-goodness}
Let $A$ be a randomized mechanism that outputs a real-valued variable $m$.
For a fixed dataset $d \in \calX$, $\graninput$, and $\beta \in (0,1)$, $m$ is \textit{$(t,\theta, \beta)$-good} with respect to target rank $k \in [1, 2, \ldots, n]$ if there exists a datapoint $a \in d$ such that with probability at least $1-\beta$,
\[ |m-a| \leq \theta ~\text{and}~ \left|\rank{d}{a}- k \right| \leq t, \]
where the probability is over the randomness of $A$.
\end{definition}

\begin{lemma}
\label{lem:exp-mech-output-goodness}
Let $d \in \calX^n$ be a dataset and $k \in [1, 2, \ldots, n]$ be a target rank.
Let $A_k^{\eps}(d)$ be the $\theta$-widened exponential mechanism with privacy parameter $\eps \in \posreals$, widening parameter $\theta \in \reals$, and range parameter $\range \subset \reals$, and let us assume that $d_{(k)} \in \range$.  For $\beta \in (0,1)$, let $t = \ln\left((|\range|-2\theta)/(2\theta\beta)\right) / \eps$. Then, the output of $A_k^{\eps}(d)$ is $(t, \theta, \beta)$-good.
\end{lemma}

\begin{proof}
We will upper bound the probability density of outputs that are \emph{not} $(t, \theta)$-good with respect to target rank $k$, ie. outputs $m$ for which there does not exist a datapoint $a \in d$ such that $|m-a| \leq \theta$ and $|\rank{d}{a} - k| \leq t$.

To do so, recall that the $\theta$-widened exponential mechanism assigns utility scores to dataset/output pairs according to Definition~\ref{def:exp-mech-pe}.
For a given $t$, let us define \textit{good outputs} as those having a utility score $\geq -t$, which are assigned unnormalized probability density of at least 1, and \textit{bad outputs} as those having a utility score $< -t$, which are assigned unnormalized probability density of at most $\exp(- t \eps)$.  By definition of the $\theta$-widened utility function, the
good outputs must span an interval of size at least $2\theta$ and the bad outputs span an interval of size at most $|\range|-2\theta$. Therefore, we have that
\begin{align*}
\Pr_{A} &\left( \nexists~ a \in d : |A_k^{\eps}(d) - a| \leq \theta ~\text{and}~ |\rank{d}{a} - k| \leq t \right) \\
&\leq \Pr_{A} \left( \nexists~ x \in \range : |A_k^{\eps}(d) - x| \leq \theta ~\text{and}~ |\rank{d}{x} - k| \leq t \right) \\
&=\Pr_{A} \left(\forall ~ x \in \range, |A_k^{\eps}(d) - x| > \theta ~\text{or}~ |\rank{d}{x} - k| > t \right) \\
&= \Pr_{A} \left( ~\forall x \in [A_k^{\eps}(d)- \theta, A_k^{\eps}(d) + \theta] \text{ we have } |\rank{d}{x} - k| > t \right) \\
&= \Pr_{A} \left( [A_k^{\eps}(d) - \theta,  A_k^{\eps}(d) + \theta] \subseteq \text{bad outputs} \right) \\
&\leq \frac{\Pr_{A} \left( [A_k^{\eps}(d) - \theta,  A_k^{\eps}(d) + \theta] \subseteq \text{bad outputs} \right)}{P_{A} \left( [A_k^{\eps}(d) - \theta,  A_k^{\eps}(d) + \theta] \subseteq \text{good outputs}\right)} \\
&\leq \frac{(| \range | - 2\theta) \exp(-t \eps)}{2\theta } 
\end{align*}
Setting this probability to be within $\beta$, we can solve for $t$ as
\begin{align*}
t \geq \frac{1}{\eps}\ln\left(\frac{|\range|-2\theta}{2\theta \beta}\right)
\end{align*}
The resulting bound is tight by virtue of the worst-case example.
\end{proof}

Next, we will consider two ways in which we can use $\EMPointEstimator$ to create a confidence interval for the median. The first ($\EMNaive$) consists of taking a union bound over the probability that the non-private interval fails to capture the true median, and the probability that the private interval fails to capture the non-private interval. The second ($\EM$) is a more nuanced approach that accounts for the noise due to sampling and noise due to privacy together.

\subsection{Union bound confidence interval}

The following pseudocode describes $\EMNaive$ or $\EM$ (depending on the boolean hyperparameter $\Naive$). The sub-algorithm $\ComputeEMTargets$ will be described in the next subsection, as it is only called when $\Naive = 0$.

\begin{singlespace}
\begin{algorithm}[h]
  \KwData{$\datainput$}
  \KwPrivacyparams{$\epsinput$}
  \KwHyperparams{$\alphainput, \Naive \in \{0,1\}$, $\rangeinput, \graninput, \privconf \in (0, \alpha)$}
  
  $t = \frac{1}{\eps} \cdot \ln \left( \frac{|\range| - 2\theta}{\theta \cdot \beta_2} \right) $
  
  \If{$\Naive$}{
    $\beta_1 = \frac{\alpha - \beta_2}{1 - \beta_2/2}$
    
    $\PPLeps{\alpha} = \lfloor \PNPL{\beta_1} - t \rfloor$ \tcp{ $[d_{(\PNPL{\beta_1})}, d_{(\PNPU{\beta_1})}]$ is the nonprivate ($1-\beta_1$)-confidence interval for the median (see Lemma~\ref{nonprivCI}).}
    
    $\PPUeps{\alpha} = \lceil \PNPU{\beta_1} + t \rceil$
  }
  \Else{
   $\PPLeps{\alpha}, \PPUeps{\alpha} = \ComputeEMTargets(n, \eps, \alpha, \range, \theta)$
  }
  
  $\privNPL{\alpha}(d) = \EMPointEstimator(d, \eps/2, ( \PPLeps{\alpha}, \range, \theta))- \theta$
  
  $\privNPU{\alpha}(d) = \EMPointEstimator(d, \eps/2 (\PPUeps{\alpha}, \range, \theta)) + \theta$
    
\Return $[\privNPL{\alpha}(d), \privNPU{\alpha}(d)]$
  \caption{$\EM(\Naive)$:  
  $\eps$-DP Algorithm}  \label{alg:exp-mech-ci}
\end{algorithm}
\end{singlespace}

First, we show that both $\EMNaive$ and $\EM$ are $\eps$-DP.
\begin{lemma}
\label{lem:wem-ci-privacy}
$\EM(\Naive)$ (Algorithm~\ref{alg:exp-mech-ci}) is $\eps$-DP.
\end{lemma}
\begin{proof}
The computations of $t, \beta_1, \PPLeps{\alpha}$, and $\PPUeps{\alpha}$ do not depend on the dataset $d$. Therefore, when analyzing the privacy loss, we simply need to consider the two calls the algorithm makes to \EMPointEstimator, each with privacy parameter $\eps/2$. By Lemma~\ref{lem:exp-mech-pe-privacy}, each of these algorithms is $\eps/2$-DP, so by composition, $\EM$
is $\eps$-DP.
\end{proof}

Then, we show that $\EMNaive$ produces a valid confidence interval.

\begin{lemma}\label{lem:exp-mech-ci-naive-validity}
Let dataset $d$ be drawn i.i.d. from a distribution $P \in \gooddist$ with population median $\median(P)$. Let $\eps > 0, \range \in \reals, \theta > 0$, and let us assume that $\median(P) \in \range$. For $\alpha \in (0,1)$ and $\privconf \in (0, \alpha)$, let $[\privNPL{\alpha}, \privNPU{\alpha}]$ be the output of $\EMNaive(d, \eps, (\alpha, \Naive =1, \range, \theta, \privconf))$.
Then, with probability at least $1-\alpha$,
\[
\median(P) \in [\privNPL{\alpha}, \privNPU{\alpha}],
\]
where the probability is over the randomness in both the dataset $d$ and the mechanism \EMNaive.
\end{lemma}
\begin{proof}
First, letting $\NPL{\nonprivconf} = d_{(\PNPL{\nonprivconf})}$ and $\NPU{\nonprivconf} = d_{( \PNPU{\nonprivconf})}$, for any $\nonprivconf \in (0,1)$ we have by Lemma~\ref{nonprivCI} that
\begin{align}
\label{eq:nonpriv-prob}
\Pr_{d} \left( \median(P) < \NPL{\nonprivconf} \right) =
\Pr_{d} \left( \rank{d}{\median(P)} < \PNPL{\nonprivconf} \right) \leq \nonprivconf/2.
\end{align}
Then, 
for $\PPLeps{\alpha} = \PNPL{\nonprivconf} - t$, and $\PPUeps{\alpha} = \PNPU{\nonprivconf} + t$, 
 $\EMNaive$ (Algorithm~\ref{alg:exp-mech-ci}) 
 outputs the interval $[\privNPL{\alpha}, \privNPU{\alpha}]$, where $\privNPL{\alpha} = A_{\PPLeps{\alpha}}^{\eps/2}(d) - \theta$ and 
$\privNPU{\alpha} = A_{\PPUeps{\alpha}}^{\eps/2}(d) + \theta$.
By Lemma~\ref{lem:exp-mech-output-goodness}, the output of $A_{\PPLeps{\alpha}}^{\eps/2}(d)$ is $(t, \theta, \privconf/2)$-good with respect to rank $\PPLeps{\alpha}$.
By Definition~\ref{def:t-theta-beta-goodness}, this means that
\begin{align}
\label{eq:privprob}
    \Pr_{A} \left( \privNPL{\alpha} > \NPL{\nonprivconf} \right) = 
    \Pr_{A}\Bigl( \rank{d}{A_{\PPLeps{\alpha}}^{\eps/2}(d) - \theta} > \PNPL{\beta_1}  \Bigr) \leq \beta_2/2 
\end{align}
Putting these together, we consider the lower endpoint of the interval $[\privNPL{\alpha}, \privNPU{\alpha}]$. We can upper bound the failure probability as follows.
\begin{align*}
    \Pr_{A, d} \left( \median(P)  < \privNPL{\alpha} \right) 
    &= \Pr_{A} \left( \median(P) < \privNPL{\alpha} \mid \median(P) < \NPL{\beta_1} \right) \cdot \Pr_{d} \left( \median(P) < \NPL{\nonprivconf} \right) \\
    &~~~~~~~~~~~~ + \Pr_{A} \left( \median(P) < \privNPL{\alpha} \mid \median(P) \geq \NPL{\nonprivconf} \right) \cdot \Pr_{d} \left( \median(P) \geq \NPL{\nonprivconf} \right) \\
    &\leq 1 \cdot \Pr_{d} \left( \median(P) < \NPL{\nonprivconf} \right) + \Pr_{A} \left( \privNPL{\alpha} \geq \NPL{\nonprivconf} \right) \cdot \Pr_{d} \left( \median(P) \geq \NPL{\nonprivconf} \right)  \\
    &\leq \nonprivconf/2 + (\privconf/2) \cdot (1 - \nonprivconf/2 )  \\ 
    &= \alpha/2
\end{align*}
where the second inequality follows from (\ref{eq:nonpriv-prob}) and (\ref{eq:privprob}), and the final equality follows from the definition of $\nonprivconf$ in Algorithm~\ref{alg:exp-mech-ci}.
A similar inequality holds for the upper endpoint of the interval, so a union bound gives the desired result.
\end{proof}

\subsection{Tighter Confidence Interval}

Next, we consider the more nuanced approach.
Let $P \in \gooddist$ be a population distribution function, where $\median = \median(P)$ is the population median. For a dataset $d = (d_1, \ldots, d_n)$ where $d_i$ is sampled i.i.d. from distribution $P$, let $\rank{d}{a}$ denote the rank of real value $a$ within dataset $d$. 
Let $A_k^{\eps}(d)$ be the output of the $\theta$-widened exponential mechanism on dataset $d$ that estimates the value at rank $k$. For a given $k_L, k_U$, we would like to control the probability that the interval $[ A(d, k_L) - \theta, A(d, k_U) + \theta]$ fails to contain the true median $\median$. In particular, for $\alpha \in (0, 1)$, we would like to find the target ranks $k_L$ and $k_U$ closest to $n/2$ such that
\begin{align*}
    \Pr_{A, d} \left( A^{\eps/2}_{k_L}(d) - \theta > \median \right) \leq \alpha/2 \\
    \Pr_{A, d} \left(A^{\eps/2}_{k_U}(d) + \theta < \median \right) \leq \alpha/2
\end{align*}

\begin{algorithm}[ht]
  \KwInput{$\ninput, \epsinput, \alphainput, \rangeinput, \graninput$}
  
  \For{$k_L \in \mathbb{N}, 1 \leq k_L \leq n/2 $}{
  
  $p_{k_L}= C_{\text{Bin}}(k_L -1) + \sum_{m = k_L}^{n} C'_{\text{Bin}}(m) \cdot \frac{(|\range| - 2\theta) \exp(- (m-k_L) \cdot \eps / 2 )}{2\theta}$
  }
  $\PPLeps{\alpha} = \max_{k_L \in \mathbb{N}, 1 \leq k_L < \lceil n/2 \rceil} \{k_L : p_{k_L} \leq \alpha/2 \}$
  
  \For{$k_U \in \mathbb{N}, n/2 \leq k_U < n$}{
  
  $p_{k_U} = \sum_{m = 1}^{k_U} C'_{\text{Bin}}(m) \cdot \frac{(|\range| - 2\theta) \exp(- (k_U-m) \cdot \eps / 2 )}{2\theta} + (1 - C_{\text{Bin}}(k_U +1))$
  }
  $\PPUeps{\alpha} = \min_{k_U \in \mathbb{N}, \lceil n/2 \rceil \leq k_U < n} \{j: p_{k_U} \leq \alpha/2 \}$
  
  \Return $\PPLeps{\alpha}, \PPUeps{\alpha}$
  
  \caption{\texttt{Compute\EM Targets}}  \label{alg:exp-mech-compute-targets}
\end{algorithm}

In Algorithm~\ref{alg:exp-mech-compute-targets} ($\ComputeEMTargets$), we find these target ranks by first computing the probabilities above for all possible $k_L$ and $k_U$'s, and then by numerically searching for the target ranks closest to $n/2$ such that the probabilities above are both within $\alpha/2$.\footnote{This search can be implemented more efficiently by noting that $k_L$ is greater than or equal to $\lfloor \PNPL{\beta_1} - t \rfloor$ as defined in Algorithm~\ref{alg:exp-mech-ci}, and similarly $k_U$ is less than or equal to $\lceil \PNPU{\beta_1} + t \rceil$.} The following lemma characterizes these probabilities.

\begin{lemma}
\label{lem:exp-mech-prob-given-rank}
Let $P \in \gooddist$ be a population distribution function, where $\median = \median(P)$ is the population median. For a dataset $d = (d_1, \ldots, d_n)$ where $d_i$ is sampled iid. from distribution $P$, let $\rank{d}{a}$ denote the rank of real value $a$ within dataset $d$. 
Let $k_L, k_U \in [1, 2, \ldots, n]$ be target ranks, and let $A_{k_L}^{\eps}(d)$ and $A_{k_U}^{\eps}(d)$ be $\theta$-widened exponential mechanisms on dataset $d$ that estimate the value at rank $k_L$ and $k_U$, respectively. Let $\CDFbin$ and $\PDFbin$ be the CDF and PDF of the binomial random variable \Bin(n, 1/2). Then, 
\begin{align*}
   \Pr_{A, d} \left( A^{\eps/2}_{k_L}(d) - \theta > \median \right)
   &\leq \CDFbin(k_L) + \sum_{m = k_L + 1}^{m = n}  
     \PDFbin(m) \cdot \frac{(|\range|-2\theta) \exp \left( -(m - k_L) \eps /2 \right) }{2\theta}  \\
   \Pr_{A, d} \left( A^{\eps/2}_{k_U}(d) + \theta < \median \right)
    &\leq (1-\CDFbin(k_U + 1)) + \sum_{m = 1}^{m = k_U}  
     \PDFbin(m) \cdot \frac{(|\range|-2\theta) \exp \left( -(k_U - m) \eps /2 \right) }{2\theta} 
\end{align*}
\end{lemma}
\begin{proof}
For simplicity, we consider just the first statement pertaining to the lower endpoint of the interval. We can split up the probability into two cases: first, when $\rank{d}{\median} < k_L$, and second, when $\rank{d}{\median} \geq k_L$.
\begin{align*}
    \Pr_{A, d} \left( A^{\eps/2}_{k_L}(d) - \theta > \median \right) 
    &= \sum_{m = 1}^{m = k_L-1}  
    \Pr_{A} \left( A^{\eps/2}_{k_L}(d) - \theta > \median \mid \rank{d}{\median} = m \right) \cdot \Pr_{d}(\rank{d}{\median} = m) \\
    &+ \sum_{m = k_L}^{m = n}  
    \Pr_{A} \left( A^{\eps/2}_{k_L}(d) - \theta > \median \mid \rank{d}{\median} = m \right) \cdot 
    \Pr_{d}(\rank{d}{\median} = m)
\end{align*}
For the first case, where $\rank{d}{\median} < k_L$, we simply upper bound the first probability in the summation by $1$ and note that the random variable $\indicator_{\rank{d}{\median} = m}$ follows a binomial distribution.
\begin{align*}
    \sum_{m = 1}^{m = k_L-1}  
    &\Pr_{A} \left( A^{\eps/2}_{k_L}(d) - \theta > \median \mid \rank{d}{\median} = m \right) \cdot \Pr_{d}(\rank{d}{\median} = m) 
     \leq \CDFbin(k_L-1)
\end{align*}
In the second case, we first observe that the probability of any real value $a$ being greater than $\median$ is monotonically increasing in $a$, which gives
\begin{align*}
    \Pr_{A} \left( A^{\eps/2}_{k_L}(d) - \theta > \median \mid \rank{d}{\median} = m \right) 
    &\leq \Pr_{A} \left( A^{\eps/2}_{k_L}(d) > \median \mid \rank{d}{\median} = m \right)  \\
    &= \Pr_{A} \left( \rank{d}{A^{\eps/2}_{k_L}(d)} > m  \right) \\
    &\leq \Pr_{A} \left( \lvert \rank{d}{A^{\eps/2}_{k_L}(d)} - k_L \rvert > m - k_L \right) \\
    &\leq \frac{(|\range|-2\theta) \exp \left( -(m - k_L) \eps/2 \right) }{2\theta},
\end{align*}
where the last inequality follows from Lemma~\ref{lem:exp-mech-output-goodness}. Therefore, we have that
\begin{align*}
   \Pr_{A, d} \left( A^{\eps/2}_{k_L}(d) - \theta > \median \right)
   &\leq \CDFbin(k_L-1) + \sum_{m = k_L}^{m = n}  
     \PDFbin(m) \cdot \frac{(|\range|-2\theta) \exp \left( -(m - k_L) \eps/2 \right) }{2\theta}
\end{align*}
A similar result holds for $\Pr_{A, d} \left( A^{\eps/2}_{k_U}(d)  + \theta < \median \right)$. 
\end{proof}

Validity of the $\EM$ confidence interval then follows directly from the selection of the target ranks.

\begin{lemma}
\label{lem:exp-mech-ci-validity}
Let dataset $d$ be drawn i.i.d. from a distribution $P \in \gooddist$ with population median $\median(P)$. For (hyper)parameters $\eps > 0, \range \in \reals, \theta > 0$, and $\alpha \in (0,1)$, let $[\privNPL{\alpha}, \privNPU{\alpha}]$ be the output of $\EM(d, \eps, (\alpha, \Naive = 0, \range, \theta, \cdot))$.
If $\median(P) \in \range$, then with probability at least $1-\alpha$,
\[
\median(P) \in [\privNPL{\alpha}, \privNPU{\alpha}],
\]
where the probability is over the randomness in both the dataset $d$ and the mechanism \EM.
\end{lemma}
\begin{proof}
\texttt{Compute\EM Targets} (Algorithm~\ref{alg:exp-mech-compute-targets}, relying on Lemma~\ref{lem:exp-mech-prob-given-rank}) returns target ranks $\PPLeps{\alpha}$ and $\PPUeps{\alpha}$ such that $Pr_{A, d} \left( A^{\eps/2}_{\PPLeps{\alpha}}(d) - \theta > \median \right) \leq \alpha/2$ and $\Pr_{A, d} \left(A^{\eps/2}_{\PPUeps{\alpha}}(d) + \theta < \median \right) \leq \alpha/2$.
$\EM$ (Algorithm~\ref{alg:exp-mech-ci}) then sets $\privNPL{\alpha}(d) = A_{\PPL{\alpha}}^{\eps/2}(d) - \theta$ and $\privNPU{\alpha}(d) = A_{\PPU{\alpha}}^{\eps/2}(d) + \theta$. The result follows from a union bound.
\end{proof}

\section{Details: Confidence intervals based on noisy binary search, \noisyBS}\label{online supplement:BS}
\label{sec:bin-search-details}
\begin{algorithm}[h!]
  \KwData{$\datainput$}
  \KwPrivacyparams{$\rhoinput$}
  \KwHyperparams{$\alphainput$, $\rangeinput, \graninput, \betasplit$, \texttt{LB}, \texttt{UB} $\in (0,1)$}
  $\nonprivconf = \betasplit \alpha$
  
  $\privconf = \frac{\alpha - \nonprivconf}{1 - \nonprivconf/2}$
  
  $n = |d|$
  
  $m = \log( (\dataub - \datalb)/\theta )\;\;\;$ \tcp{number of steps required to get to desired granularity} 
  
  $\rhoperstep = \rho / (2m) $
  
  $\betaperstep = \privconf / (2m)$

  $t^{\rhoperstep}_{\betaperstep} = \sqrt{\frac{\log(1/\betaperstep)}{\rhoperstep n}}$ 
  
  $q_L = \min\{\texttt{LB}, \PNPL{\nonprivconf}/n -t^{\rhoperstep}_{\betaperstep}\}$
  
  $q_U = \max\{\texttt{UB}, \PNPU{\nonprivconf}/n +t^{\rhoperstep}_{\betaperstep} \}$
  
  \texttt{noisy-counts-lower} = \texttt{GetNoisyCounts}($d, \rho/2, (n, \privconf/2, q_L, q_L, q_U, \emptyset, \range, \rhoperstep, \betaperstep)$)
  
  \texttt{noisy-counts-upper} = \texttt{GetNoisyCounts}($d, \rho/2, (n, \privconf/2, q_U, q_L, q_U, \texttt{noisy-counts-lower}, \range, \rhoperstep, \betaperstep)$)
  
  \texttt{noisy-counts} =  \texttt{noisy-counts-lower}$\;\cup\;$\texttt{noisy-counts-upper} 

  \Return \texttt{PostProcessUnion}(\texttt{noisy-counts}, $n$, $\PNPL{\nonprivconf}, \PNPU{\nonprivconf}$, $\privconf$)
  
  \caption{$\noisyBS$:  
  $\rho$-CDP Algorithm}  \label{alg:bs}
\end{algorithm}

\begin{algorithm}[ht!]
    \KwData{$\datainput$}
    \KwPrivacyparams{$\rhoinput$}
    \KwHyperparams{$\ninput, \privconf \in (0,1), \targetquantile, q_L, q_U \in (0,1)$, \texttt{prev-queries}, $\rangeinput, \rhoperstep, \betaperstep$}

    \tcp{\texttt{prev-queries} = $\{(x, r_x, \sigma_x)\}$ is a collection of noisy measurements where $x\in[r_l, r_u]$ and $r_x = \rank{d}{x}+\mathcal{N}(0,\sigma_x^2)$}

     \texttt{lower} = $\datalb$, \texttt{upper} = $\dataub$ \tcp{The initial search space is the entire range.}

     $\rhoinit = \rhoperstep/10$, $\betainit = \betaperstep/10$ \tcp{Budget for initial measurement at every query point; can be arbitrarily small.}
     
    $t = 0$ \tcp{Counter for number of query points.}
          
     $\rhoused = 0$ \tcp{Counter for used privacy budget.}

     \While{$\rhoused + \rhoinit \leq \rho$}{
        
          $x_t = (\texttt{lower} + \texttt{upper})/2$ \tcp{Query point}

          \texttt{est-good-enough} = \texttt{False}\\
          
          \If {\rm there exists $r_{x_t}$ and $\sigma_{x_t}$ such that $(x_t, r_{x_t}, \sigma_{x_t})\in$\texttt{prev-queries}}
          {
              $\texttt{avg-noisy-count}_t = r_{x_t}$,
              $\texttt{avg-var}_t = \sigma_{x_t}$
              
              \texttt{est-good-enough} = \texttt{True}
          }

        \texttt{numMeasurements} = 0, $\rhousedthisstep = 0$, $\betausedthisstep = 0$
              
        \While{\rm \texttt{est-good-enough} = \texttt{False} and $\rhousedthisstep + \rhoinit \leq \rhoperstep$ and $\rhoused + \rhousedthisstep + \rhoinit \leq \rho$}{

                \texttt{numMeasurements} = \texttt{numMeasurements}+1, $\rhousedthisstep = \rhousedthisstep + \rhoinit$, $\betausedthisstep = \betausedthisstep + \betainit$
                
                $\texttt{noisy-count}_{\texttt{numMeasurements}}\sim \calN(\rank{d}{x_t}, 1/2  \rhoinit)$

                $\texttt{var}_{\texttt{numMeasurements}}=1/2\rhoinit$

                $\texttt{avg-noisy-count}_t = \sum_{k=1}^{\texttt{numMeasurements}} \texttt{noisy-count}_k / \texttt{numMeasurements}$
                
               $\texttt{avg-var}_t = (\sum_{k=1}^{\texttt{numMeasurements}}\texttt{var}_k)/\texttt{numMeasurements}$
                
                $K = \sqrt{\texttt{avg-var}} \cdot \Phi^{-1}(1-\betausedthisstep)$ \tcp{$\Phi$ is the standard normal distribution function}
                
                \If{\rm ($\texttt{avg-noisy-count}_t - K > q_U\cdot n$ or $\texttt{avg-noisy-count}_t + K < q_U\cdot n$) and ($\texttt{avg-noisy-count}_t - K > q_L\cdot n$ or $\texttt{avg-noisy-count}_t + K < q_L\cdot n$) }{
                    \texttt{est-good-enough} = \texttt{True}
                }
              }
          
          \If {\rm $\texttt{avg-noisy-count}_t  < n \cdot \targetquantile$} {
            \texttt{lower}$ = x_t$
            
          } 
          \Else{
            \texttt{upper} $= x_t$
          }
          $\rhoused = \rhoused + \rhousedthisstep$  
          
          $t = t + 1$
          
     }
     \Return $(x_1, \texttt{avg-noisy-count}_1, \texttt{avg-var}_1) , (x_2, \texttt{avg-noisy-count}_2, \texttt{avg-var}_2), \cdots$
    
  \caption{$\texttt{GetNoisyCounts}$: 
  $\rho$-CDP Algorithm}  \label{alg:ada-bs}
\end{algorithm}

\begin{algorithm}[h!]
\KwInput{$(x_1, \texttt{ns}_1, \texttt{var}_1) , (x_2, \texttt{ns}_2, \texttt{var}_2), \cdots, (x_T, \texttt{ns}_T, \texttt{var}_T), \ninput, q_L, q_U, \privconf \in (0,1)$} 

\For{$t\in[T]$}{
$R_t= \sqrt{\texttt{var}_t} \Phi^{-1}(1 - 2T/\privconf)$ 
\tcp{$\Phi$ is the standard normal distribution function}
$L_t=\texttt{ns}_t+R_t$\\
$U_t=\texttt{ns}_t-R_t$}
$l = \max\{x_t\;|\; \forall t'<t, L_t < q_L\}$\\
$u = \min\{x_t\;|\; \forall t'>t, U_t> q_U\}$\\

\Return $[l,u]$
\caption{\texttt{PostProcessUnion}}
\label{alg:consBS}
\end{algorithm}

In this section we provide the algorithmic details and validity and privacy proofs for $\noisyBS$.

Given a dataset $d \in \calX^n$, and target quantile $\targetquantile \in (0,1)$, an initial range $\range$ and granularity $\theta$, $\noisyBS$ (outlined in Algorithm~\ref{alg:bs}) consists of two steps. In the first step,
the mechanism $\texttt{GetNoisyCounts}$ uses noisy measurements of the empirical CDF to search for $d_{n(\targetquantile)}$ using binary search. This noisy binary search step is designed so that with high probability it moves in the right direction at each step, however there is some probability of making a wrong move, hence we need to perform a post-processing step that takes the noisy measurements as input and returns a valid confidence interval.

Pseudo-code for the first step, which we will call \texttt{GetNoisyCounts} is given in Algorithm~\ref{alg:ada-bs}. Let us focus on finding the lower limit of the confidence interval. Given a target quantile $\targetquantile$, this algorithm iterates reduces the search domain by querying the rank of the mid-point $x_t$ of the range. If it is confident that the mid-point is to left of the target quantile then it cuts the domain in half and only keep the right half (similarly if it is confident that the mid-point is to the right, it keep the left half of the range). It continues this process until the entire privacy budget is consumed. 

At each iteration we use a portion of the privacy budget $\rho$ to release the noisy rank of the query point $x_t$. The total privacy budget consumed by the algorithm is the sum of the privacy budget consumed by each step (Lemma~\ref{composition}). One option for allocating the privacy budget is to decide in advance the number of iterations and divide the privacy budget by the number of iterations to obtain a \emph{per step} privacy budget. However, we can actually improve on this approach by noticing that if $|\rank{d}{x_t}-\targetquantile \cdot n|$ is large then we can tolerate a lot of noise in our estimate of $\rank{d}{x_t}$ and still determine with high confidence whether $\rank{d}{x_t}> \targetquantile \cdot n$ or $\rank{d}{x_t}<\targetquantile \cdot n$. Thus, we may be able to only allocate a very small amount of privacy budget to some steps. For a given query point $x_t$ we do not know a priori how large $|\rank{d}{x_t}-\targetquantile \cdot n|$ is, and hence how much noise the query can handle. Thus, we start by adding a large amount of noise (using only a small amount of the privacy budget $\rhoinit$) to $\rank{d}{x_t}$. If this noisy estimate is far enough from $\targetquantile$ that we can confidently determine which direction to continue with the binary search, then we move and this step has only consumed $\rhoinit$ privacy budget. Otherwise, we take another noisy measurement of $\rank{d}{x_t}$ and average the two together. This produces a less noisy estimate, and consumes $2\rhoinit$. We continue in this way until either the variance of the estimate is low enough that we can confidently move, or this step has consumed the maximum amount of privacy budget per step $\rhoperstep$, and we move in the more likely direction. While we search for the left and right hand limit of the confidence interval separately, in many settings the early query points of the binary search will be same for both. Thus, we can improve accuracy by not repeating these noisy queries. This is why we pass \texttt{prev-queries} into \texttt{GetNoisyCounts}.

The next step is processing the noisy counts to obtain a valid confidence interval. Pseudo-code is given in Algorithm~\ref{alg:consBS} and a validity proof is given in Lemma~\ref{lem:noisy-bs-validity}. This step does not consume additional privacy budget since it is simply post-processing on top of the $\rho$-CDP output of \texttt{GetNoisyCounts}.

\begin{lemma} 
\label{lem:noisy-bs-privacy}
Mechanism $\noisyBS$ (Algorithm~\ref{alg:bs}) is $\rho$-CDP.
\end{lemma}
\begin{proof}
The lemma follows immediately from Lemma~\ref{composition} in the main text. The privacy budget $\rho$ is divided between the two calls to $\texttt{GetNoisyCounts}$, which each use privacy budget $\rho/2$. By \citep[Proposition 1.6]{BunS16}, each new noisy measurement $\texttt{noisy-count}_{\texttt{numMeasurements}}$ is $\rhoinit$-CDP, so $\rho_t$ and thus $\rho_{used}$ accurately capture the privacy budget consumed per step, and in total at any point during the algorithms run. 
\end{proof}

\begin{lemma}
\label{lem:noisy-bs-validity}
Let dataset $d$ be drawn i.i.d. from a distribution $P \in \gooddist$ with population median $\median(P)$. Given (hyper)parameters $\rangeinput, \graninput, \betasplitinput$, failure rate $\alphainput$ and privacy parameter $\rhoinput$, let $[\privNPL{\alpha}, \privNPU{\alpha}] = \noisyBS(d, \rho, (\alpha, \range, \theta, \betasplit, 0.5, 0.5))$. If $\median(P) \in \range$, then with probability at least $1-\alpha$,
\[
    \median(P) \in [\privNPL{\alpha}, \privNPU{\alpha}],
\]
where the probability is over the randomness in both the dataset $d$ and the mechanism $\noisyBS$.
\end{lemma}
\begin{proof}
While $\texttt{GetNoisyCounts}$ is designed to ensure the final output is close to the right quantile, the validity of the confidence interval really comes from the post-processing function \texttt{PostProcessUnion}. We have that 
\begin{align*}
    \Pr_{A, d}(\median(P)< \privNPL{\alpha}) 
    &\le \Pr_{A}( \median(P)< \privNPL{\alpha} \mid \rank{d}{\median(P)} < \PNPL{\nonprivconf}) \cdot \Pr_{d}(\rank{d}{\median(P)} < \PNPL{\nonprivconf}) \\
    &~~~~~~~~\;\;\;\;\;\;\;\;\;\;\; + 
    \Pr_{A}( \median(P)< \privNPL{\alpha} \mid \rank{d}{\median(P)} \geq \PNPL{\nonprivconf}) \cdot \Pr_{d}(\rank{d}{\median(P)} \geq \PNPL{\nonprivconf})
    \\
    &\leq \Pr_{d}(\rank{d}{\median(P)}<\PNPL{\nonprivconf})+\Pr_{A}(\median(P)< \privNPL{\alpha}\;|\;\rank{d}{\median(P)}\ge\PNPL{\nonprivconf}) \cdot \Pr_{d} (\rank{d}{\median(P)} \geq \PNPL{\nonprivconf}) \\
    &\le \nonprivconf/2+\Pr_{A}(\PNPL{\nonprivconf}<\rank{d}{\privNPL{\alpha}}) \cdot (1-\nonprivconf/2),
\end{align*}
where the subscript $A$ denotes that the probability is over the randomness of the mechanism $\noisyBS$.
Now, let $\texttt{noisy-counts}=(x_1, \texttt{ns}_1, \texttt{var}_1) , (x_2, \texttt{ns}_2, \texttt{var}_2), \cdots, (x_T, \texttt{ns}_T, \texttt{var}_T)$ be the concatenated outputs of the two runs of \texttt{GetNoisyCounts} in \texttt{NoisyBinSearch}; these are inputs to \texttt{PostProcessUnion}. We have the guarantee that for all $t \in T$, $\texttt{ns}_t=\rank{d}{x_t}+\mathcal{N}(0,\texttt{var}_t)$, and therefore with probability $1-\privconf/2$, for all $x\in\{x_1, \cdots, x_T\}$, $|\texttt{ns}_t-\rank{d}{x_t}|\le R_t$. Now, if $\PNPL{\nonprivconf}<\rank{d}{\privNPL{\alpha}}$ implies that there exists $x_t$ such that $x_t<\median(P)$ but  $\texttt{ns}_t\ge\PNPL{\nonprivconf}+R_t$. But this would imply that $|\texttt{ns}_t-\rank{d}{x_t}|\ge R_t$. Therefore, $\Pr(\median(P)< \privNPL{\alpha})\le\nonprivconf/2+\privconf/2 (1-\nonprivconf/2) =\alpha/2$. Similarly we can argue that $\Pr(\median(P)> \privNPU{\alpha})\le\nonprivconf/2+\privconf/2 (1-\nonprivconf/2) =\alpha/2$, so we are done.
\end{proof}

\section{Details: Confidence intervals based on CDF estimator \newline \CDF}\label{online supplement:CDF}
\label{sec:cdf-details}
Instead of generating a series of count queries that vary asymmetrically across the data set as in the binary search process, we could generate the entire empirical CDF of the data and use a similar approach to the binary search algorithm to generate confidence intervals. In our setting, where the number of queries are limited by the privacy budget, if the empirical CDF is of separate interest to a researcher, this is a particularly compelling method. There are many methods for generating a DP CDF \citep{Diakonikolas:2015, Brunel:2020}. We focus on a tree-based mechanism introduced in \citet{Li:2010,Dwork:2010}, and \citet{ Chan:2011} and refined in \citet{Honaker:2015}. We rely on \citet{Honaker:2015}'s algorithm.

\subsection{A Tree-Based Approach to Differentially Private CDFs}
Note that a simple way to estimate the CDF would be to create a differentially private histogram with a set bin size
and sum the bins to the left of a point of interest to generate an estimate of the CDF at that point. However, this means summing multiple noisy counts together, so the accuracy will diminish the more bins that you have to sum together. To avoid summing too many points together, one might instead use a tree-based approach, which uses a tree of multiple histograms that have multiple levels of granularity: we denote this as $\dptree;\dptreespace{\depth}$, a binary tree of counts in $\mathbb{N}$ with $\depth$ levels, $L_1, \ldots, L_{\depth}$, as depicted in Figure ~\ref{fig:cdf-tree} and described in detail in  Algorithm~\ref{alg:cdf-basic}. Note that to get a CDF estimate at point $\lowerrange + 2\theta$ one need only look at $t_{10}$, and to get the estimate at $\lowerrange+3\theta$, one need only sum $t_{10}$ and $t_{110}$.

\begin{singlespace}
\begin{algorithm}[H]
\caption{\texttt{DPTree}: $\rho$-CDP Histogram Tree Algorithm} 
\label{alg:cdf-basic}
\KwData{$\datainput$}
\KwPrivacyparams{$\rhoinput$}
\KwHyperparams{$\rangeinput$, $\depth \in \naturals$}

$n = |d|$

Let $\dptree \in \dptreespace{\depth}$ be a binary tree with $\depth$ levels, $L_1, \ldots, L_{\depth}$

\For{$j \in $ [\depth]} {
    
    Let $\text{bin}_1, \ldots, \text{bin}_{2^j}$ be $2^j$ equally-sized partitions of the range $\range$.
    
    Generate histogram $\texttt{hist} = \{\# i : d_i \in \text{bin}_b, ~~ 1 \leq b \leq 2^j\} \in \naturals^{2^j}$
    
    Add noise sampled from $\mathcal{N}(0, 2 \depth /\rho)$ to each element of $\texttt{hist}$.
    
    Set $L_j = \texttt{hist}$
  }
\Return $\dptree$
\end{algorithm}
\end{singlespace}
\begin{figure}
\begin{centering}
\tikzset{every picture/.style={line width=0.75pt}} 

\begin{tikzpicture}[x=0.75pt,y=0.75pt,yscale=-1,xscale=1]

\draw   (323.71,26.03) .. controls (318.61,26.02) and (314.48,21.88) .. (314.49,16.77) .. controls (314.5,11.67) and (318.64,7.54) .. (323.75,7.55) .. controls (328.85,7.56) and (332.98,11.7) .. (332.97,16.81) .. controls (332.96,21.91) and (328.81,26.04) .. (323.71,26.03) -- cycle ;
\draw   (243.71,86.03) .. controls (238.61,86.02) and (234.48,81.88) .. (234.49,76.77) .. controls (234.5,71.67) and (238.64,67.54) .. (243.75,67.55) .. controls (248.85,67.56) and (252.98,71.7) .. (252.97,76.81) .. controls (252.96,81.91) and (248.81,86.04) .. (243.71,86.03) -- cycle ;
\draw   (403.71,85.03) .. controls (398.61,85.02) and (394.48,80.88) .. (394.49,75.77) .. controls (394.5,70.67) and (398.64,66.54) .. (403.75,66.55) .. controls (408.85,66.56) and (412.98,70.7) .. (412.97,75.81) .. controls (412.96,80.91) and (408.81,85.04) .. (403.71,85.03) -- cycle ;
\draw   (364.71,146.03) .. controls (359.61,146.02) and (355.48,141.88) .. (355.49,136.77) .. controls (355.5,131.67) and (359.64,127.54) .. (364.75,127.55) .. controls (369.85,127.56) and (373.98,131.7) .. (373.97,136.81) .. controls (373.96,141.91) and (369.81,146.04) .. (364.71,146.03) -- cycle ;
\draw   (442.71,146.03) .. controls (437.61,146.02) and (433.48,141.88) .. (433.49,136.77) .. controls (433.5,131.67) and (437.64,127.54) .. (442.75,127.55) .. controls (447.85,127.56) and (451.98,131.7) .. (451.97,136.81) .. controls (451.96,141.91) and (447.81,146.04) .. (442.71,146.03) -- cycle ;
\draw   (286.71,145.03) .. controls (281.61,145.02) and (277.48,140.88) .. (277.49,135.77) .. controls (277.5,130.67) and (281.64,126.54) .. (286.75,126.55) .. controls (291.85,126.56) and (295.98,130.7) .. (295.97,135.81) .. controls (295.96,140.91) and (291.81,145.04) .. (286.71,145.03) -- cycle ;
\draw   (203.71,146.03) .. controls (198.61,146.02) and (194.48,141.88) .. (194.49,136.77) .. controls (194.5,131.67) and (198.64,127.54) .. (203.75,127.55) .. controls (208.85,127.56) and (212.98,131.7) .. (212.97,136.81) .. controls (212.96,141.91) and (208.81,146.04) .. (203.71,146.03) -- cycle ;
\draw    (236.5,85) -- (203.75,127.55) ;
\draw    (396.5,83) -- (364.75,127.55) ;
\draw    (252.5,84) -- (280.75,126.55) ;
\draw    (412.5,82) -- (442.75,127.55) ;
\draw    (331.97,21.81) -- (403.75,66.55) ;
\draw    (314.5,22) -- (243.75,67.55) ;
\draw    (152.5,187) -- (495.5,188) ;
\draw  [dash pattern={on 0.84pt off 2.51pt}]  (203.71,146.03) -- (204.5,189) ;
\draw  [dash pattern={on 0.84pt off 2.51pt}]  (286.71,145.03) -- (287.5,188) ;
\draw  [dash pattern={on 0.84pt off 2.51pt}]  (364.71,146.03) -- (365.5,189) ;
\draw  [dash pattern={on 0.84pt off 2.51pt}]  (442.71,146.03) -- (443.5,189) ;

\draw (163,126.4) node [anchor=north west][inner sep=0.75pt]    {$t_{100}$};
\draw (250,127.4) node [anchor=north west][inner sep=0.75pt]    {$t_{101}$};
\draw (326,129.4) node [anchor=north west][inner sep=0.75pt]    {$t_{110}$};
\draw (406,127.4) node [anchor=north west][inner sep=0.75pt]    {$t_{111}$};
\draw (206,64.4) node [anchor=north west][inner sep=0.75pt]    {$t_{10}$};
\draw (369,64.4) node [anchor=north west][inner sep=0.75pt]    {$t_{11}$};
\draw (287,8.4) node [anchor=north west][inner sep=0.75pt]    {$t_{1}$};
\draw (147,193.4) node [anchor=north west][inner sep=0.75pt]  [font=\footnotesize]  {$r_{l}$};
\draw (188,193.4) node [anchor=north west][inner sep=0.75pt]  [font=\footnotesize]  {$r_{l} +\theta $};
\draw (269,194.4) node [anchor=north west][inner sep=0.75pt]  [font=\footnotesize]  {$r_{l} +2\theta $};
\draw (348,194.4) node [anchor=north west][inner sep=0.75pt]  [font=\footnotesize]  {$r_{l} +3\theta $};
\draw (438,193.4) node [anchor=north west][inner sep=0.75pt]  [font=\footnotesize]  {$r_{u}$};

\end{tikzpicture}
\caption{Tree representation of the CDF algorithm, where the counts $t_{100}, \ldots, t_{111}$ at the leaves of the tree represent the counts for histogram bins with width $\theta$, their parent nodes represent a histogram with bin width $2\theta$, and so on.}
\label{fig:cdf-tree}
\end{centering}
\end{figure}
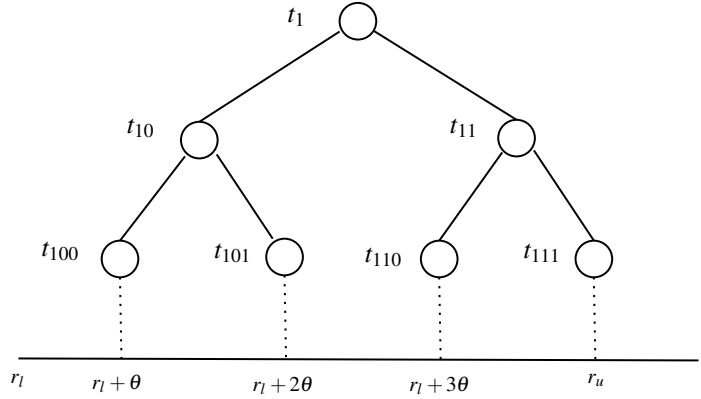

 This concept may be further improved through post-processing by noting that the different noisy counts at different levels of granularity ought to sum to the same values. \citet{Honaker:2015} proposes an optimal method to leverage this information, which we use here. 
Following Honaker's notation, label each node in the tree in binary, so the root node is $1$, the left child of the root is node $10$ and its right child is node $11$, and so on as in Figure \ref{fig:cdf-tree}. Let the count at node $i$ be $t_i$, and let the leaf nodes of the tree in Figure \ref{fig:cdf-tree} represent a histogram with granularity $\theta,$ so that $t_{100}$ is the number of data-points that lie between the left of the histogram's range, $\lowerrange$, to $\lowerrange + \theta$, $t_{101}$ is the number of data-points between $\lowerrange + \theta$ and $\lowerrange + 2\theta$, and so on. Let the histogram at any level have granularity twice that of its children. Note that if the counts were perfectly accurate, it then follows that the parent node's value should be equal to the sum of its children nodes' value, e.g. $t_{10} = t_{100} + t_{101}$ in Figure \ref{fig:cdf-tree}. Similarly, a child nodes' value in a perfectly accurate tree should be equal to its parent node, minus the adjacent child node's count, e.g. $t_{100} = t_{10} - t_{101}$. Let $i\Lambda 1$ be the neighboring child of node i (e.g. $11 \Phi 1 = t_{10}$), and let $i \Phi 1$ be the parent of node i.

Honaker leverages these relationships between the counts at each of the nodes to generate an optimal tree in a recursive process. First, the counts at each node using the children node are incorporated into a weighted estimate of each of the counts, where the weight of the child counts at node $i$ is denoted $w_i^-$ and the optimized count ``from below" is denoted $t_i^-$. These are then recursively combined with the counts ``from above" (i.e. the count using the parent and the adjacent child) to generate $t_i^+$, which weights the count from above with $w_i^+$. Finally, the two are combined to get a fully efficient estimate of each of the nodes, with weight $w$ and efficient count $t_i^*$.\footnote{See Honaker equations 10, 11, and 13 for the full statement of the values of these weights and counts.} In our setting, we only consider trees with an equal amount of noise on each of the nodes in the tree. This results in several simplifications of the equations in Honaker 2015. 

\begin{lemma}
\label{lem:cdf-weights}
 If each noisy count in the tree has noise with variance $s$ added to it, then the weight vectors $w^-$ and $w^+$ only need to be calculated once per level of the tree, the weight at node $i$ for the summation from below is 
$$ w_i^- = \frac{2w_{2i}^-}{2w_{2i}^- + 1},$$
the weight at node $i$ for the summation from above is
$$ w_i^+ = \frac{1}{1 + (w_{i \Phi 1}^+ + w_{i \Lambda 1}^-)^{-1}}, $$
and the optimal weight $w$ is equal to $w^+$.
\end{lemma}

Note that if each node has noise with variance $s$ added to it, then $\forall i,$ the variance at node $i, \sigma_i=s$, so $\sigma^-(t_{2i}^-) = \sigma^-(t_{2i+1}^-)$. Then, the weight at node $i$, $w_i$ may be recursively defined as 
\begin{align*}
w_i^- &= \frac{s^{-2}}{s^{-2} + (1/2)(\sigma_{2i}^-)^{-2}}\\
	&= \frac{s^{-2}}{s^{-2} + (1/2)\left(s\sqrt{w_{2i}^-}\right)^{-2}}\\
	&= \frac{s^{-2}}{s^{-2}(1+(1/2)(w_{2i}^-)^{-1})}\\
	&= \frac{2w_{2i}^-}{2w_{2i}^- + 1},
\end{align*}
where the first line comes from Honaker's definition in his Equation 10. Similarly, from Honaker's Equation 11 it follows that when the noise at each node has the same variance,
\begin{align*}
w_i^+ &= \frac{(\sigma_i)^{-2}}{(\sigma_i)^{-2} + \left[(\sigma^+_{i\Phi 1})^2 + (\sigma^-_{i \Lambda 1})^2\right]^{-1}} \\
    &= \frac{1}{1 + (w_{i \Phi 1}^+ + w_{i \Lambda 1}^-)^{-1}}, \\
\end{align*}
which is equivalent to then the expression for the optimal weights in this setting.

Once there is a tree of fully optimized counts, one can read off the CDF at an arbitrary point by traversing the tree in a root-to-leaf path, summing as few of the values together as possible to get the desired estimate, as shown in Algorithm \ref{alg:cdf-tree}.

\begin{singlespace}
\begin{algorithm}[H]
\caption{ 
\texttt{TreeToCDF}}
\label{alg:cdf-tree}
\KwInput{$\ninput$, 
$\treeinput$, $\discreterangeinput$, $\graninput$, $\depthinput$}

noisy-cdf $= []$

\For{$x \in \discreterange$}{
$\min \leftarrow \lowerrange, \max \leftarrow r_u$

count $\leftarrow 0, i \leftarrow 0$

\For{$0 \leq j < \depth$} {
    
    mid $\leftarrow (\min + \max)/2 $
    
    \uIf {$x = \max$ or $j = m$}{
        $k \leftarrow 2^j + i$
        
       count $\leftarrow$ count +  $t_k$ 
       
       \textbf{break}
        }
    \uElseIf{$x \le$ mid}{
        $\max \leftarrow$ mid
        
        $i \leftarrow 2i$
    }
    \Else{
        $\min \leftarrow$ mid
        
        $i \leftarrow 2i+1$
        
        $k \leftarrow  2^{j+1} + i -1$
        
       count $\leftarrow \dptree_{k}$ 
           }
  }
  Add $(x, \text{count}/n)$ to noisy-cdf
}
\Return noisy-cdf
\end{algorithm}
\end{singlespace}

\subsection{Confidence Intervals for Quantiles Estimated from Tree-Based CDF}
In order to generate a confidence interval for the desired quantiles, we would need to understand what the uncertainty of each of the counts in our estimated CDF was. Since each of these counts is a combination of all of the different counts in the tree, weighted in a way that is recursively defined, this is not trivial to do in a closed-form manner. We can compute the effect of each node on any other node by generating a tree with every node valued at 0 except the node we are interested in, then running the recursive weighting algorithm on that tree, as shown in Alg.~\ref{alg:pp-node-effect}.

\begin{singlespace}
\begin{algorithm}[H]
\label{alg:pp-node-effect}
\caption{
\texttt{ComputeNodeEffect}}
\KwInput{$\sigma^2 \in \nonnegreals$, $i^* \in \naturals$, $\depthinput$}
Construct a binary tree $\dptree \in \dptreespace{\depth}$, where for $0 \leq i < 2^\depth$,
\begin{align*}
    \begin{cases}
    T_i = 1 &\text{if}~~~~ i = i^* \\
    T_i = 0 &\text{o.w.}
    \end{cases}
\end{align*}

Let $\dptree' \in \dptreespace{\depth}$ be the output of the CDF post-processing algorithm from~\cite{Honaker:2015} on differentially private tree $\dptree$, where each noisy count in $T$ has variance $\sigma^2$. 

\Return $\dptree'$
\end{algorithm}
\end{singlespace}

We now have a method to understand how much each node effects any other node. If we run this on every single node of the tree, we can then combine them to generate a tree for every node on the tree that describes how much its optimized count is affected by any other node.\footnote{Since we add identically distributed noise added to each node's count, there are symmetries in the node effects that can be leveraged to make this process substantially more efficient in practice.} When summing the counts to generate the CDF, we can then keep track of the total weight of each node in the final count, and from here generate the variance of the count.

\begin{singlespace}
\begin{algorithm}[H]
\caption{
\texttt{GetVariances} 
}
\label{alg:getVariances}
\KwInput{$\treeinput, \depthinput, \rangeinput, \rhoinput$}
    
    Let $\sigma^2 \leftarrow 2 \depth / \rho$
    
    Create binary tree $E \in \dptreespace{\depth}$ with all nodes are set to 0.
    
    $ \mathbf{T}' \leftarrow \{\texttt{ComputeNodeEffect}(\sigma^2, i, \depth)\}_{0 \leq i < 2^{\depth}}$ 
    
    $\mathbf{v} \leftarrow \emptyset$ 
    
    \For{$0 \leq i < 2^{\depth}$}{

    $\min \leftarrow \lowerrange, \max \leftarrow r_u$ 
    
    \For{$0 \leq j < 2^{\depth}$}{
    $\text{mid} \leftarrow (\min + \max) / 2$ 
    
    \If{$i$ is a leftmost node of the tree}{
        \text{break}
        }
        
    \If{$T_j$ corresponds to a bin with upper endpoint $\max$ or $T_j$ is a leaf node}{
        \For{$0 \leq k < 2^{\depth}$}{
            $E_k \leftarrow E_k+\mathbf{T}'_{j,k}$
        }
    }
    \uElseIf{$T_j <  \text{mid}$}{
        $\max = \text{mid}$
        
        $j \leftarrow 2j$ 
    }
    \Else{
        $\min \leftarrow \text{mid}$ 
        
        $j \leftarrow 2j+1$\\ 
        \For{$0 \leq k < 2^{\depth}$}{
            $E_k \leftarrow E_k+\mathbf{T}'_{2j-1,k}$
        }
    }
    }
    $v\leftarrow0$ 
    
    \For{$0 \leq j <2^{\depth}$}{
    
        $v \leftarrow v + E_i^2 \cdot \sigma^2$ 
    }
    $\mathbf{v}_i \leftarrow v$
    }
\Return $\mathbf{v}$
\end{algorithm}

Now that we have a way to estimate the variance of the count at each of the nodes, we need to generate the actual confidence interval. One way to do this is with the same $\texttt{PostProcessUnion}$ algorithm used in the binary search approach (Alg. \ref{alg:consBS}); the validity of this interval follows the proof of the algorithm's validity for binary search. However, we can do slightly better here, since the choice of query points is just based on the granularity of the tree's histograms rather than dependent on previous queries. This improved method is described in Alg.~\ref{alg:fancy} and the entire confidence interval generation process is summarized in Alg.~\ref{alg:cdf}.

\begin{algorithm}[H]
\label{alg:fancy}
\caption{
\texttt{PostProcess}}
\KwInput{$\ninput$, $\discreterangeinput$, noisy CDF counts $\{x, \tilde{C}(x), \sigma_x\}_{x \in \discreterange}$}

\For{$x \in \discreterange$}{
 
 $a_{x}^u
    = \min \{ a \mid \int_q \PDFbin(qn) \cdot \Pr( q + \calN(0, \sigma_x^2) >  a) \leq \alpha/2 \}$ \tcp{can approximate using binary search}
 
  $a^l_x \leftarrow 1 -a^u_x$
}

$\ell = \max \{x \in \discreterange \mid \forall x' \leq x \in \discreterange, \tilde{C}(x') < a_{x'}^l \}$

$u = \min\{i \in [N] \mid \forall x' \geq x \in \discreterange, \tilde{C}(x') > a_{x'}^u \}$

\Return $[\ell,u]$
\end{algorithm}

\begin{algorithm}[H]
\KwData{$\datainput$}
\KwPrivacyparams{$\rhoinput$}
\KwHyperparams{$\alphainput, \Naive \in \{0,1\}$, $\rangeinput, \graninput, \betasplitinput$}

        $n = |d|$
        
        $m = \lceil \log((\dataub - \datalb)/\theta) \rceil$
        
        $\dptree$ = \texttt{DPTree}($d, \rho, (\range, \depth)$)

        $\dptree^*$ = \texttt{OptimizedTree}($\dptree, \rho$)
        \tcp{Optimized post-processing algorithm from~\cite{Honaker:2015} with upper and lower weights as in Lemma~\ref{lem:cdf-weights}}
        
        $ \discreterange = \{\lowerrange, \lowerrange + \theta, \lowerrange + 2\theta, \ldots, \lowerrange + 2^m \theta \}$        
                        
        $(x_i, \tilde{C}(x_i))_{x_i \in \discreterange} = \texttt{TreeToCDF}(n, T^*, \discreterange, \theta, \depth)$
        
        $\{\texttt{var}_i\}_{x_i \in \discreterange} = \texttt{GetVariances}(\dptree^*, \rho)$
        
        \If{\Naive}{
        
        $ \nonprivconf = \betasplit \alpha$
        
        $\privconf = \frac{\alpha - \nonprivconf}{1 - \nonprivconf/2}$
        
        $[l, u] = \texttt{PostProcessUnion}(\discreterange, \{ (x, n \tilde{C}(x), n \sigma_x^2 \}_{x \in \discreterange}, \PNPL{\nonprivconf}, \PNPU{\nonprivconf}, \privconf)$  \tcp{Algorithm~\ref{alg:consBS}}
        
        }
        \Else{
        $[l, u] = \texttt{PostProcess}(n, \discreterange, (x, \tilde{C}(x), \sigma_x)_{x \in \discreterange})$
        }
        
        \Return $[l, u]$
        
\label{alg:cdf}
\caption{$\CDF(\Naive)$: $\rho$-CDP algorithm}
\end{algorithm}
\end{singlespace}

We now need to show that our algorithm is differentially private and that the intervals that Algorithm \ref{alg:cdf} returns are valid confidence intervals. 
\begin{lemma}
\label{lem:cdf-privacy}
Mechanism $\CDF(\Naive)$ (Algorithm~\ref{alg:cdf}) is $\rho$-CDP.
\end{lemma}
\begin{proof}
Note that the only step in $\CDF(\Naive)$ that touches the dataset $d$ is the call to $\texttt{DPTree}$, which creates a tree of $m$ differentially private histograms. Each histogram is  $\rho/\depth$-CDP, and by composition (\citep[Proposition 1.6]{BunS16}), $\texttt{DPTree}$ is a $\rho$-CDP algorithm. The rest of the computations in $\CDF(\Naive)$ apply post-processing to the output of $\texttt{DPTree}$, so they do not affect the privacy guarantee.
\end{proof}

\begin{lemma}
\label{lem:cdf-validity}
For any dataset $d \overset{iid}{\sim} P$, where $P \in \gooddist$, and any hyperparameters $\theta, \range=[\lowerrange, r_u], \betasplitinput$, failure rate $ \alpha$ and privacy parameter $\rho$, let $\CDF(d, \rho, (\alpha, \Naive=0, \range, \theta, \betasplit))$ return an interval $[\privNPL{\alpha}(d), \privNPU{\alpha}(d)]$. If $\median(P)\in \range$, then with probability at least $1 - \alpha$,
\begin{align*}
    \median(P) \in [\privNPL{\alpha}(d), \privNPU{\alpha}(d)]
\end{align*}
where the probability is taken over the randomness of both the dataset $d$ and the mechansim $\CDF$.
\end{lemma}
\begin{proof}
Let us consider the upper endpoint of the interval. First, given a set of DP measurements $\tilde{C}(x) = \hat{C}(x) + \calN(0, \sigma_x^2)$, for all $x \in \discreterange$, recall that we define $a_{x}^u$ as follows.
\begin{align*}
    a_{x}^u
    = \min \{ a \mid \int_q \Pr_{d}(\hat{C}(\median(P)) = q) \cdot \Pr_{N \sim \calN(0, \sigma_x^2)}( q + N >  a) \leq \alpha/2 \}
\end{align*}
Then, recall that the post-processing algorithm (\texttt{PostProcess}) outputs $\privNPU{\alpha}(d) = \min\{x \in \discreterange \mid \forall x' \geq x, \tilde{C}(x') > a_{j}^u \}$. Let $x^{*} = \max \{ x \in \discreterange \mid x < \median(P) \}$, with corresponding $\sigma_{x^*}$ and $a_{x^*}^u$.
Then, using the subscript $A$ to denote randomness of the DP mechanism, we have that
\begin{align*}
    \Pr_{A, d}( \privNPU{\alpha}(d) < \median(P) )
    &= \Pr_{A, d}( \min\{x \in \discreterange \mid \forall x' \geq x ~~ \tilde{C}(x') > a_x^u \} < \median(P)) \\
    &\leq \Pr_{A, d}( \tilde{C}(x^{*}) > a_{x^*}^u ) \\
    &= \int_q \Pr_d(\hat{C}(\median(P)) = q) \cdot \Pr_A(\tilde{C}(x^{*}) > a_{x^*}^u \mid \hat{C}(\median(P)) = q) \\
    &\leq \int_q \Pr_{d}(\hat{C}(\median(P)) = q) \cdot \Pr_{N \sim \calN(0, \sigma_{x^*}^2)}(q + N > a_{x^*}^u )) \\
    &\leq \alpha/2
\end{align*}
where the last line follows by definition of $a_{x^*}^u$.
A similar argument holds for $\privNPL{\alpha}(d)$, so we are done.
\end{proof}

\section{Details: Range-robust estimator based on CDF estimator, \BSCDF}\label{online supplement:BSCDF}
\label{sec:bs-cdf-details}
Recall that $\noisyBS$ (Algorithm~\ref{alg:bs}) is useful for finding the dataset when it lies within a large range $\rangelarge$, while $\CDF$ (Algorithm~\ref{alg:cdf}) offers highly optimized estimates of the CDF within a small range $\rangesmall$. The combination $\BSCDF$ leverages the strengths of both of these algorithms: it uses $\noisyBS$ to narrow down the search space from $\rangelarge = [r_l, r_u]$ to $\rangesmall = [r_l', r_u']$, clips the data to within $\rangesmall$, and runs $\CDF$ with the remaining privacy budget within this smaller range to obtain a confidence interval for the population median. The pseudocode for $\BSCDF$ is given in Algorithm~\ref{alg:bs-cdf}.

The privacy budget $\rho$ and coverage failure probability $\alpha$ both need to be partitioned between the two stages of the algorithm. We expect the optimal split to be distribution dependent. In particular, it likely depends on how large a region the data occupies within the range $\range$. We found experimentally, for the parameter regimes we studied, using $\rho/4$ for the first step, and $3\rho/4$ for the second step ($\gamma = 1/4$) seemed to be a good choice. Similarly, we ensure that the region found in the first step contains the median with probability $1-\alpha/4$, and the second step finds a $1-3\alpha/4$-confidence interval within that region.

\begin{algorithm}[h!]
  \KwData{$\datainput$}
  \KwPrivacyparams{$\rhoinput$}
  \KwHyperparams{$\alphainput$, $\rangeinput, \graninput, r, r_1, \gamma \in (0,1)$}
  
  $n = |d|$
  
  $\rhoforbs = \gamma \cdot \rho$
  
  $\alpha_{\noisyBS} = \alphasplit \cdot \alpha$
  
    $\rhoforcdf = (1-\rhosplit) \cdot \rho$
    
  $\alphaforcdf = (1- \alphasplit) \cdot \alpha$
  
  $[r'_l, r'_u] = \noisyBS(d, \rhoforbs, (\alpha_{\noisyBS}, \range, \theta, \betasplit, 0.25, 0.75))$ 
  
  \Return $\CDF(d, \rhoforcdf, (\alphaforcdf, [r_l', r_u'], \theta, \betasplit))$
  
  \caption{$\BSCDF$:  
  $\rho$-CDP Algorithm}  \label{alg:bs-cdf}
\end{algorithm}

\begin{lemma}
Mechanism $\BSCDF$ (Algorithm~\ref{alg:bs-cdf}) is $\rho$-CDP.
\end{lemma}
\begin{proof}
$\BSCDF$ is a composition of two algorithms -- $\noisyBS$ which by Lemma~\ref{lem:noisy-bs-privacy} is $\gamma \rho$-CDP, and $\CDF$ which by Lemma~\ref{lem:cdf-privacy} is $(1-\gamma) \rho-CDP$. By Lemma~\ref{composition}, this means $\BSCDF$ satisfies $\rho$-CDP. 
\end{proof}

The coverage analysis of $\BSCDF$ follows immediately from Lemma~\ref{lem:cdf-validity} and Lemma~\ref{lem:noisy-bs-validity}, and a union bound.

\begin{lemma}
Given any dataset $d \overset{i.i.d}{\sim} P^n$, where $P \in \gooddist$, any hyperparameters $\graninput, \rangeinput, \betasplit, \alphasplit, \rhosplit$, failure rate $\alphainput$ and privacy parameter $\rhoinput$, if $\median(P)\in \range$ then $\BSCDF(d, \rho, (\alpha, \theta, \range, \betasplit, \alphasplit, \rhosplit))$ is a valid $1-\alpha$-confidence interval for $\median(P)$.
\end{lemma}

\begin{proof}
By Lemma~\ref{lem:noisy-bs-validity}, if $\median(P)\in \range$ then $\Pr(\median(P)\in[r_l', r_u'])\ge 1-\alpha_{\noisyBS}$. Then, by Lemma~\ref{lem:cdf-validity}, $\Pr(\median(P)\in\BSCDF(d)\;|\; \median(P)\in[r_l', r_u'])\ge 1-\alpha_{\CDF}$. Therefore, \[\Pr(\median(P)\in\BSCDF(d))\ge 1-\alpha_{\noisyBS}-\alpha_{\CDF}=1-\alpha.\]
where the probability is over the randomness of both the dataset $d$ and the mechanism $BSCDF$.
\end{proof}
\section{Details: Other Algorithms Explored}
\label{online supplement: otheralgs}

In this section we give a brief overview of additional CDP confidence intervals and CDP median estimators that we explored. These algorithms were not included in the main body of this paper since they are outperformed by other algorithms in every parameter regime we studied. The additional CDP confidence interval algorithms were:
\begin{itemize}
    \item \texttt{CDF+BS CI} computes a CDP estimate to the empirical CDF in the same way as $\CDF$. However, instead of using the post-processing algorithm described in Algorithm~\ref{alg:fancy}, it performs binary search using the noisy CDF measurements.
    \item \texttt{BinSearch} is the same as $\noisyBS$ except it uses the same privacy budget at every iteration. We expect this algorithm to perform strictly worse than $\noisyBS$ which uses its budget more carefully.
\end{itemize}

Figure~\ref{fig:otherCIS} shows the performance of \texttt{CDF+BS CI} and \texttt{BinSearch}, as well as the naive estimators \texttt{ExpMechUnion} and \texttt{CDFPostProcessUnion} and the four CDP estimators we presented in the main body. We can see that for all values of $\rho$, at least one of the four main CDP estimators outperforms each of the other algorithms.

\begin{figure}[H]
    \centering
    \includegraphics[scale=0.35]{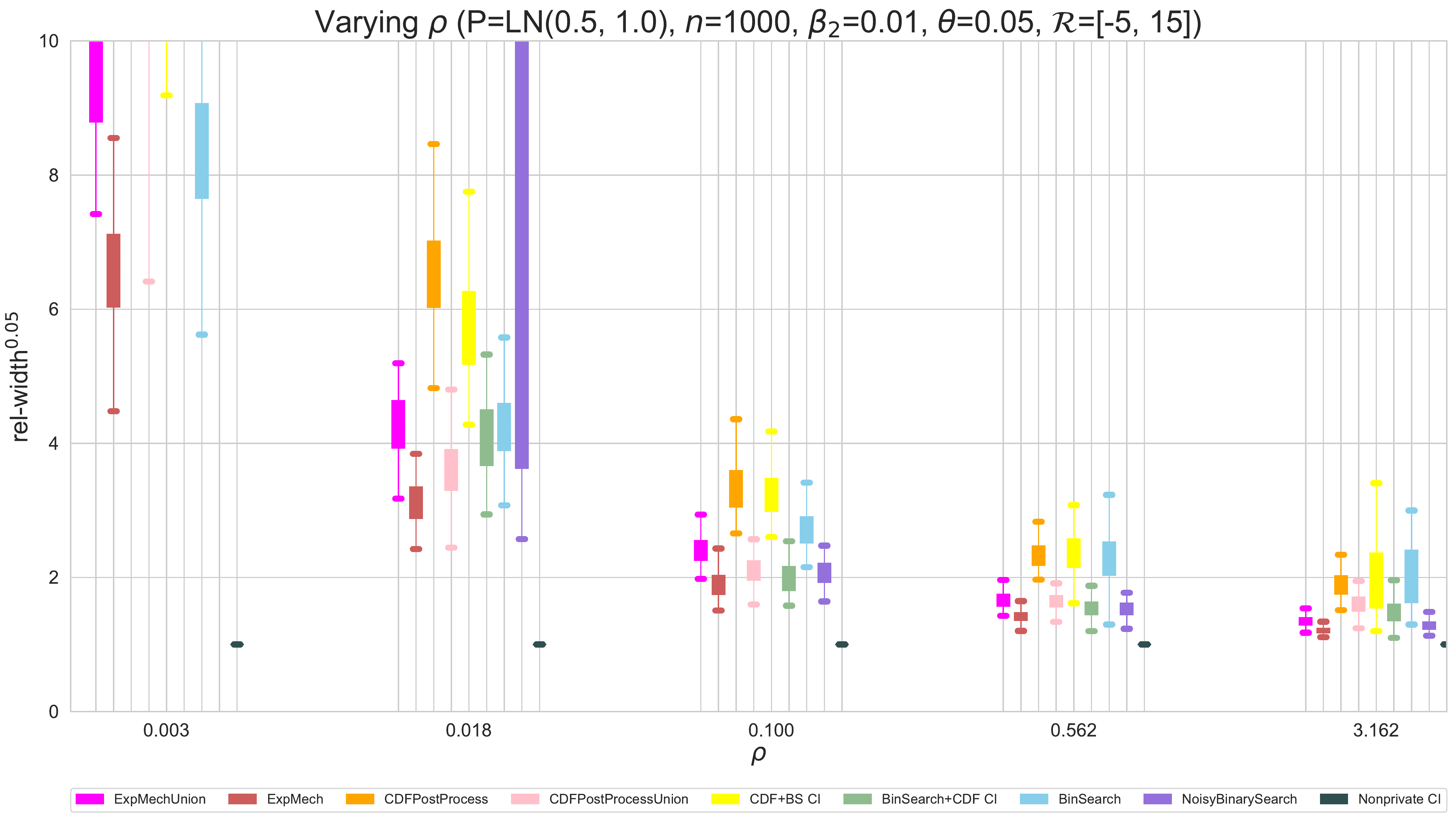}
    \caption{Performance of various CDP confidence intervals for the median as we vary $\rho$. Performance is measured in terms of the relative width with $\alpha=0.05$. Box plots are computed using 100 random datasets of 1000 data points drawn i.i.d. from Lognormal($\ln(1.5), 1$). Each CDP algorithm is run 5 times on each dataset.}
    \label{fig:otherCIS}
\end{figure}

We also explored several CDP median estimators. For point estimators that directly correspond to analogues of our CDP confidence intervals, we use the same name to denote both. However, note that these point estimators for the median are different to the mid-point of the confidence interval estimators that we used in the main body. The algorithms presented in Figure~\ref{fig:pointestimates} are all directly estimating the median, and do not additionally release a confidence interval for the median. 

\begin{itemize}
    \item \texttt{ExpMech} is the point estimator version of our confidence interval algorithm $\EM$. It uses the exponential mechanism with target quantile $n/2$ to estimate the median.
    \item \texttt{SmoothSens} releases the median using the smooth sensitivity framework \cite{NRS07, BunS19}. This algorithm Gaussian noise to the empirical median where the standard deviation of the noise is data dependent, and carefully calibrated to ensure differential privacy. 
    \item \texttt{BinSearch} is the point estimator version of \texttt{BinSearch} described above. It uses binary search with target quantile $n/2$.
    \item \texttt{NoisyStartBinarySearch} was a preliminary version of altering the privacy budget through the iterations of the algorithm, with little budget initially then increasing through the search process. 
    \item \noisyBS is the point estimator version of our confidence interval algorithm $\noisyBS$. It uses noisy binary search to search for the quantile $n/2$.
    \item \texttt{FancyBinarySearch} is similar to \texttt{BinSearch}. However, instead of halving the range at each iterate, it makes more conservative steps when it is not confident whether the median is to the left or right. 
    \item \texttt{CDFPostProcess} is the point estimator version of our confidence interval algorithm $\CDF$. It computes a CDP estimate the CDF in the same way, then computes the median based on the CDP CDF. An important note is that since both this algorithm and $\CDF$ are post-processing on the CDP CDF estimate, they can be performed at the same time without additional privacy budget.
    \item \texttt{GradDescent} uses CDP gradient descent to solve the optimisation problem $\arg\min\sum_{i=1}^n|m-d_i|$. We use the private stochastic gradient descent technique proposed by \cite{Bassily:2014}.
\end{itemize}

Figure~\ref{fig:pointestimates} shows the performance of each of our point estimators on log-normal data. We can see that \texttt{SmoothSens}, the only unbiased estimator, has among the highest variability in all regimes, and particularly poor performance for small $\rho$. It has comparable performance to the other algorithms for large $\rho$, but extending this point estimator to a confidence interval algorithm remains an open problem. All the variants of binary search perform similarly as point estimators for the median. Even as a point estimator, $\EM$ slightly outperforms the other algorithms, except \texttt{GradDescent}. Extending \texttt{GradDescent} to a confidence interval remains an open problem. 

\begin{figure}[H]
    \centering
    \includegraphics[scale=0.4
    ]{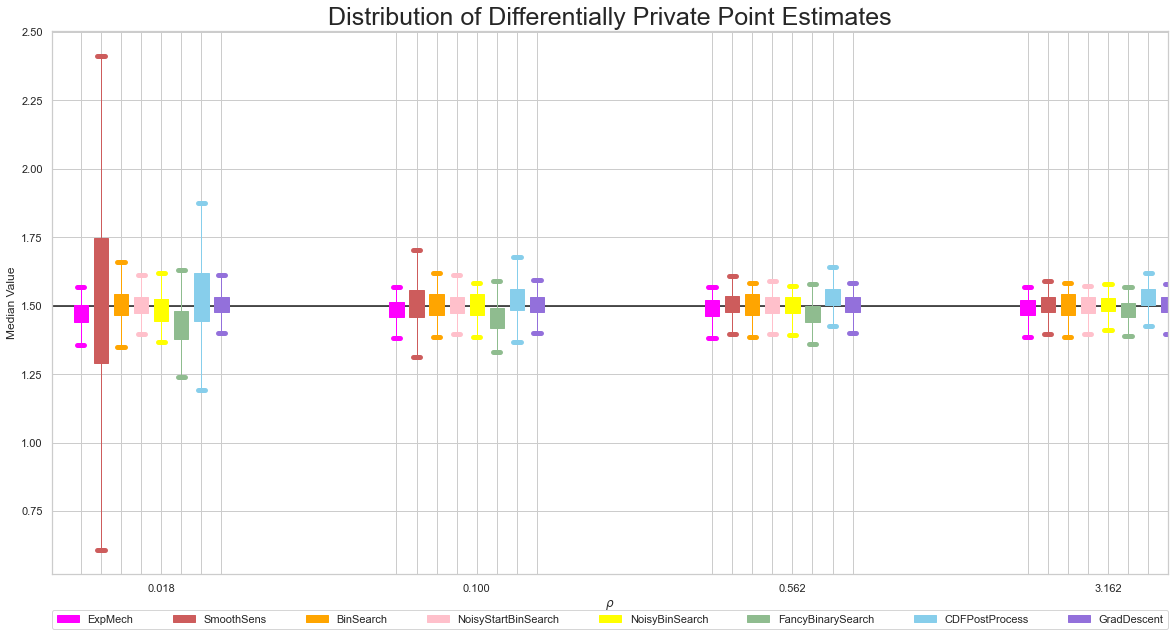}
    \caption{Performance of various CDP point estimators for the median as we vary $\rho$. Box plots are computed using 100 random datasets of 1000 data points drawn i.i.d. from Lognormal($\ln(1.5), 1$). Each CDP algorithm is run 5 times on each dataset. }
    \label{fig:pointestimates}
\end{figure}
\section{Other Regimes}
\label{online supplement: otherregimes}
The relative ordering of algorithms can depend on the scale of the data ($\sigma_d$) relative to the range.  In the figures below, we display the relative widths of the algorithms on data sampled from a Lognormal$(\ln(1.5, \sigma_d^2)$ distribution, where $\sigma_d = 5.0$, as we vary the size of the dataset $n$ and the privacy loss parameter $\rho$. Note that although we have drastically increased the scale of the data, the range is left the same as in Figure~\ref{fig:rel-width-boxplots}: $\range = [-5, 15]$. 
From these plots, we can see that when $n$, $\rho$, and $\sigma_d$ are large, $\CDF$ performs slightly better than $\EM$. In Figure~\ref{fig:rel-width-boxplots}, we saw that when $\sigma_d$ is small, $\EM$ remains the best performing algorithm in both the large $n$ and large $\rho$ regimes. Hence we conjecture that $\sigma_d$ needs to be large, and either $\rho$ or $n$ need to be large for $\CDF$ begins outperforming $\EM$. This conjecture is supported by Figure~\ref{fig:rel-width-large-sigma}, where we see $\CDF$ only beginning to outperform $\EM$ when either $n$ and $\rho$ are large.

\begin{figure}
     \begin{subfigure}[b]{0.5\textwidth}
         \centering
         \includegraphics[width=\textwidth]{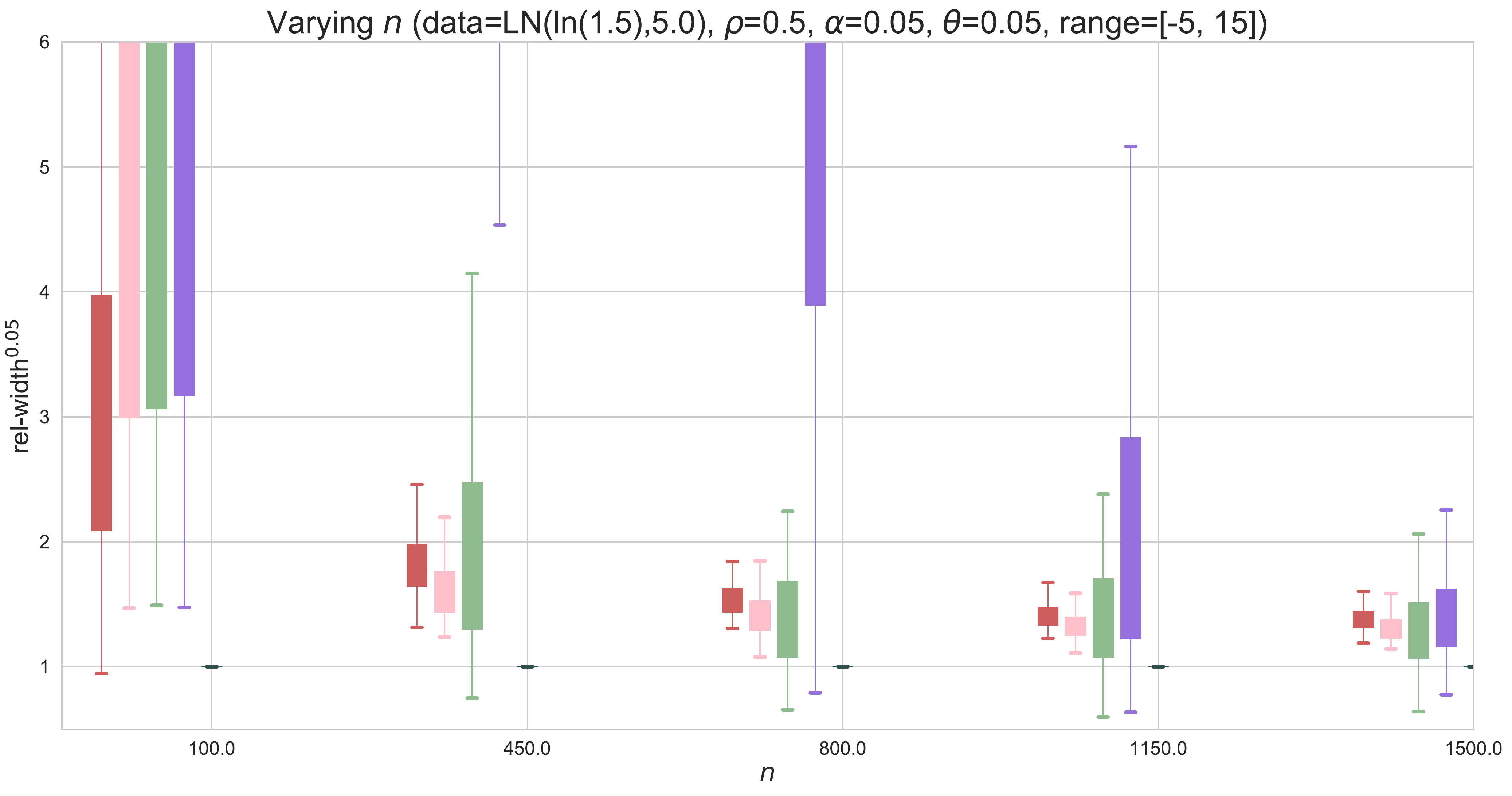}
         \caption{Varying $n$}
         \label{fig:rel-width-large-sigma-n}
     \end{subfigure}
     \begin{subfigure}[b]{0.5\textwidth}
         \centering
         \includegraphics[width=\textwidth]{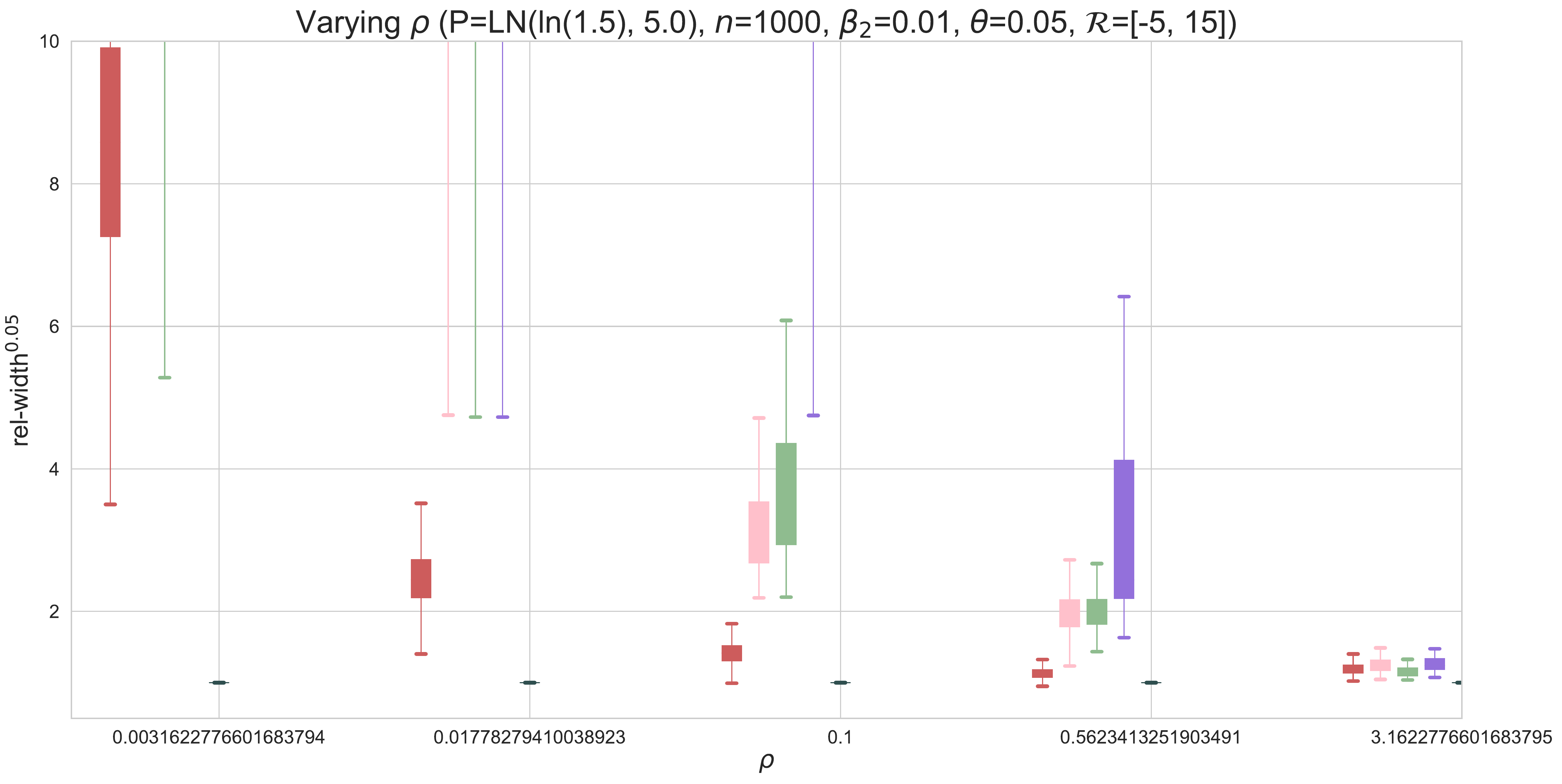}
         \caption{Varying $\rho$}
         \label{fig:rel-width-large-sigma-rho}
     \end{subfigure}
     \begin{subfigure}[b]{\textwidth}
         \centering
         \includegraphics[width=\textwidth]{main-labels.pdf}
     \end{subfigure}
     \caption{Relative width of DP confidence intervals on well-spread data $(\sigma_d = 5.0)$ as we vary (a) $n$ and (b) $\rho$.
     }\label{fig:rel-width-large-sigma}
\end{figure}

\end{document}